\documentclass[journal]{IEEEtran}
\usepackage{cite}
\usepackage{epsfig}
\usepackage{amssymb}
\usepackage{amsmath}
\usepackage{mathtools}
\usepackage{breqn}
\usepackage{enumitem}
\usepackage[caption=false,font=footnotesize]{subfig}
\usepackage{fixltx2e}
\usepackage{bm}
\usepackage{url}
\usepackage{scalerel}
\newtheorem{theorem}{Theorem}

\newtheorem{example}{Example}
\newtheorem{lemma}{Lemma}
\newtheorem{definition}{Definition}
\newtheorem{corollary}{Corollary}
\newcommand{\cB}[1]{{\left\{#1\right\}}}

\newcommand{\cA}{\ensuremath{\mathcal{A}}}

\newcommand{\cX}{{\mathcal X}}
\newcommand{\cY}{{\mathcal Y}}
\newcommand{\bx}{{\mathbf x}}

\newcommand{\bbP}{\mathbb{P}}
\newcommand{\bbPrep}{\bbP_{\rm rep}}
\newcommand{\bbE}{\mathbb{E}}
\newcommand{\norm}[1]{\left\lVert#1\right\rVert}
\newcommand{\sfe}{{\mathsf e}}
\newcommand{\sff}{{\mathsf f}}
\newcommand{\sfg}{{\mathsf g}}
\newcommand{\es}{{0}}
\newcommand{\hW}{\hat{W}}
\newcommand{\ws}{w^*}
\newcommand{\tw}{\overline{w}}
\newcommand{\tcW}{\overline{\mathcal W}}
\newcommand{\E}[1]{\mathbb E\left[#1\right]}
\newcommand{\Prob}[1]{\mathbb P\left[#1\right]}

\newcommand{\bigo}[1]{O\left(#1\right)}

\newcommand{\Var}[1]{\mathrm{Var}\left[#1\right]}
\newcommand{\thmref}[1]{Theorem~\ref{#1}}

\newcommand{\secref}[1]{Section~\ref{#1}}
\newcommand{\lemref}[1]{Lemma~\ref{#1}}

\newcommand{\corref}[1]{Corollary~\ref{#1}}
\newcommand{\appref}[1]{Appendix~\ref{#1}}

\newcommand{\figref}[1]{Fig.~\ref{#1}}
\usepackage{hyperref}
\hypersetup{backref, colorlinks=true, linkcolor=blue, citecolor=blue, urlcolor = blue}
\usepackage{ifthen}
\newif\ifshowtodo
\showtodofalse

\newcommand{\VersionLength}{long}
\providecommand{\ver}{\ifthenelse{\equal{\VersionLength}{long}}}

\allowdisplaybreaks
\title{Random Access Channel Coding\\ in the Finite Blocklength Regime}
\author{Recep Can Yavas, Victoria Kostina, and Michelle Effros
    \thanks{R.~C.~Yavas, V.~Kostina, and M.~Effros are with the Department of Electrical Engineering, California Institute of Technology, Pasadena, CA~91125, USA. \mbox{E-mails}: \mbox{{\em \{ryavas, vkostina, effros\}@caltech.edu}}. This work was supported in part by the National Science Foundation (NSF) under grant CCF-1817241. A part of this work was presented at ISIT'18\cite{effros2018Random}.}
}
\IEEEoverridecommandlockouts
\begin{document}
\maketitle 
\begin{abstract}
Consider a random access communication scenario over a channel whose operation is defined for any number of possible transmitters. As in the model recently introduced  by Polyanskiy for the Multiple Access Channel (MAC) with a fixed, known number of transmitters, the channel is assumed to be invariant to permutations on its inputs, and all active transmitters employ identical encoders. Unlike the Polyanskiy model, in the proposed scenario, neither the transmitters nor the receiver knows which transmitters are active. We refer to this agnostic communication setup as the Random Access Channel (RAC). Scheduled feedback of a finite number of bits is used to synchronize the transmitters. The decoder is tasked with determining from the channel output the number of active transmitters, $k$, and their messages but not which transmitter sent which message. The decoding procedure occurs at a time $n_t$ depending on the decoder's estimate, $t$, of the number of active transmitters, $k$, thereby achieving a rate that varies with the number of active transmitters. Single-bit feedback at each time $n_i, i \leq t$, enables all transmitters to determine the end of one coding epoch and the start of the next. The central result of this work demonstrates the achievability on a RAC of performance that is first-order optimal for the MAC in operation during each coding epoch. While prior multiple access schemes for a fixed number of transmitters require $2^k - 1$ simultaneous threshold rules, the proposed scheme uses a single threshold rule and achieves the same dispersion.
\end{abstract}

\section{Introduction}
\'e \`a Access points like WiFi hot spots and cellular base stations are, for wireless devices, the gateway to the
network. Unfortunately, access points
are also the network's most critical bottleneck. 
As more kinds of devices become network-reliant, both the
number of communicating devices and the diversity of their communication needs grow.
Little is known  
about how to code under high variation in the number and variety of communicators. 

Multiple-transmitter single-receiver channels are well understood in information theory when the number and identities of transmitters are fixed and known. Unfortunately, even in this known-transmitter regime, information-theoretic solutions are too complex to implement. As a result, orthogonalization methods, such as TDMA, FDMA, and orthogonal CDMA, are used instead. Orthogonalization strategies simplify coding by allocating resources (e.g., time slots) among the transmitters, but applying such methods to discrete memoryless MACs can at best attain a sum-rate equal to the single-transmitter capacity of the channel, which is often significantly smaller than the maximal
multi-transmitter sum-rate.

Most random access protocols currently in use rely on collision avoidance, which cannot surpass the single-transmitter capacity of the channel and may be significantly worse since the unknown transmitter set makes it difficult to schedule or coordinate among transmitters. Collision avoidance is achieved through variations of the legacy (slotted) ALOHA and carrier sense multiple access (CSMA) algorithms.  
ALOHA, which uses random transmission times and back-off schedules, achieves only about $37\%$ of the single-transmitter capacity of the channel~\cite{roberts1975aloha}.  In CSMA, each transmitter tries to avoid collisions by verifying the absence of other traffic before starting a transmission over the shared channel; when collisions do occur, 
all transmissions are aborted, and a jamming signal is sent to ensure that all transmitters are aware of the collision. The procedure starts again at a random time, which again introduces inefficiencies. The state of the art in random access coding is ``treating interference as noise,'' which is part of newer CDMA-based standards. While this strategy can deal with random access better than ALOHA, it is still far inferior to the theoretical limits.

Even from a purely theoretical perspective, a satisfactory solution to random access  remains to be found. The MAC model in which a fixed number, $k$, out of the total available $K$ transmitters are active was studied by
D'yachkov and Rykov~\cite{dyachkov1981mac} and Mathys \cite{mathys1990class} for zero-error coding on a noiseless adder MAC, and by Bassalygo and Pinsker \cite{bassalygo1983restricted} for an asynchronous model in which the information is considered erased if more than one transmitter is active at a time. See \cite{polyanskiy2017perspective} for a more detailed history.
Two-layer MAC decoders, with outer layer codes that work to remove channel noise and inner layer codes that work to resolve conflicts, are proposed in \cite{ordentlich2017mac,ebrahimi2017coded}. Like the codes in \cite{dyachkov1981mac, mathys1990class, bassalygo1983restricted}, 
the codes in \cite{ordentlich2017mac, polyanskiy2017perspective} are designed for a predetermined number of transmitters, $k$; 
it is not clear how robust they are to randomness in the transmitters' arrivals and departures. In \cite{minero2012randomaccess}, Minero \textit{et al.} study a random access model in which the receiver knows the transmitter activity pattern, and the transmitters opportunistically send data at the highest possible rate. The receiver recovers only a portion of the messages sent, depending on the current level of activity in the channel. 

\subsection{Our Contributions and Related Works}
This paper poses the question of whether it is possible, in a scenario where no one knows how many transmitters are active, for the receiver to almost always recover the messages sent by all active transmitters. Surprisingly, we find that not only is reliable decoding possible in this regime, but, for the class of permutation-invariant channels considered in \cite{polyanskiy2017perspective}, our proposed RAC code performs as well in its capacity and dispersion terms as the best-known code for a MAC with the transmitter activity known a priori \cite{ huang2012finite, jazi2012simpler, tan2014dispersions,  scarlett2015constantcompMAC}. 
Since the capacity region of a MAC varies with the number of transmitters, 
it is tempting to believe that the transmitters of a random access system must somehow vary their codebook size 
in order to match their transmission rate to the capacity region of the MAC in operation. 
Instead, we here allow the decoder to vary its decoding time 
depending on the observed channel output---thereby adjusting the rate at which each transmitter communicates 
by changing not the size but the blocklength of each transmitter's codebook.

Codes that can accommodate variable decoding times are called \emph{rateless codes}. Rateless codes originate with the work of Burnashev \cite{burnashev1976data}, who computed the error exponent of variable-length coding over a known point-to-point channel. Polyanskiy \textit{et al.} \cite{polyanskiy2011feedback} provide a dispersion-style analysis of the same scenario.  A practical implementation of rateless codes for an erasure channel with an unknown erasure probability appears in \cite{luby2002lt}. An analysis of rateless coding over an unknown binary symmetric channel appears in \cite{tchamkerten2002feedback} and is extended to an arbitrary discrete memoryless channel in \cite{draper2004efficient,shulman2003unknown} using a decoder that tracks Goppa's empirical mutual information and decodes once that quantity passes a threshold. In \cite{blits2012universal}, Jeffrey's prior is used to weight unknown channels. A rateless code for noiseless random access communication is described in \cite{stefanovic2013aloha}; each user transmits replicas of its message in multiple time slots, possibly colliding with the messages of other transmitters. At the end of each time slot, the decoder attempts to apply successive interference cancellation starting with the messages received without collision and subsequently removing the associated interference from the time slots in which replicas are transmitted. The decoder then decides whether to terminate an epoch or to ask the transmitters to send more replicas. 



Unlike the codes described in \cite{burnashev1976data,polyanskiy2011feedback,luby2002lt,tchamkerten2002feedback,draper2004efficient,shulman2003unknown,blits2012universal, stefanovic2013aloha}, which allow truly arbitrary decoding times,  in this paper we allow decoding only at a predetermined list of possible times $n_0, n_1, n_2, \ldots$. This strategy both eases practical implementation and reduces feedback. In particular, the schemes in\cite{burnashev1976data,polyanskiy2011feedback,luby2002lt,tchamkerten2002feedback,draper2004efficient,shulman2003unknown,blits2012universal, stefanovic2013aloha} transmit a single-bit acknowledgment message from the decoder to the encoder(s) once the decoder completes its decoding process. Because the decoding time is random, this so-called ``single-bit'' feedback forces the transmitter(s) to listen to the channel constantly, at every time step trying to discern whether or not a transmission was received. This either requires full-duplex devices or doubles the effective blocklength and can be quite expensive.  Thus while the receiver technically sends only ``one bit'' of feedback, the transmitters receive one bit of feedback (with the alphabet $\{\mbox{``transmission'',``no transmission''}\}$) in every time step, giving a feedback rate of 1 bit per channel use rather than a total of 1 bit.  In our framework, feedback bits are sent only at times $n_0 <n_1 < \cdots < n_t$, where each $n_i$ is the pre-determined decoding time used if the receiver believes that $i$ transmitters are active. Thus the transmitters must listen only at a sparse collection of time steps. The total number of feedback bits equals one plus the receiver's estimate of the number of transmitters, giving a feedback rate approaching $0$ bits per channel use as the blocklength grows. 

In the central portion of this paper, we view the random access channel as a collection of all possible MACs that might arise as a result of the transmitter activity pattern. Barring the intricacies of multiuser decoding, the model that views an unknown channel as a collection of possible channels without assigning an a priori probability to each is known as the \emph{compound channel} model \cite{blackwell1959capacity}. In the context of single-transmitter compound channels, it is known that if the decoding time is fixed, the transmission rate cannot exceed the capacity of the weakest channel from the collection \cite{blackwell1959capacity}, though the dispersion may be better (smaller) \cite{polyanskiy2013dispersion}. With feedback and a variable decoding time, one can do much better \cite{tchamkerten2002feedback,draper2004efficient,shulman2003unknown,blits2012universal}.

In \cite{polyanskiy2017perspective}, Polyanskiy argues for removing the transmitter identification task from the physical layer encoding and decoding procedures of a MAC. As he points out, such a scenario was previously discussed by Berger~\cite{berger1981poisson} in the context of conflict resolution. Polyanskiy further suggests studying MACs whose conditional channel output distributions are insensitive to input permutations.  For such channels, if all transmitters use the same codebook, then the receiver can at best hope to recover the messages sent without recovering who transmitted which message (the transmitter identity). In some networks the transmitter identification task can be insignificant. For example, in some sensor networks, we might be interested in the collected measurements but indifferent to the identities of the collecting sensors. In scenarios where transmitter identity is required, it can be included in the payload.

In \secref{sec:raccode}, we propose a code for a random access communication channel model built from a family of permutation-invariant MACs. Our code employs identical encoders at all transmitters and identity-blind decoding at the receiver. Although not critical for the feasibility of our approach, these assumptions lead to a number of pleasing simplifications of both our scheme and its analysis. For example, using identical encoders at all transmitters simplifies design and implementation. Further, the collection of MACs comprising our compound RAC model can be parameterized by the number of active transmitters rather than by the full transmitter activity pattern. \ver{}{If the maximum number of transmitters is finite, the analysis of identity-blind decoding differs little from traditional analyses that use independent realizations of a random codebook at each transmitter. We elaborate on this small difference in \secref{sec:transiden}, where we discuss an extension of our strategy to allow for transmitter identity decoding.} 

We provide a second-order analysis of the rate universally achieved by our multiuser scheme over all transmitter activity patterns,  taking into account the possibility that the decoder may misdetect the current activity pattern and decode for the wrong channel.  Leveraging our observation that for a symmetric MAC, the fair rate point is not a corner point of the capacity region, we are able to show that a single-threshold decoding rule attains the fair rate point. This differs significantly from traditional MAC analyses, which use $2^{k}-1$ simultaneous threshold rules.  In the context of a MAC with a known number of transmitters, second-order analyses of multiple-threshold decoding rules are given in \cite{huang2012finite,jazi2012simpler,tan2014dispersions,scarlett2015constantcompMAC} (finite alphabet MAC) and in \cite{molavianjazi2015second} (Gaussian MAC). A non-asymptotic analysis of variable-length coding with ``single-bit" feedback over a (known) Gaussian MAC appears in \cite{truong2017gaussian}. 

Other relevant recent works on the MAC include the following. 
To account for massive numbers of transmitters, in \cite{chen2014many,chen2017capacity}, Chen and Guo introduce a notion of capacity for the multiple access scenario in which the maximal number of transmitters grows with the blocklength and an unknown subset of transmitters is active at a given time. They show that time sharing, which achieves the conventional MAC capacity, is inadequate to achieve capacity in that regime. In \cite{sarwate2009feedback}, Sarwate and Gastpar show that rate-0 feedback, such as the feedback in our approach, does not increase the capacity of the discrete memoryless MAC. In compound MACs, limited feedback can increase capacity. For example, one strategy uses a simple training phase to estimate the channel state and employs feedback to send the state estimate to the transmitter. Such schemes cannot increase the capacity beyond the rate achievable when the state is known to the encoders and the decoder \cite{sarwate2009feedback}.

The sparse recovery problem is identical to a special case of the RAC problem in which each transmitter sends only its ``signature" to the receiver. Here, the decoder's only task is to determine who is active. Active transmitters in this variant of the RAC problem may correspond to defective items or positive test outcomes in the sparse recovery problem, and successful decoding is identified with successfully detecting the set of defective or confirmed-positive elements. A group testing problem in which an unknown subset of $k$ defective items out of $K$ items total is observed through an OR MAC, is studied in \cite{malyutov1978, malyutov1980planning, atia2012, scarlettPhase2016, scarlett2017limits}; this problem is a special case of the sparse recovery problem. In these works, the decoder reaches a conclusion about tested items at a fixed blocklength $n$. Atia and Saligrama \cite{atia2012} consider a noiseless group testing scenario in which the number of transmitted elements, $k$, does not grow with the total number of elements, $K$, showing that the smallest possible number of measurements needed to detect the defective items is $O(k \log {\frac{K}{k}})$. In in \cite{scarlettPhase2016}, Scarlett and Cevher extend this result to the scenario where $k$ scales as $O(K^{\theta})$ for $\theta \in (0, 1)$. In \cite{scarlett2017limits}, Scarlett and Cevher derive the information-theoretic limits of the exact and partial support recovery problems for general probabilistic models, where exact recovery refers to detecting all $k$ defective items, and partial recovery refers to detecting at least $s$ out of $k$ defective items. While we consider a nonvanishing average error probability and operate in the central limit theorem regime, \cite{malyutov1978, malyutov1980planning,  atia2012, scarlettPhase2016, scarlett2017limits} assume vanishing average error probability and operate in the large deviations regime. The main difference between the decoder designs in  \cite{malyutov1978, malyutov1980planning, atia2012, scarlettPhase2016, scarlett2017limits} and our decoder design is that  \cite{malyutov1978, malyutov1980planning, atia2012, scarlettPhase2016, scarlett2017limits} use $2^k - 1$ simultaneous information density threshold tests at a single blocklength $n$, while our decoder uses a single information density threshold test at multiple decoding times, allowing successful detection with a computationally less complex decoder even when the number of active transmitters to be detected is unknown.

\subsection{Paper Organization}
Our system model and proposed communication strategy are laid out in \secref{sec:setup}. The main result, showing that for a nontrivial class of channels our proposed RAC code performs as well in terms of capacity and dispersion as the best-known code for a MAC with the transmitter activity known a priori, is presented in \secref{sec:main}. The proofs are presented in \secref{sec:raccode}. \secref{sec:furtherDis} includes discussions of the effect of using maximum likelihood decoding, the choice of an input distribution in the random code design, the difficulties in proving a converse, an extension of our strategy that enables transmitter identity decoding, and performance bounds under the per-user error probability criterion. Interestingly, the problem of decoding for $k \geq 1$ unknown transmitters is substantially different from the problem of detecting whether there are any active transmitters at all. In \secref{sec:analysis_0test}, we employ universal hypothesis testing to solve the latter problem. \secref{sec:summary} concludes the paper with a discussion of our results and their implications.

\section{Problem Setup}\label{sec:setup}
For any positive integers $i,j$, let $[i]=\{1,\dots,i\}$ and
$[i:j]=\{i,\ldots,j\}$, where $[i:j]=\emptyset$ when $i>j$.
We denote an $n$ dimensional vector by $x^n = (x_1, \dots, x_n)$. When the dimension of a vector $x^n$ is clear from the context, we denote $x^n$ by $\bx$. All-zero and all-one vectors are denoted by $\boldsymbol{0}$ and $\boldsymbol{1}$, respectively. For a collection of length-$n$ vectors $x_1^n, \dots, x_K^n$ and any subset $\mathcal{C} \subseteq [K]$, we denote the corresponding sub-collection of vectors by $x_{\mathcal{C}}^n = (x_c^n \colon c \in \mathcal{C})$. For collection of vectors $x_{\mathcal{C}}^n$ and index $i \in [n]$, $x_{\mathcal{C}, i}$ denotes the collection of scalars obtained by taking $i$-th coordinate from each vector in $x_{\mathcal{C}}^n$. For any vectors $x_{\mathcal{C}}$ and $y_{\mathcal{C}}$, we write $x_{\mathcal{C}} \leq y_{\mathcal{C}}$ if $x_c \leq y_c$ for all $c \in \mathcal{C}$, $x_{\mathcal{C}} \stackrel{\pi}{=} y_{\mathcal{C}}$ if there exists a permutation $\pi$ of $y_{\mathcal{C}}$ such that $x_{\mathcal{C}} = \pi(y_{\mathcal{C}})$, and 
$x_{\mathcal{C}} \stackrel{\pi}{\neq} y_{\mathcal{C}} $ if $x_{\mathcal{C}} \neq \pi(y_{\mathcal{C}})$ for all permutations $\pi$ of $y_{\mathcal{C}}$. For any set $\mathcal{A}$ and integer $k \leq |\mathcal{A}|$, $\binom{\mathcal{A}}{k} = \{ \mathcal{B} \colon \mathcal{B} \subseteq \mathcal{A}, |\mathcal{B}| = k \}$. For a random variable $X$, we write $X \sim P_X$ to specify that $X$ is distributed according to distribution $P_X$. We use $Q(\cdot)$ to denote the Gaussian complementary cumulative distribution function, giving $Q(x) \triangleq  \frac{1}{\sqrt{2\pi}} \int_{x}^{\infty} \exp \left \lbrace \frac{-u^2}{2} \right\rbrace du$.
We employ the standard $o(\cdot)$ and $O(\cdot)$ notations, giving $f(n) = o(g(n))$ if $\lim_{n \to \infty} \left \lvert \frac{f(n)}{g(n)} \right\rvert = 0$ and $f(n) = O(g(n))$ if $\limsup_{n \to \infty} \left \lvert \frac{f(n)}{g(n)} \right \rvert < \infty$.

A \emph{stationary, memoryless, symmetric, random access channel} (henceforth called simply a RAC)
is a memoryless channel with one receiver and an unknown number of
transmitters. It is described by a family of stationary, memoryless MACs
\begin{align}
\left\{\left(\cX^k, P_{Y_k | X_{[k]}}(y_k|x_{[k]}),
\cY_k\right)\right\}_{k=0}^K \label{eq:collection},
\end{align}
each indexed by a number of transmitters, $k$; the maximal number of transmitters is $K \leq \infty$. When $k = 0$, no transmitters are active; we discuss this case separately below.
For $k \geq 1$, the $k$-transmitter MAC has input alphabet $\cX^k$, output alphabet $\cY_k$,
and conditional distribution $P_{Y_k | X_{[k]}}$.
When $k$ transmitters are active, the RAC output is $Y=Y_k$. The input and output alphabets $\mathcal{X}$ and $\mathcal{Y}_k$ can be abstract.

\subsection{Assumptions on the Channel}
We assume that the impact of a channel input
on the channel output is independent of the transmitter from which it comes; therefore,
each channel in \eqref{eq:collection}
is assumed to be \emph{permutation-invariant} \cite{polyanskiy2017perspective}, giving
\begin{align}
P_{Y_k | X_{[k]}}(y_k | x_{[k]}) = P_{Y_k | X_{[k]}}(y_k | \hat{x}_{[k]})
\label{eq:permutationinvariance}
\end{align}
for all $\hat{x}_{[k]} \stackrel{\pi}{=} x_{[k]}$ and $y_k \in \mathcal{Y}_k$, $k \in [K]$.
We further assume that for any $s<k$,
an $s$-transmitter MAC is physically identical to a $k$-transmitter MAC
operated with $s$ active and $k-s$ silent transmitters. At each time step of the communication period, each silent transmitter transmits a silence symbol, here denoted by $0\in\cX$. This \emph{reducibility} constraint gives 
\begin{align}
 P_{Y_s | X_{[s]}}(y|x_{[s]}) = P_{Y_k
| X_{[k]}}(y|x_{[s]},0^{k-s}) \quad 
\label{eq:reducible}
\end{align}
for all $s < k$, $x_{[s]}\in\cX_{[s]}$, and $y\in \cY_s$.
An immediate consequence of reducibility is that $\cY_s
\subseteq \cY_k$ for any $s < k $. Another consequence is that when there are no active transmitters, the MAC $\left(\cX^0, P_{Y_0 | X_{[0]}}(y|x_{[0]}), \cY_0 \right)$ satisfies $\mathcal{X}^0 = \{0\}$ and $P_{Y_0 | X_{[0]}}(y|x_{[0]}) = P_{Y_k | X_{[k]}}(y|0^k) $ for all $k$.

\subsection{RAC Communication Strategy}
We here propose a new RAC communication strategy.
In the proposed strategy, communication occurs in epochs, with each epoch beginning in the time step following the previous epoch's end. Each epoch ends when the receiver's scheduled broadcast to all transmitters indicates a decoding event, signaling that the prior transmission can stop and a new transmission can begin. 
At this point, each transmitter
decides whether to be active or silent in the new epoch;
the decision is binding for the length of the epoch,
meaning that a transmitter must either actively transmit for all time steps
in the epoch
or remain silent for the same period.
Thus, while
the total number of transmitters, $K$, is potentially unlimited and can change arbitrarily from one epoch to the next, the number of
active transmitters, $k$, remains constant throughout each epoch.

Each active transmitter uses the epoch to describe a message $W$ from the alphabet $[M]$. When the active transmitters are $[k]$, the messages are $W_{[k]}\in[M]^k$, where the messages $W_1, \ldots, W_k$ of different transmitters are independent and uniformly distributed. The proposed strategy fixes the potential decoding times $n_0 < n_1 < \cdots < n_K$.\footnote{We focus the exposition on the scenario where the decoding blocklengths are ordered both for simplicity and because a particular choice of ordered blocklengths emerges as optimal within our architecture (see \eqref{eq:diff} in \secref{sec:achproof}, below).}
The receiver chooses to end the epoch (without decoding) at time $n_0$ if it believes at time $n_0$ that no transmitters are active and chooses to end the epoch and decode at time $n_t$ if it believes at time $n_t$ that the number of active transmitters is $t$. The transmitters are informed of the decoder's decision through a single-bit feedback $Z_s$ at each time $n_s$ with $s\in\{0, 1, \dots, t\}$; here $Z_s=0$ for all $s<t$ and $Z_t=1$, with ``1'' signaling the end of one epoch and the beginning of the next. Since the blocklength for a given epoch is the decoding time chosen by the receiver, the result is a \emph{rateless code}. As we show in Section \ref{sec:raccode} below, with an
appropriately designed decoding rule, correct decoding is performed at time $n_{k}$ with high probability.

It is important to stress that in this domain
each transmitter knows nothing about the set of active
transmitters $\mathcal A \subset \mathbb N$ beyond its own membership
and what it learns from the receiver's feedback,
and the receiver knows nothing about $\mathcal A$ beyond what it learns
from the channel output $Y$; we call this \emph{agnostic} random access.
In addition, since designing a different encoder
for each transmitter is expensive from the perspective
of both code design and code operation, as in \cite{polyanskiy2017perspective},
we assume through most of this paper that every transmitter employs the same encoder; we call this \emph{identical encoding}. Under the assumptions of permutation-invariance
and identical encoding,
what the transmitters and receiver can learn about $\cA$ is quite limited.
Together, these properties imply that the decoder can at best distinguish
{\em which messages were transmitted} rather than {\em by whom they were sent}. 
In practice, transmitter identity could be included
in the header of each $\log M$-bit message
or at some other layer of the stack;
transmitter identity is not, however, handled by the RAC code.
Instead, since the channel output statistics depend on
the dimension of the channel input
but not the identity of the active transmitters,
the receiver's task is to decode the messages transmitted
but not the identities of their senders.
We therefore assume without loss of generality
that $|\mathcal A|=k$ implies $\mathcal A=[k]$.
Thus the family of $k$-transmitter MACs in \eqref{eq:permutationinvariance}  fully describes the behavior of a RAC.\footnote{\secref{sec:transiden} treats a variant of our RAC communication strategy that enables decoding of transmitter identity. Mathematically, the variants are quite similar. \label{footnote:treats}}

\subsection{Code Definition}
The following
definition formalizes our code.
\begin{definition}\label{def:codecompound}
For any number of messages $M$, ordered blocklengths $n_0 < n_1 < \cdots < n_K$, and error probabilities $\epsilon_0, \dots, \epsilon_K$, an $(M, \lbrace (n_k, \epsilon_k) \rbrace_{k = 0}^K )$ RAC code comprises a (rateless) encoding function
\begin{align}
\mathsf f \colon \, \mathcal{U} \times [M] \to \mathcal X^{n_K} \label{eq:f}
\end{align}
and a collection of decoding functions
\begin{align}
\mathsf g_k \colon \, \mathcal{U} \times \mathcal Y_k^{n_k} \to [M]^k \cup \{\mathsf
e\}, \quad k = 0, 1, \ldots, K,
\label{eq:g}
\end{align}
where $\mathsf{e}$ denotes the erasure symbol, which is the decoder's output when it is not ready to decode. At the start of each epoch, a common randomness random variable $U \in \mathcal{U}$, with $U \sim P_U$, is generated independently of the transmitter activity and revealed to the transmitters and the receiver,
thereby  initializing the encoders and the decoder. If $k$ transmitters are active, then with probability at least $1 - \epsilon_k$, the $k$ messages are correctly decoded
at time $n_k$. That is,\footnote{Recall that $\stackrel{\pi}=$ and $\stackrel{\pi}\neq$ denote equality and inequality up to a permutation.}
\begin{align}
\frac{1}{M^k}&\sum_{w_{[k]} \in [M]^k} \mathbb{P}\Bigg[\cB{\mathsf
g_k(U, Y_k^{n_k}) \stackrel{\pi}{\neq} w_{[k]}}\bigcup \notag\\
& \left.\cB{\bigcup_{t = 0}^{k-1}\cB{\mathsf{g}_t(U, Y_k^{n_t}) \neq \sfe}} \middle | \right. W_{[k]} = w_{[k]} \Bigg] \leq
\epsilon_k, \label{eq:errorcriterion}
\end{align}
where $W_{[k]}$ are the independent and equiprobable messages of transmitters $[k]$, and the given probability is calculated using the conditional distribution $P_{Y_k^{n_k} | X_{[k]}^{n_k}} = P_{Y_k | X_{[k]}}^{n_k}$; here $X_i^{n_k} = \mathsf f(U, W_i)^{n_k}$, $i = 1, \ldots, k$. At time $n_s$, the decoder outputs the erasure symbol ``$\mathsf{e}$" if it decides that the number of active transmitters is not $s$.
If $k = 0$ transmitters are active, the unique message ``\es", denoted $[M]^0 \triangleq \{\es\}$ to simplify the notation, is decoded at time $n_0$ with probability at least $1-\epsilon_0$. That is,
\begin{align}
\Prob{\mathsf{g}_0(U, Y_0^{n_0}) \neq \es | W_{[0]} = \es} \leq
\epsilon_0. \label{eq:probzeroerror}
\end{align}
\end{definition}

In Definition~\ref{def:codecompound}, we index the family of possible codes by the elements of some set $\mathcal{U}$ and include $u \in \mathcal{U}$ as an argument for both the RAC encoder and the RAC decoder. We then represent encoding as the application of a code indexed by some random variable $U \in \mathcal{U}$ chosen independently for each new epoch. Deterministic codes are represented under this code definition by setting the distribution on $U$ as $\Prob{U = u_0} = 1$ for some $u_0 \in \mathcal{U}$. 
In practice, we can implement a RAC code with random code choice $U$ using common randomness. Common randomness available to the transmitters and the receiver allows all nodes to choose the same random variable $U$ to specify a new codebook in each epoch. Operationally, this common randomness can be implemented by allowing the receiver to choose random instance $U$ at the start of each epoch and to broadcast that value to the transmitters just after the feedback bit that ends the previous epoch. Alternatively, all communicators can use synchronized pseudo-random number generators. Broadcasting the value of $U$ increases the epoch-ending feedback from $1$ bit to $\lceil \log |\mathcal{U}|\rceil + 1$ bits; \thmref{thm:caratheodory} shows that $|\mathcal{U}| \leq K + 1$ suffices to achieve the optimal performance. 

In \secref{sec:raccode}, we employ a general random coding argument to show that a given error vector $(\epsilon_0, \dots, \epsilon_K)$ is achievable when averaged over the ensemble of codes. Unfortunately, this traditional approach does not show the existence of a deterministic RAC code (i.e., a code with $|\mathcal{U}| = 1$) that achieves the given error vector $(\epsilon_0, \dots, \epsilon_K)$. The challenge here is that our proof showing that the random code's expected error probability meets each of the $K + 1$ error constraints does not suffice to show that any of the codes in the ensemble meets all of our error constraints simultaneously. A similar issue arises in \cite{polyanskiy2011feedback, tchamkerten2006variable}. For example, in \cite{polyanskiy2011feedback}, a variable-length feedback code is designed with the aim of achieving average error probability no greater than $\epsilon$ and expected decoding time no greater than $\ell$. To design a single code satisfying both constraints, \cite{polyanskiy2011feedback} relies on common randomness. Similarly, \cite{tchamkerten2006variable} describes a variable-length feedback code designed to satisfy an error exponent criterion for every channel in a continuum of binary symmetric or Z channels. Their proof that a single, deterministic code can simultaneously satisfy this continuum of constraints exploits the ordering among the channels in the given family. While channel symmetry can sometimes be leveraged to show the existence of a deterministic code \cite[eq. (29)]{polyanskiy2011feedback}, the symmetries in a RAC are quite different from those in point-to-point channels. We leave the question of whether a single-code solution exists for the RAC 
to future work.

The code model introduced in Definition \ref{def:codecompound} employs identical encoding in addition to common randomness. Under identical encoding, each transmitter uses the same encoder, $\mathsf f$, to form a codeword of length $n_K$. That codeword is fed into the channel symbol by symbol. According to Definition~\ref{def:codecompound}, if $k$ transmitters are active, then with probability at least $1 - \epsilon_k$, the decoder recovers the transmitted messages correctly after observing the first $n_k$ channel outputs. As noted previously, the decoder $\mathsf g_k$ does not attempt to recover transmitter identity; successful decoding means that the list of messages in the decoder output coincides with the list of messages sent. The error event defined in Definition \ref{def:codecompound} differs from the one in \cite{polyanskiy2017perspective}. Our definition \eqref{eq:errorcriterion} requires that all transmitted messages are decoded correctly. In contrast, \cite{polyanskiy2017perspective} bounds a per-user probability of error (PUPE),
which measures the fraction of transmitted messages that are missing from the list of decoded messages. In \secref{sec:PUPE}, we discuss the error probability for our code under the PUPE criterion.

\subsection{Information Density Definitions}
The following definitions are useful for the discussion that follows.
When $k$ transmitters are active, the input distribution is $P_{X_{[k]}}$, and the marginal output distribution is $P_{Y_k}$.
The information density and conditional information density are defined\footnote{We here employ notation for discrete alphabets. In the general case, it can be replaced by the logarithm of the Radon-Nikodym derivative, giving 
$\imath_k(x_{\mathcal A};y_k) = \log \frac{dP_{Y_k| X_{\scaleto{\mathcal{A}}{2.5pt}} = x_{\scaleto{\mathcal{A}}{2.5pt}}}}{dP_{Y_k}}(y_k)$.} as
\begin{align}
\imath_k(x_{\mathcal A};y_k) & \triangleq  \log\frac{P_{Y_k|X_{\mathcal A}}(y_k|x_{\mathcal A})}{P_{Y_k}(y_k)} \\
\imath_k(x_{\mathcal A} ; y_k | x_{\mathcal B}) & \triangleq  \log \frac{P_{Y_k|X_{\mathcal A}, X_{\mathcal B}}(y_k| x_{\mathcal A},x_{\mathcal B})}
{P_{Y_k | X_B}(y_k| x_{\mathcal B})}
\end{align}
for any ${\mathcal A},{\mathcal B}\subseteq [k]$, $x_{\mathcal{A}} \in \mathcal{X}_{\mathcal{A}}, x_{\mathcal{B}} \in \mathcal{X}_{\mathcal{B}}$, and $y_k \in \mathcal{Y}_k$;
here $\imath_k(x_{\mathcal A} ; y_k | x_{\mathcal B}) \triangleq \imath_k(x_{\mathcal A} ; y_k)$ when $\mathcal B=\emptyset$ and $\imath_k(x_{\mathcal A} ; y_k|x_{\mathcal B}) \triangleq 0$ when $y_k \notin \mathcal{Y}_k$ or ${\mathcal A} = \emptyset$.
The corresponding mutual informations are
\begin{align}
I_k(X_{\mathcal A} ; Y_k) & \triangleq \mathbb{E}[\imath_k(X_{\mathcal A} ; Y_k)]  \\
I_k(X_{\mathcal A} ; Y_k | X_{\mathcal B}) & \triangleq
\mathbb{E}[\imath_k(X_{\mathcal A} ; Y_k | X_{\mathcal B})].
\end{align}
Throughout the paper, we also denote for brevity
\begin{align}
 I_k &\triangleq I_k(X_{[k]}; Y_k)\\
 V_k &\triangleq \Var{\imath_k(X_{[k]}; Y_k)}. \label{eq:dispersionVk}
\end{align}
The multi-letter information density and conditional information densities are defined as
\begin{align}
    \imath_k(x_{\mathcal A}^n;y_k^n) & \triangleq  \log\frac{P_{Y_k^n|X_{\mathcal A}^n}(y_k^n|x_{\mathcal{A}}^n)}{P_{Y_k^n}(y_{k}^n)} \label{eq:imathnt}\\
\imath_k(x_{\mathcal{A}}^n ; y_{k}^n | x_{\mathcal{B}}^n) & \triangleq \log \frac{P_{Y_k^n|X_{\mathcal{A}}^n, X_{\mathcal{B}}^n}(y_k^n| x_{\mathcal{A}}^n,x_{\mathcal{B}}^n)}
{P_{Y_k^n | X_{\mathcal{B}}^n}(y_k^n| x_{\mathcal{B}}^n)}. \label{eq:imathn}
\end{align}


\subsection{Assumptions on the Input Distribution} \label{sec:assumpchannel}
To ensure the existence of codes satisfying the error constraints in
Definition~\ref{def:codecompound}, we assume that there exists a $P_X$ such that when $X_1, X_2, \ldots, X_K$ are distributed i.i.d. $P_X$, then the conditions in  \eqref{eq:silence}--\eqref{eq:moment2a} below are satisfied. 

The \emph{friendliness} assumption states that for all $s \leq k \leq K$,
\begin{align}
I_k(X_{[s]}; Y_k | X_{[s+1:k]} =  0^{k-s})
\geq I_k(X_{[s]}; Y_k | X_{[s+1:k]}). \label{eq:silence}
\end{align}
Friendliness implies that by remaining silent, inactive transmitters enable communication by the active transmitters at rates at least as large as those achievable if the inactive transmitters had actively participated and their codewords were known to the receiver.

The \emph{interference} assumption states that for any $s$ and $t$,
$X_{[s]}$ and $X_{[s+1:t]} $ are conditionally dependent given $Y_k$, giving 
\begin{align}
P_{X_{[t]} | Y_k} \neq P_{X_{[s]}|Y_k} \, P_{X_{[s+1:t]}|Y_k} \quad \forall \, 1 \leq s < t \leq k, \, \forall k. \label{eq:interference}
\end{align}
Assumption \eqref{eq:interference} eliminates trivial RACs in which transmitters do not interfere.

In order for the decoder to be able to distinguish the time-$n_0$ output $Y_0^{n_0}$ that results when no transmitters are active from the time-$n_0$ output $Y_k^{n_0}$ that results when $k \geq 1$ transmitters are active, we assume that there exists a $\delta_0 > 0$ such that the output distributions satisfy
\begin{align}
\sup_{y \in \mathcal{Y}_K} \left \lvert F_k(y) - F_0(y) \right \rvert \geq \delta_0 \, \text{ for all } k \in [K], \label{eq:assump:Pyk}
\end{align}
where $F_k(y)$ denotes the cumulative distribution function (CDF) of $P_{Y_k}$ for $k \in \{0, \dots, K\}$.\footnote{Although the CDF is defined for real-valued random variables, i.e., $\mathcal{Y}_k \subseteq \mathcal{Y}_K \subseteq \mathbb{R}$ is required, it can be generalized to abstract alphabets by introducing a partial order $\leq$ on the set $\mathcal{Y}_K$. Then $F_k(y) \triangleq \Prob{Y_k \leq y}$.} The measure of discrepancy between distributions on the left-hand side of \eqref{eq:assump:Pyk} is known as the Kolmogorov-Smirnov distance. The assumption in \eqref{eq:assump:Pyk} is only needed to detect the scenario when no transmitters are active; the remainder of the code functions proceed unhampered when \eqref{eq:assump:Pyk} fails. When $K$ is finite, \eqref{eq:assump:Pyk} is equivalent to $P_{Y_0} \neq P_{Y_k}$ for all $k \in [K]$.

Finally, the \emph{moment} assumptions 
 \begin{align}
\mathrm{Var}\left[\imath_k(X_{[k]} ; Y_k) \right] &> 0 \label{eq:moment2}\\
 \hspace{-3mm} \bbE[|\imath_k(X_{[k]}; Y_k) - I_k(X_{[k]}; Y_k) |^3] &<\infty
\label{eq:moment3}
\end{align}
enable the second-order analysis presented in Theorem~\ref{thm:ach}, below.
In the case when $\imath_t(X_{[s]} ; Y_k) > -\infty$ almost surely, we also require 
\begin{align}
\mathrm{Var}\left[\imath_t(X_{[s]} ; Y_k) \right] &< \infty \quad \forall s \leq t \leq k  \label{eq:moment2a}. 
\end{align}
Moment assumptions like \eqref{eq:moment2}--\eqref{eq:moment2a} are common in the finite-blocklength literature, e.g.,  \cite{polyanskiy2010Channel, tan2014dispersions}. 

In the discussion that follows, we say that a channel satisfies our channel assumptions (\eqref{eq:permutationinvariance}, \eqref{eq:reducible}, \eqref{eq:silence}--\eqref{eq:moment2a}) if there exists an input distribution $P_X$ under which those conditions are satisfied. All discrete memoryless channels (DMCs) satisfy finite second- and third-moment assumptions
\eqref{eq:moment3}--\eqref{eq:moment2a} \cite[Lemma 46]{polyanskiy2010Channel}, as do Gaussian noise channels. Common channel models from the literature typically satisfy a non-zero second-moment assumption \eqref{eq:moment2} as well.
Example channels that meet our channel assumptions (\eqref{eq:permutationinvariance}, \eqref{eq:reducible}, and \eqref{eq:silence}--\eqref{eq:moment2a}) include the Gaussian RAC, 
\begin{align}
Y_k = \sum_{i = 1}^{k} X_i  + Z, \label{eq:Gaussian}
\end{align}
where each $X_i \in \mathbb R$ operates under power constraint $P$ and
$Z \sim \mathcal N(0, N)$ for some $N > 0$, 
and the adder-erasure RAC \cite{ebrahimi2017coded},
\begin{align}
Y_k =
\begin{cases}
\sum_{i = 1}^{k} X_i, & \text{w.p. } 1 - \delta\\
\mathsf e & \text{w.p. } \delta,
\end{cases}
\label{eq:addererasure}
\end{align}
where $X_i \in \{0, 1\}$ and $Y_k \in \left\{ 0, \ldots, k\right\} \cup \{\mathsf e\}$. In \cite{ebrahimi2017coded}, the adder-erasure RAC \eqref{eq:addererasure} is used to model a scenario where a digital encoder and decoder communicate over an analog channel using a modulator and demodulator. The modulator converts the bits into analog signals; the channel output equals the sum of the transmitted signals plus random noise; the demodulator quantizes that output, declaring an erasure, $\mathsf{e}$, if reliable quantization is not possible due to high noise. Thus, one can view the adder-erasure RAC as a discretization of the Gaussian RAC.

For the Gaussian RAC, $\imath_t(X_{[s]};Y_k)>-\infty$ almost surely, and \eqref{eq:moment2a} is satisfied. For the adder-erasure RAC, $\imath_t(X_{[s]} ; Y_k) = -\infty$ for some channel realizations and user activity patterns, and \eqref{eq:moment2a} is not required.

We conclude this section with a series of lemmas that describe the natural orderings possessed by RACs that satisfy our permutation-invariance, reducibility, friendliness, and interference constraints  (\eqref{eq:permutationinvariance}, \eqref{eq:reducible}, \eqref{eq:silence}, and  \eqref{eq:interference}). These properties are key to the feasibility of the approach proposed in our achievability argument in \secref{sec:main}. Proofs are relegated to \appref{app:proofs}.

The first lemma shows that the quality of the channel for each active transmitter deteriorates as the number of active transmitters grows (even though the sum capacity may increase). 
\begin{lemma} \label{lem:mutualinfofriendly}
Let $X_1, X_2, \ldots, X_k\sim \text{ i.i.d. } P_X$. Under permutation-invariance \eqref{eq:permutationinvariance},  reducibility \eqref{eq:reducible}, friendliness \eqref{eq:silence}, and interference \eqref{eq:interference},
\begin{align}
\frac{I_k}{k} < \frac{I_s}{s} \quad \mbox{ for } k > s \geq 1.
\end{align}
\end{lemma}

The second lemma shows that a similar relationship holds even when the number of transmitters is fixed.
\begin{lemma}\label{lem:mutualinfo}
Let $X_1, X_2, \ldots, X_k\sim \text{i.i.d. }P_X$. Under permutation-invariance \eqref{eq:permutationinvariance}, reducibility \eqref{eq:reducible} and interference \eqref{eq:interference}, 
\begin{align}
\frac 1 k I_k(X_{[k]}; Y_k) < \frac 1 s I_k(X_{[s]}; Y_k | X_{[s+1:k]}) \quad \mbox{ for } k > s \geq 1. \label{eq:mutualinfo}
\end{align}
\end{lemma}
\lemref{lem:mutualinfo} ensures that the equal-rate point of the $k$-MAC lies on the sum-rate boundary and away from all the corner points of the rate region achieved with $P_X$. In their work on the group testing problem \cite[Th.~3]{malyutov1980planning}, Malyutov and Mateev prove a non-strict version of \eqref{eq:mutualinfo} for permutation-invariant channels \eqref{eq:permutationinvariance}. They use this non-strict version of \eqref{eq:mutualinfo} to conclude that their achievability and converse results in \cite[Th.~1 and 2]{malyutov1980planning} coincide for permutation-invariant channels. Adding the reducibility \eqref{eq:reducible} and interference \eqref{eq:interference} assumptions to the permutation-invariance assumption \eqref{eq:permutationinvariance} enables us to prove the strict inequality in \lemref{lem:mutualinfo}, which in turn enables the use of a single threshold rule at the decoder, as discussed in \secref{sec:raccode}.

\lemref{lemma:expectation} compares the expected values of the information densities for different channels. 

\begin{lemma}\label{lemma:expectation}
Let $X_1, X_2, \ldots, X_k \sim \text{i.i.d. }P_X$. If a RAC is permutation-invariant \eqref{eq:permutationinvariance}, reducible \eqref{eq:reducible}, friendly \eqref{eq:silence}, and exhibits interference \eqref{eq:interference}, then for any $1 \leq s \leq t < k$,
\begin{align}
\mathbb{E}[\imath_t(X_{[s]}; Y_k)] \leq I_k(X_{[s]}; Y_k) < I_t(X_{[s]}; Y_t).
\end{align}
\end{lemma}

The orderings in \lemref{lem:mutualinfofriendly}--\ref{lemma:expectation} are used in bounding the performance of our agnostic random access code.

\section{Main Result}
\label{sec:main}
\subsection{An Asymptotic Achievability Result}
Our main result is the following bound on achievable rates for the RAC.
\begin{theorem}\label{thm:ach}(Achievability) 
For any RAC \[\left\{\left(\cX^k, P_{Y_k | X_{[k]}}(y_k|x_{[k]}), \cY_k\right)\right\}_{k=0}^K \] satisfying \eqref{eq:permutationinvariance} and \eqref{eq:reducible}, any $K < \infty$, and any fixed $P_X$ satisfying \eqref{eq:silence}--\eqref{eq:moment2a}, 
there exists an $(M, \{(n_k, \epsilon_k) \}_{k = 0}^K )$ code provided that
\begin{align}
k \log M \leq n_k I_k - \sqrt{n_k V_k} Q^{-1}(\epsilon_k) - \frac{1}{2}\log n_k + O(1)
\label{eq:mainresult}
\end{align}
for all $k \in [K]$, 
and
\begin{align}
n_0 \geq c_0 \log n_1 + o(\log n_1),  \label{eq:n0logn}
\end{align}
where $c_0$ is a known positive constant. The $O(1)$ term in \eqref{eq:mainresult} is constant with respect to $n_1$; it depends on the number of active transmitters, $k$, but not on the total number of transmitters, $K$.

\end{theorem}
The code in \thmref{thm:ach} assigns equal rates $R_{[k]}=(R,\ldots,R)$, $R=\frac{\log M}{n_k}$, to all active transmitters. 
The sum-rate $kR$ converges as $O\left(\frac1{\sqrt{n_k}}\right)$ to $I_k(X_{[k]}; Y_k)$
for some input distribution $P_{X_{[k]}}(x_{[k]})=\prod_{i=1}^kP_{X}(x_i)$ for all $k$. Note that
$P_X$ is independent of the number of active transmitters, $k$. 
If the RAC is discrete and memoryless and a single $P_X$ maximizes $I_k(X_{[k]}; Y_k)$ for every $k$, then
the achievable rate in \eqref{eq:mainresult} not only converges to the symmetrical rate point
on the capacity region of the MAC in operation but also achieves the best-known second-order term \cite{ huang2012finite, jazi2012simpler, tan2014dispersions,  scarlett2015constantcompMAC}\footnote{Note that we are comparing the RAC achievable rate with rate-0 feedback 
to the MAC capacity without feedback.  
Wagner \textit{et al.} \cite{wagner2020ANew} show that if a discrete, memoryless, point-to-point channel has at least two capacity-achieving input distributions and their dispersions $V_1$ \eqref{eq:dispersionVk} are distinct, then using one-bit feedback improves the achievable second-order term. Although rate-0 feedback does not change the capacity region of a discrete memoryless MAC \cite{sarwate2009feedback}, in light of \cite{wagner2020ANew} it is plausible that even one-bit feedback can improve the achievable second-order term for some MACs.} (see \secref{sec:comparison} for details.)

To better understand \thmref{thm:ach}, consider a channel satisfying \eqref{eq:silence}--\eqref{eq:moment2a} for which the same distribution $P_X$ maximizes $I_k$ for all $k$.  For example, for the adder-erasure RAC in \eqref{eq:addererasure}, setting $P_X$ to be Bernoulli(1/2) maximizes $I_k$ for all $k$. By \lemref{lem:mutualinfofriendly}, for $M$ large enough and any $\epsilon_1, \epsilon_2, \ldots, \epsilon_K$, one can pick $n_1 < n_2 < \cdots < n_K$ so that equality holds in \eqref{eq:mainresult} for all $k$.
Therefore, \thmref{thm:ach} certifies that for some channels, rateless codes with encoders that are, until feedback, agnostic to the transmitter activity pattern perform as well in both first- and second-order terms as the best-known scheme  \cite{ huang2012finite, jazi2012simpler, tan2014dispersions,  scarlett2015constantcompMAC} designed with complete knowledge of transmitter activity. Moreover for any fixed $0 < \epsilon_0 < 1$, the probability that at time $n_0 \geq c_0 \log n_1 + o(\log n_1)$ the decoder correctly detects the scenario where no transmitters are active is no smaller than $1-\epsilon_0$. Thus, a new epoch can begin very quickly when no transmitters are active in the current epoch. 
  
  The constant $c_0$ in \eqref{eq:n0logn} depends on the output distributions $P_{Y_k}$, $k = 0, \dots, K$, and on the hypothesis test chosen in \secref{sec:analysis_0test} but not on the target probability of error $\epsilon_0$. In contrast, the $o(\log n_1)$ term in \eqref{eq:n0logn} depends on $\epsilon_0$. See \secref{sec:analysis_0test} (eq. \eqref{eq:huangn0}) for an example where we bound the dependence of the $o(\log n_1)$ term on $\epsilon_0$ under the log-likelihood ratio test. 
  
 Our achievability result in \thmref{thm:ach} assumes that the total number of transmitters, $K$, is constant. The asymptotic regime in which $K$ grows with the decoding times, $n_1, n_2, \dots, n_K$, seeks to characterize scenarios with massive numbers of communicators \cite{polyanskiy2017perspective, chen2017capacity, scarlettPhase2016}. Understanding the fundamental limits of random access communications in that regime presents an interesting challenge for future work.
  
 \subsection{Comparison With the Existing Achievability Results} \label{sec:comparison}
 \subsubsection{Discrete Memoryless RACs}
Our achievable region (\thmref{thm:ach}) is consistent with the achievability results for the 2-transmitter MACs given in \cite{huang2012finite, jazi2012simpler, tan2014dispersions, scarlett2015constantcompMAC}. The proofs in \cite{huang2012finite, jazi2012simpler, tan2014dispersions} use i.i.d. random code design, an approach that we follow in \thmref{thm:ach}. In \cite{scarlett2015constantcompMAC}, Scarlett \textit{et al.} use constant-composition codes.
In \cite{huang2012finite, jazi2012simpler, tan2014dispersions}, the achievable rate region of a discrete memoryless MAC is expressed as a three-dimensional vector inequality that relies on a $3 \times 3$ dispersion matrix $\mathsf{V}_2$ defined in \cite[eq.~(48)]{tan2014dispersions}; the entry of $\mathsf{V}_2$ at location $(3, 3)$ is $V_2$ \eqref{eq:dispersionVk} for some input distribution $(P_{X_1}, P_{X_2})$. 
For rate pairs approaching interior (i.e., non-corner) points on the sum-rate boundary for $(P_{X_1^*}, P_{X_2^*})$, i.e., rate pairs satisfying
\begin{align}
    (R_1, R_2) \in \big\{ &(r_1 + o(1), r_2 + o(1)) \colon \notag \\
    &r_1 < I_2(X_1^*; Y_2^*|X_2^*) \notag \\
    &r_2 < I_2(X_2^*; Y_2^*|X_1^*) \notag \\
    &r_1 + r_2 = I_2(X_1^*, X_2^*; Y_2^*) \big\}, \label{eq:rateregion}
\end{align}
the achievable region in \cite{huang2012finite, jazi2012simpler, tan2014dispersions} reduces to the scalar inequality
\begin{align}
R_1 + R_2 \leq I_2^* -\sqrt{\frac{V_2^*}{n}}Q^{-1}(\epsilon) + O\left( \frac{\log n}{n} \right), \label{eq:L1L2}
\end{align}
where 
\begin{align}
     I_2^* \triangleq I_2(X_1^*, X_2^*; Y_2^*) 
\end{align}
is the sum-rate capacity and $V_2^*$ is the dispersion $V_2$ \eqref{eq:dispersionVk} evaluated using $(P_{X_1^*}, P_{X_2^*})$. The bound in \eqref{eq:L1L2} implies that
the only component of $\mathsf{V}_2$ employed in the second-order characterization of the region \eqref{eq:rateregion} is $V_2^*$.  
The result in \eqref{eq:L1L2} is proved in \cite[Prop.~4 case ii)]{haim2012}. 

In \cite[Th.~1]{scarlett2015constantcompMAC}, Scarlett \textit{et al.} use constant-composition codes to show that the dispersion matrix $\mathsf{V}_2$ in the second-order achievable region can be improved to $\tilde{\mathsf{V}}_2$, defined in \cite[eq.~(13)]{scarlett2015constantcompMAC}. Further, they show that $\tilde{\mathsf{V}}_2 \preceq \mathsf{V}_2$, where $\preceq$ designates positive semidefinite order. Therefore, the second-order rate region that is obtained using constant-composition codes includes that achieved with i.i.d. random coding when the target error probability satisfies $\epsilon < \frac 1 2$. Scarlett \textit{et al.} \cite{scarlett2015constantcompMAC} present two examples for which $\tilde{\mathsf{V}}_2 \prec \mathsf{V}_2$, demonstrating that the inclusion can be strict. The $(3,3)$ component of $\tilde{\mathsf{V}}_2$ is
\begin{align}
    \tilde{V}_2^* &= V_2^* - \Var{\E{\imath_2(X_1^*, X_2^*; Y_2^*) | X_1^*}} \notag \\
    & \quad- \Var{\E{\imath_2(X_1^*, X_2^*; Y_2^*) | X_2^*}}, \label{eq:V2tilde}
\end{align}
where $P_{X_1^*} P_{X_2^*} P_{Y_2^*|X_1^*, X_2^*} = P_{X_1^*} P_{X_2^*} P_{Y_2|X_1, X_2}$. The right side of \eqref{eq:L1L2} is achievable with $V_2^*$ replaced by $\tilde{V}_2^*$. In \lemref{lem:scarlett}, below, we derive a saddle point condition for general MACs without cost constraints. \lemref{lem:scarlett} implies that
\begin{align}
     \tilde{V}_2^* = V_2^*. \label{eq:V2tildestar}
\end{align}
This means that while constant-composition code design can yield achievability results with second-order terms superior to those derived through i.i.d. code design, on the sum-rate boundary that superior performance is observed only at corner points. For any rate point approaching an interior point on the sum-rate boundary, the i.i.d. random code design employed in this paper achieves first- and second-order performance identical to that achieved by constant-composition code design.

\begin{lemma} \label{lem:scarlett}
Let $P_{Y_2|X_1, X_2}$ be a 2-transmitter MAC with finite sum-rate capacity. Assume that the $\sigma$-algebra on the abstract input alphabets $\mathcal{X}_i$ includes all singletons on $\mathcal{X}_i$, $i = 1, 2$.
Let $(X_1^*, X_2^*, Y_2^*) \sim P_{X_1^*} P_{X_2^*} P_{Y_2|X_1, X_2}$, where $(P_{X_1^*}, P_{X_2^*})$ is a sum-rate capacity achieving input distribution, i.e., 
\begin{align}
    I_2^* \triangleq I_2(X_1^*, X_2^*; Y_2^*) = \sup\limits_{P_{X_1} P_{X_2}} I_2(X_1, X_2; Y_2) < \infty. \label{eq:I2star}
\end{align}
Then, for $i = 1, 2$, 
\begin{align}
    \E{\imath_2(X_1^*, X_2^*; Y_2^*) | X_i^*} &= I_2^*, \label{eq:EcondC}
\end{align}
where \eqref{eq:EcondC} holds $P_{X_i^*}$-almost surely.
\end{lemma}
\begin{IEEEproof}
See \appref{app:scarlett}. 
\end{IEEEproof}
A version of \lemref{lem:scarlett} for discrete memoryless MACs appears in \cite[Prop.~1]{watanabe1996}. The result is proved by verifying that \eqref{eq:EcondC} satisfies the Karush-Kuhn-Tucker (KKT) conditions for the maximization problem in \eqref{eq:I2star} (Although the maximization problem in \eqref{eq:I2star} is not convex, it satisfies a regularity condition ensuring the necessity of the KKT conditions for optimality \cite{watanabe1996}.) We extend \cite[Prop.~1]{watanabe1996} to general MACs by demonstrating a saddle point condition for MACs. The saddle point condition is more general in the sense that it applies to abstract alphabets.

From \eqref{eq:EcondC}, we deduce that
\begin{align}
    \Var{\E{\imath_2(X_1^*, X_2^*; Y_2^*) | X_i^*}} = 0, \quad i = 1, 2. \label{eq:VcondC}
\end{align}
Substituting \eqref{eq:VcondC} into \eqref{eq:V2tilde}, we obtain \eqref{eq:V2tildestar}.

The result in \eqref{eq:I2star}--\eqref{eq:EcondC} extends the following well-known properties of point-to-point DMCs to MACs. In \cite[Th.~4.5.1]{gallager1968}, the KKT conditions in \eqref{eq:I2star}--\eqref{eq:EcondC} for point-to-point DMCs are
\begin{align}
    I_1^* &\triangleq  \max_{P_{X_1}} I_1(X_1; Y_1) \\
    \E{\imath_1(X_1^*; Y_1^*) | X_1^*} &= I_1^* \quad \textrm{ if } P_{X_1^*}(x_1) > 0 \label{eq:EcondCpoint} \\
    \E{\imath_1(X_1^*; Y_1^*) | X_1^* = x_1} &\leq I_1^* \quad \textrm{ if } P_{X_1^*}(x_1) = 0; \label{eq:EcondCpoint2}
\end{align}
these conditions are necessary and sufficient for optimality.
As noted in \cite[Lemma~62]{polyanskiy2010Channel}, \eqref{eq:EcondCpoint}--\eqref{eq:EcondCpoint2} indicate that for a capacity-achieving input distribution $P_{X_1^*}$, 
\begin{align}
     \Var{\E{\imath_1(X_1^*; Y_1^*) | X_1^*}} = 0. \label{eq:VarEis0}
\end{align}
From \eqref{eq:VarEis0} and the law of total variance, it follows that the unconditional and conditional variances of $\imath_1(X_1^*; Y_1^*)$ given $X_1^*$ are equal, i.e.,
\begin{align}
    V_1 = \E{\Var{\imath_1(X_1^*; Y_1^*)|X_1^*}}. \label{eq:V1E}
\end{align} 
For point-to-point DMCs, Moulin \cite{moulin2012cc} shows that the second-order term $\tilde{V}_1$ achievable using constant-composition coding equals the right-hand side of \eqref{eq:V1E}, meaning that i.i.d. random code design and constant-composition random code design achieve the same fundamental limits for point-to-point DMCs.

\subsubsection{The Gaussian RAC}
While the RAC code definition (Definition~\ref{def:codecompound}) does not impose cost constraints on the codewords, cost constraints can be added where needed. In the case of the Gaussian RAC defined in \eqref{eq:Gaussian}, the maximal power constraint $P$ on the codewords requires that
\begin{align}
    \left\lVert\sff(u, w)^{n_k} \right\rVert^2 \leq n_k P \label{eq:maxpower}
\end{align}
for all $u \in \mathcal{U}$, $w \in [M]$, and $k \in [K]$, where $\left\lVert\cdot\right\rVert$ denotes the Euclidean norm. If any encoder attempts to transmit a codeword that does not satisfy \eqref{eq:maxpower}, we count that event as an error. Hence, the maximal power constraints add the term 
\begin{align}
    \Prob{\bigcup_{j = 1}^k \bigcup_{i = 1}^k \left\{\left \lVert X_i^{n_j} \right \rVert^2 > n_j P \right\}}
\end{align}
to the error terms in \eqref{eq:errorcriterion}. 

For the Gaussian $k$-MAC under maximal power constraints, drawing codewords i.i.d. according to distribution $P_X \sim \mathcal{N}(0, P-\delta_{n_k})$ for any $\delta_{n_k} \to 0$ as $n_k \to \infty$ yields a worse second-order performance bound than the one achieved by drawing codewords uniformly at random from the $n_k$-dimensional power sphere \cite{molavianjazi2014On, molavianjazi2015second}. MolavianJazi and Laneman \cite{molavianjazi2015second} and Scarlett \textit{et al.} \cite{scarlett2015constantcompMAC} derive the improved second-order term for the Gaussian MAC by drawing codewords uniformly at random over an $n_k$-dimensional power sphere and by combining constant-composition code design with a quantization argument, respectively. In \cite{yavas2020Gaussian}, for the Gaussian MAC and RAC, we prove the achievability of the same second-order term as \cite{molavianjazi2015second, scarlett2015constantcompMAC} with an improved third-order term $\frac{1}{2} \log n_k$. The proof employs codewords designed by concatenating spherically distributed sub-blocks and a maximum likelihood decoding rule combined with a threshold rule based on the output power.

\subsection{An Example RAC}
The following example investigates rates achievable for the adder-erasure RAC in \eqref{eq:addererasure}.
\begin{example}
\normalfont
For the adder-erasure RAC, the capacity achieving distribution is the equiprobable (Bernoulli(1/2)) distribution for all $k$. (See the proof of \thmref{thm:ikvk} in Appendix \ref{app:adder}.) For this channel, one can exactly calculate $I_k$ and $V_k$ for this channel for every $k$ (labelled ``True'' in \figref{fig:ikvk}). The approximating characterizations
\begin{align}
&I_k =   (1-\delta) \left( \frac {1} 2  \log \frac{\pi e k}{2} - \frac{ \log e}{12 k^2 } \right)+ O(k^{-3}) \label{eq:approxIk}\\
&V_k = (1 - \delta) \Bigg[ \frac{\delta}{4} \log^2 \frac{\pi e k}{2} + \frac{\log^2 e}{2}  - \frac{ \log^2 e}{2 k} \notag \\
&\quad - \left( \frac{ \log e}{2} + \frac{\delta \log \frac{\pi e k}{2} }{12} \right) \frac{\log e}{k^2} \Bigg] + O\left( \frac{\log k}{k^3} \right)  \label{eq:approxVk},
\end{align}  
which capture the first- and second-order behavior of $I_k$ and $V_k$ for each $k$, are, nonetheless, useful since they highlight how each depends on $k$ and $\delta$. These values, without the $O(\cdot)$ terms in \eqref{eq:approxIk}--\eqref{eq:approxVk}, are labelled ``Approximation" in \figref{fig:ikvk}.
The approximations are quite tight even for small $k$. Both $I_k$ and $\sqrt{V_k}$ are of order $O(\log k)$, indicating that as $k$ grows, the sum-rate capacity grows, albeit slowly, while the per-user rate vanishes as $\bigo{\frac {\log k}{k}}$. The dispersion $V_k$ also grows, and the speed of approach to the sum-rate capacity is slower. Interestingly, the dispersion behavior is different for the pure adder RAC $(\delta = 0)$, in which case $V_k = \frac{\log^2 e}{2}  + \bigo{\frac{1}{k}}$ is almost constant as a function of $k$. The derivation of \eqref{eq:approxIk} and \eqref{eq:approxVk} relies on an approximation for the probability mass function of the $(k, 1/2)$ Binomial distribution using a higher order Stirling's approximation (\appref{app:adder}). 
\begin{figure}[htbp] 
    \centering
    \subfloat[]{
    	\includegraphics[width=0.46\textwidth]{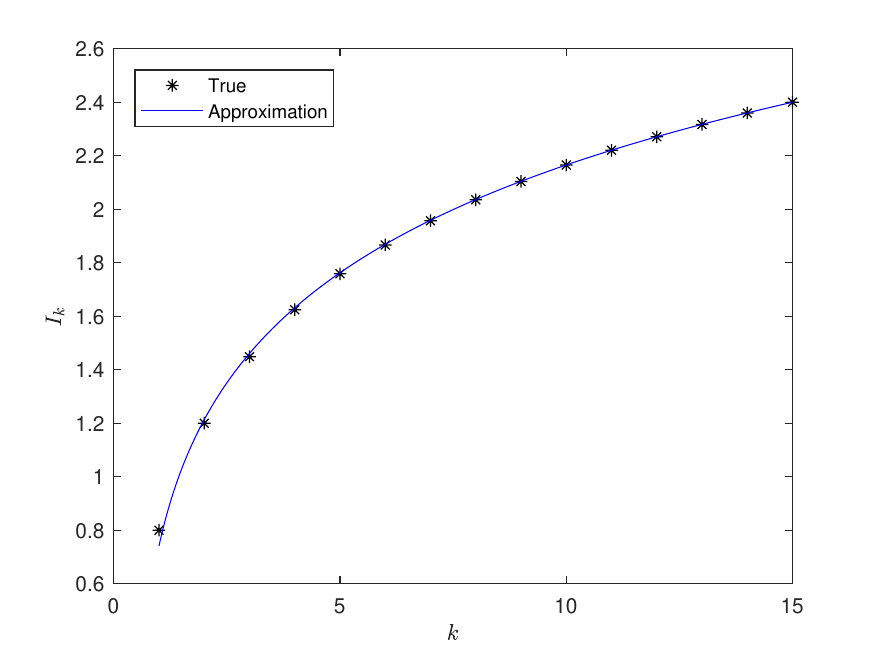} 
    	\label{fig:ik}} \\
    \subfloat[]{	
    	\includegraphics[width=0.46\textwidth]{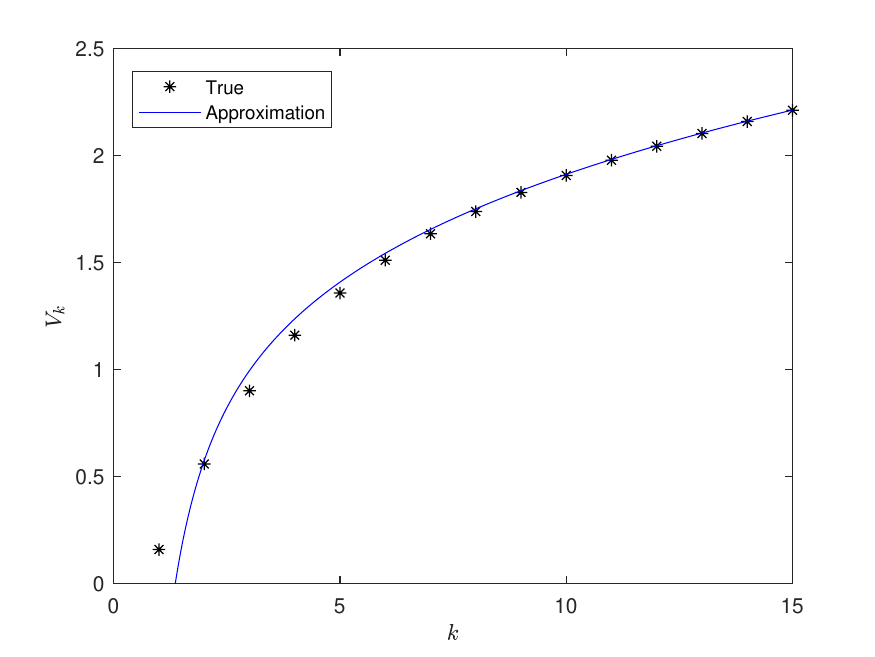}
    	\label{fig:vk}}
    \caption{\label{fig:ikvk} (a) Sum-rate capacity $I_k$ (in bits) and (b) dispersion $V_k$ (in $\mbox{bits}^2$) for the adder-erasure RAC with $\delta = 0.2$.}
\end{figure} 
\begin{figure}[htbp]  
    \centering
    \includegraphics[width=0.44\textwidth]{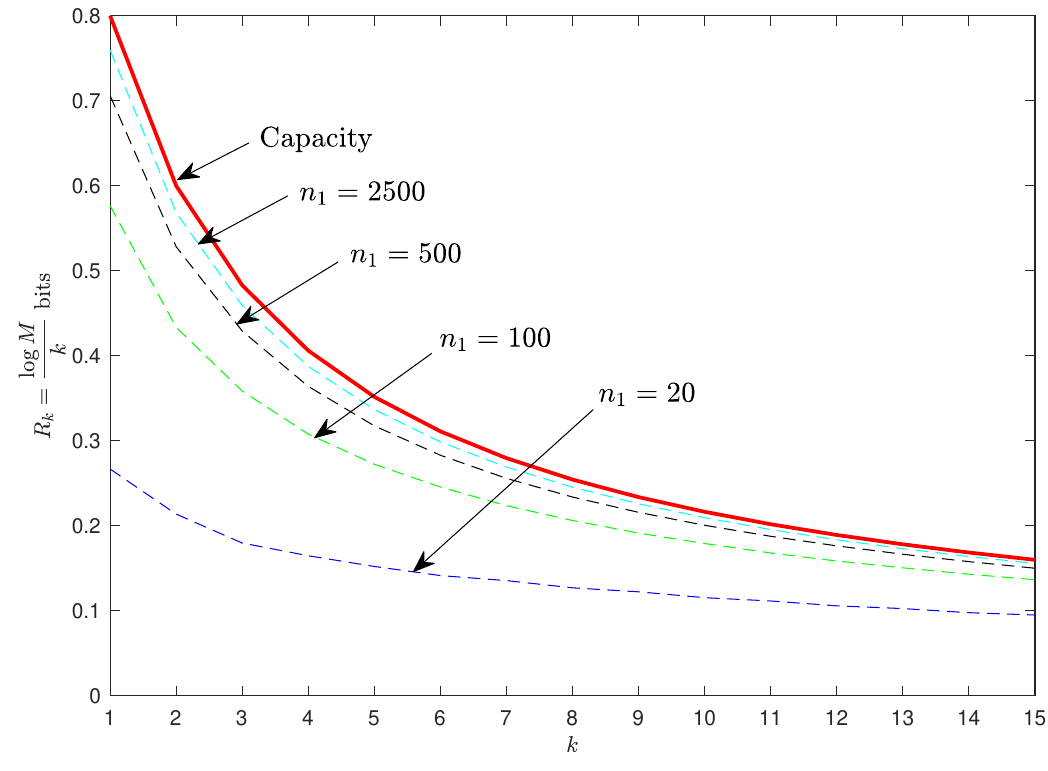}
    \caption{Capacity and approximate achievable rates (in bits per user) for the adder-erasure RAC with erasure probability $\delta = 0.2$ are given for the target error probability $\epsilon_k = 10^{-6}$ for all $k$. For each curve, the message size $M$ is fixed so that the rates $\{R_k\}$ are achievable with $n_1$ set to $20, 100, 500,$ and $2500$, respectively.}
    \label{fig:rk}
\end{figure} 

\figref{fig:rk} shows the approximate rate per transmitter, $R_k = \frac {\log M}{n_k}$ (neglecting the $O(1)$ term in \eqref{eq:mainresult}), achieved by the proposed scheme as a function of the number of active transmitters, $k$, and the choice of blocklength $n_1$ for a fixed error probability $\epsilon_k = 10^{-6}$ for all $k$. Fixing $n_1$ and $\epsilon_k$ fixes the maximum achievable message size, $M$, according to \eqref{eq:mainresult}. The remaining $n_k$ for $k \geq 2$ are found by choosing the smallest $n_k$ that satisfies \eqref{eq:mainresult} using the given $M$ and $\epsilon_k$. Each curve illustrates how the rate per transmitter ($R_k$) decreases as the number of active users $k$ increases. The curves differ in their choice of blocklength $n_1$ and the resulting changes in $M$ and $n_0, n_2, \dots, n_K$. Here $n_1$ is fixed to $20, 100, 500$ and $2500$. For a fixed $k$, the points on the same vertical line demonstrate how the gap between the per-user capacity and the finite-blocklength achievable rate decreases as blocklength increases.
\end{example}

\subsection{A Non-asymptotic Achievability Result}
Theorem~\ref{thm:ach} follows from Theorem~\ref{thm:nonasymp}, stated next,
which bounds the error probability of the RAC code 
defined in Section~\ref{sec:raccode}.  
When $k$ transmitters are active, the error probability $\epsilon_k$ captures both errors in the estimate $t$ of $k$ 
and errors in the reproduction $\hW_{[t]}$ of $W_{[k]}$ when $t=k$. Theorem~\ref{thm:nonasymp} is formulated for an arbitrary choice of a statistic $h\colon \mathcal{Y}^{n_0} \mapsto \mathbb{R}$ used to decide whether any transmitters are active. Possible choices for $h(\cdot)$ appear in \eqref{eq:hoefftest} and \eqref{eq:ks_test} in \secref{sec:analysis_0test}, below.

\begin{theorem}\label{thm:nonasymp}
Fix constants $\gamma_0$, $\lambda_{s, t}^k \geq 0$, and $\gamma_t > 0$ for all $1 \leq s \leq t \leq k$.
For any RAC \[\left\{\left(\cX^k, P_{Y_k | X_{[k]}}(y_k|x_{[k]}), \cY_k\right)\right\}_{k=0}^K\] satisfying \eqref{eq:permutationinvariance} and \eqref{eq:reducible}, any $K \leq \infty$\footnote{Note that while \thmref{thm:ach} requires $K<\infty$, \thmref{thm:nonasymp} allows $K=\infty$. For $K=\infty$, \eqref{eq:nonasymptotic} holds for every finite $k$ since the bound on $\epsilon_k$ depends only on the RAC with at most $k$ active transmitters.}, and any fixed input distribution $P_X$, there exists an $(M, \lbrace (n_k, \epsilon_k) \rbrace_{k = 0}^K )$ code such that
\begin{align}
    \epsilon_0 \leq \Prob{ h(Y_0^{n_0}) > \gamma_0  }, \label{eq:0treps0}
\end{align}
and for all $k\geq 1$,
		\begin{IEEEeqnarray}{rCl}
        \IEEEyesnumber \label{eq:nonasymptotic} \IEEEyessubnumber*
        \epsilon_k &\leq& \bbP[\imath_k(X_{[k]}^{n_k}; Y_k^{n_k}) \leq \log \gamma_k] \label{eq:dominating} \\
        &&+ \Prob{  h(Y_k^{n_0}) \leq \gamma_0  } \label{eq:0transmitter} \\
        &&+ \frac{k(k-1)}{2M} \label{eq:repetition} \\
        &&+\sum_{t = 1}^{k - 1} \binom{k}{t} \bbP[\imath_t(X_{[t]}^{n_t}; Y_k^{n_t}) > \log \gamma_t ] \label{eq:wrongtime} \\
        &&+\sum_{t = 1}^k \sum_{s = 1}^{t-1} \binom{k}{t-s}  \mathbb{P}\Big[\imath_t(X_{[s+1:t]}^{n_t}; Y_k^{n_t}) \notag \\
        &&\quad > n_t \mathbb{E} [\imath_t(X_{[s+1:t]}; Y_k) ] + \lambda_{s, t}^k \Big] \label{eq:confususer1} \\
        &&+ \sum_{t = 1}^k\sum_{s = 1}^t \binom{k}{t-s} \binom{M-k}{s}  \mathbb{P} \Big[\imath_t( \bar{X}_{[s]}^{n_t} ; Y_k^{n_t} | X_{[s+1: t]}^{n_t}) \notag \\
        &&\quad > \log \gamma_t - n_t \mathbb{E} [\imath_t(X_{[s+1:t]}; Y_k) ] - \lambda_{s, t}^k \Big],  \label{eq:confususer} 
		\end{IEEEeqnarray}
where for any $n$, $(X_{[k]}^n, \bar{X}_{[k]}^n, Y_{k}^n) $ is a random sequence drawn i.i.d. $\sim P_{X_{[k]} \bar{X}_{[k]} Y_{k}}(x_{[k]}, \bar{x}_{[k]}, y_{k}) = \left( \prod_{i = 1}^k P_X(x_i) P_X(\bar{x}_i) \right) P_{Y_k|X_{[k]}}(y_k | x_{[k]}) $.
\end{theorem} 

The operational regime of interest is when $\epsilon_0, \dots, \epsilon_k$ are constant; that is, $\epsilon_k$ does not vanish as $n_k$ grows. For $k = 0$, the error term in \eqref{eq:0treps0} is the probability that the decoder does not correctly determine that the number of active transmitters is 0 at time $n_0$. For $k > 0$,  \eqref{eq:dominating} is the probability that the true codeword set produces a low information density. This is the dominating term in the regime of interest. All remaining terms are negligible, as shown in the refined asymptotic analysis of the bound in \thmref{thm:nonasymp} (see Section~\ref{sec:achproof}, below.)
The remaining terms bound the probability that the decoder incorrectly estimates the number of active transmitters as 0 \eqref{eq:0transmitter}, the probability that two or more transmitters send the same message \eqref{eq:repetition},\footnote{Given the use of identical encoders, multiple encoders sending the same message can be beneficial or harmful, depending on the channel. To simplify the analysis, we treat this (exponentially rare) event as an error.} the probability that the decoder estimates the number of active transmitters as $t$ for some $1 \leq t<k$ and decodes those $t$ messages correctly \eqref{eq:wrongtime}, and the probability that the decoder estimates the number of active transmitters as $t$ for some $1 \leq t\leq k$ and decodes to $s$ messages that were not transmitted and $t-s$ messages that were transmitted \eqref{eq:confususer1}--\eqref{eq:confususer}. 

For $k = 1, 2$, the expression in \eqref{eq:nonasymptotic} particularizes to
 \begin{align}
 \epsilon_1 &\leq \bbP[\imath_1(X_1^{n_1}; Y_1^{n_1})\leq \log \gamma_1 ] + \Prob{  h(Y_1^{n_0} ) \leq \gamma_0  }\notag \\
  &\quad + (M-1) \bbP[ \imath_1(\bar{X}_1^{n_1} ; Y_1^{n_1}) > \log \gamma_1 - \lambda_{1,1}^1 ] \\
 \epsilon_2 &\leq \bbP[\imath_2(X_{[2]}^{n_2}; Y_2^{n_2})\leq \log \gamma_2 ] + \Prob{  h(Y_2^{n_0} ) \leq \gamma_0  } \nonumber\\
 &\quad + \frac{1}{M} + 2 \bbP[\imath_1(X_1^{n_1} ; Y_2^{n_1}) > \log \gamma_1 ] \nonumber \\
 &\quad + 2 \bbP[ \imath_2(X_{2}^{n_2}; Y_2^{n_2}) \geq n_2 I_2(X_2; Y_2) + \lambda_{1,2}^2 ] \nonumber \\
 &\quad + (M-1)\bbP[ \imath_1(\bar{X}_{1}^{n_1}; Y_2^{n_1}) > \log \gamma_1 - \lambda_{1,1}^2 ]  \notag\\
&\quad + 2 (M-2) \bbP[ \imath_2(\bar{X}_{1}^{n_2}; Y_2^{n_2} | X_{2}^{n_2} ) \nonumber \\
 &\quad  \quad > \log \gamma_2 - n_2 I_2(X_2; Y_2) - \lambda_{1, 2}^2] \notag\\
 &\quad  + \frac{(M-2)(M-3)}2 \bbP[ \imath_2(\bar{X}_{[2]}^{n_2}; Y_2^{n_2} ) > \log \gamma_2 - \lambda_{2,2}^2].
\end{align}
    
For the MAC with $K$ transmitters, i.e., the scenario where $K$ transmitters are always active, the only decoding time is $n_K$. The error terms associated with incorrect decoding times are no longer needed in this case, and the error probability bound in \eqref{eq:nonasymptotic} becomes
\begin{IEEEeqnarray}{rCl}
        \IEEEyesnumber \label{eq:MACbound} \IEEEyessubnumber*
        \epsilon_K &\leq& \bbP[\imath_K(X_{[K]}^{n_K}; Y_K^{n_K}) \leq \log \gamma_K] + \frac{K(K-1)}{2M} \\
        &&+ \sum_{s = 1}^{K-1} \binom{K}{K-s}  \mathbb{P}\Big[\imath_K(X_{[s+1:K]}^{n_K}; Y_K^{n_K}) \notag \\
        &&\quad > n_K \mathbb{E} [\imath_K(X_{[s+1:K]}; Y_K) ] + \lambda_{s, K}^K \Big]  \label{eq:MACbound1} \\
        &&+ \sum_{s = 1}^K \binom{K}{K-s} \binom{M-K}{s}  \mathbb{P} \Big[\imath_K( \bar{X}_{[s]}^{n_K} ; Y_K^{n_K} | X_{[s+1: K]}^{n_K}) \notag \\
        &&\quad > \log \gamma_K - n_K \mathbb{E} [\imath_K(X_{[s+1:K]}; Y_K) ] - \lambda_{s, K}^K 
        \Big]. \IEEEeqnarraynumspace \label{eq:MACbound2}
\end{IEEEeqnarray}

A description of the proposed RAC code and the proofs of Theorems~\ref{thm:ach} and \ref{thm:nonasymp} 
appear in Section~\ref{sec:raccode}. 

\section{The RAC Code and Its Performance}\label{sec:raccode}
\subsection{Code Design}\label{sec:codedesign}

We construct the RAC code used in the proofs of Theorems~\ref{thm:ach} and \ref{thm:nonasymp} 
as follows.

\noindent
{\bf Encoder Design:} The common randomness random variable $U = (U(1), \dots, U(M))$ has distribution
\begin{align}
P_U &\triangleq P_{U(1)}  \times \cdots \times  P_{U(M)}, \label{eq:PUtriangleq}
\end{align}
where $P_{U(w)} = P_X^{n_K}$, $w = 1, \dots, M$, and $P_{X}$ is a fixed distribution on alphabet $\cX$. Each realization of $U$ defines a codebook with $M$ i.i.d. vectors $U(1), \dots, U(M)$ of dimension $n_K$ (the codewords). Note that the cardinality of the alphabet $U$ is $|\cX|^{M n_K}$. In \cite[Th.~19]{polyanskiy2011feedback}, Polyanskiy \textit{et al.} use Carath\'eodory's Theorem to show that the common randomness $U$ can be replaced with common randomness $U'$ with cardinality at most $K + 2$. We reduce this alphabet size to $K + 1$ in \appref{app:caratheodory}. As described in Definition~\ref{def:codecompound}, 
an $(M,\{(n_k,\epsilon_k)\}_{k = 0}^K)$ RAC code with identical encoders employs the same encoder $\sff(\cdot)$ at every transmitter. The encoder $\sff(U, \cdot)$ depends on $U$ as
\begin{align}
\sff(U, w) = U(w) \quad \mbox{ for } w = 1, \dots, M.
\end{align}
For brevity, we omit $U$ in the encoding and decoding functions and write $\sff(U, w) = \sff(w) $ for $w = 1, \dots, M$, and $\sfg_k(U, y^{n_k}) = \sfg_k(y^{n_k})$ for $y^{n_k} \in \mathcal{Y}_K^{n_k}, k \in \{0, \dots, K\}$. Recall that $\sff(w)$ is a $n_K$-dimensional vector. We use $\sff(w)^{n_k}$ to denote the first $n_k$ coordinates of vector $\sff(w)$.
For any collection of messages $w_{[k]}\in[M]^k$, we use $\sff\left(w_{[k]}\right) \triangleq \left(\sff(w_1),\ldots,\sff(w_k)\right)$ 
to denote the collection of encoded descriptions produced by the encoders.

{\bf Decoder Design:} 
Upon receiving the first $n_0$ samples of the channel output $Y$, the decoder runs the following composite hypothesis test
\begin{align}
\sfg_0(y^{n_0})=
        \left\{\begin{array}{cl}
        \es & \mbox{if ~}
        h(y^{n_0}) \leq \gamma_0 \\
        \sfe & \mbox{otherwise}
        \end{array}\right. \label{eq:0test}
\end{align}
to decide whether there are any active transmitters.
Decoder output $\es$ signifies that the decoder decides that all transmitters are silent, sending a feedback bit `1' to all transmitters to start a new coding epoch. Decoder output $\sf{e}$ indicates that the receiver believes that there are active transmitters; the decoder transmits feedback bit `0' to the transmitters, telling them that it is not ready to decode, and therefore that transmissions must continue. Statistic $h\colon \, \mathcal{Y}^{n_0} \mapsto \mathbb{R}$ is used to decide whether any transmitters are active.

For each $k \geq 1$, decoder $\mathsf g_k$ observes output $y^{n_k}$ and employs a single threshold rule 
\begin{align}
\sfg_k(y^{n_k})=
        \begin{cases}
        w_{[k]} & \mbox{if ~}
        \imath_k(\sff\left(w_{[k]}\right)^{n_k};y^{n_k})>\log\gamma_k \\
        & \phantom{if ~}\mbox{and } w_i < w_j \,\,\forall \, i < j \\ 
        \sfe & \mbox{otherwise} 
        \end{cases}
        \label{eq:thres}
\end{align}
 for some constant $\gamma_k$ chosen before the transmission starts. 
Under permutation-invariance \eqref{eq:permutationinvariance} and identical encoding \eqref{eq:f}, all permutations of the message vector $w_{[k]}$ give the same information density. We use the ordered permutation specified in \eqref{eq:thres} as a representative of the equivalence class with respect to the binary relation $\stackrel{\pi}{=}$. The choice of a representative is immaterial since decoding is identity-blind.
When there is more than one ordered $w_{[k]}$ that satisfies the threshold 
condition, decoder $\mathsf g_k$ chooses among these options arbitrarily.  
All such events are counted as errors in the analysis in \secref{sec:thmpe}, below. If the decoder output is a message vector $w_{[k]}$, then the decoder sends feedback bit `1', telling them to stop transmission. Otherwise, the decoder sends feedback bit `0', and the epoch continues. For $k \geq 1$, the decoder $\sfg_k(y^{n_k})$ depends on $U$ through its dependence on the encoding function $\sff\left(w_{[k]}\right)$; for $k = 0$, $\sfg_0(y^{n_0})$ does not depend on $U$. 

The proof of Theorem~\ref{thm:nonasymp}, below, bounds the error probability for the proposed code.

\subsection{Proof of Theorem~\ref{thm:nonasymp}}
\label{sec:thmpe}

In the discussion that follows, we bound the error probability of the code $(\sff,\{\sfg_k\}_{k = 0}^K)$ 
defined above.  For $k = 0$, the only error event is that the received vector at time $n_0$, $Y_0^{n_0}$, fails to pass the test
\begin{align}
\epsilon_0 &\leq \Prob{\sfg_0(Y_0^{n_0}) \neq \es | W_0 = 0}
\end{align}
given in \eqref{eq:0test}.
For $k > 0$, the analysis relies on the independence of codewords $\sff(W_i)$ and $\sff(W_j)$ 
from distinct transmitters $i$ and $j$.
Given identical encoders and i.i.d. codeword design, this assumption is valid provided that $W_i\neq W_j$; 
we therefore count events of the form $W_i=W_j$ as errors.
Let $\bbPrep$ denote the probability of such a repetition; the union bound gives
\begin{align}
    \bbPrep \leq \frac{k(k-1)}{2M}. \label{eq:Prep}
\end{align}

The discussion that follows uses
$\ws_{[k]}=(1, 2, \ldots,k)$ as an example instance of a message vector $w_{[k]}$ 
in which $w_i\neq w_j$ for all $i \neq j$. The set
$\tcW_{[s]}$ describes all ordered message vectors that do not share any messages in common with $\ws_{[k]}$, i.e., 
\begin{align}
\tcW_{[s]} \triangleq \{\tw_{[s]}\in[M]^s \colon \tw_1  > k, \tw_i < \tw_j \,\,\, \forall i < j\}. \label{eq:Ws}
\end{align}
Let the components of the vectors $(X_{[k]}^{n_k}, \bar{X}_{[k]}^{n_k}, Y_k^{n_k})$ be i.i.d. with joint distribution
\begin{align}
    &P_{X_{[k]} \bar{X}_{[k]} Y_k}(x_{[k]}, \bar{x}_{[k]}, y_k) \notag \\ &\quad = P_{X_{[k]}}(x_{[k]}) P_{X_{[k]}}(\bar{x}_{[k]}) 
    P_{Y_k | X_{[k]}}(y_k | x_{[k]}). \label{eq:PXXbarY}
\end{align}

\begin{table*}[htbp]
\newcounter{mytempeqncnt}
\normalsize
\setcounter{mytempeqncnt}{\value{equation}}
\setcounter{equation}{58}
\vspace*{4pt}
\begin{IEEEeqnarray}{rCl}
\epsilon_k &&=  \frac1{M^k}\sum_{w_{[k]}\in [M]^k} 
        \bbP [\{\sfg_0(Y_k^{n_0}) \neq \mathsf{e}\} \cup \{\cup_{t=1}^{k-1}\sfg_t(Y_k^{n_t})\neq\sfe\} 
\cup \{\sfg_k(Y_k^{n_k}) \stackrel{\pi}{\neq} w_{[k]}\} | W_{[k]} = w_{[k]} ] \IEEEyesnumber \label{eqn:achboundfirst} \\
&&\leq \bbPrep+(1-\bbPrep)\bbP[\{\sfg_0(Y_k^{n_0}) \neq \mathsf{e}\} \cup \{\cup_{t=1}^{k-1}\sfg_t(Y_k^{n_t})\neq\sfe\}
\cup
        \{\sfg_k(Y_k^{n_k}) \stackrel{\pi}{\neq} \ws_{[k]}\} | W_{[k]} =\ws_{[k]}] \label{eqn:rep} \\
&&\leq \bbPrep + \bbP[\sfg_0(Y_k^{n_0}) \neq \mathsf{e} | W_{[k]} = \ws_{[k]} ] \hspace{-0.5mm} +  \sum_{t=1}^{k-1}\binom{k}{t}\bbP[\sfg_t(Y_k^{n_t}) \stackrel{\pi}{=} \ws_{[t]}|W_{[k]}=\ws_{[k]}] \label{eqn:instance} \\
&& \quad     +\sum_{t=1}^k\sum_{s=1}^t\binom{k}{t-s}
        \bbP[\cup_{\tw_{[s]}\in\tcW_{[s]}}\{\sfg_t(Y_k^{n_t})\stackrel{\pi}{=}(\tw_{[s]},\ws_{[s+1:t]})\}    
  |W_{[k]}=\ws_{[k]}] + \bbP[ \sfg_k(Y_k^{n_k}) = \sfe | W_{[k]}=\ws_{[k]} ] \label{eqn:errors}  \\
&&\leq \frac{k(k-1)}{2M} + \Prob{h(Y_k^{n_0}) \leq \gamma_0} 
     +\sum_{t=1}^{k-1}\binom{k}{t}\bbP[
                \imath_t(X_{[t]}^{n_t};Y_k^{n_t})>\log\gamma_t]   \label{eqn:replacedecoder}\\
&&  \quad    +\sum_{t=1}^k\sum_{s=1}^t\binom{k}{t-s}
        \bbP[\cup_{\tw_{[s]}\in\tcW_{[s]}}\{\imath_t(\bar{X}_{[s]}^{n_t}(\tw_{[s]}),X_{[s+1:t]}^{n_t};Y_k^{n_t}) 
	>\log\gamma_t \} ] + \bbP[\imath_k(X_{[k]}^{n_k}; Y_k^{n_k}) \leq \log \gamma_k]  \label{eq:barterms}
\end{IEEEeqnarray}
\hrulefill
\vspace{-10pt}
\setcounter{equation}{\value{mytempeqncnt}}
\end{table*}
\setcounter{equation}{64}

Recall that the information density $\imath_t(x_{[t]}^{n_t}; y_t^{n_t})$ in \eqref{eq:imathnt} is defined with respect to $(X_{[t]}^{n_t}, Y_t^{n_t})$, not with respect to $(\bar{X}_{[t]}^{n_t}, Y_t^{n_t})$. The resulting error bound proceeds as shown in \eqref{eqn:achboundfirst}--\eqref{eq:barterms}; 
here $X_{[k]}$ is the vector of transmitted codewords, 
and $\bar{X}_{[s]}(\tw_{[s]})$ is an i.i.d. copy of $\bar{X}_{[s]}$, which represents the codeword for a collection of messages $\tw_{[s]} \in \tcW_{[s]}$ 
that was not transmitted.
Line (\ref{eqn:rep}) separates the case where at least one message is repeated
from the case where there are no repetitions.
Lines \eqref{eqn:instance}--\eqref{eqn:errors} enumerate the error events in the no-repetition case; 
these include 
        all cases where the transmitted codeword passes the binary hypothesis test \eqref{eq:0test} for ``no active transmitters" \eqref{eqn:instance}, 
        all cases where a subset of the transmitted codewords meets the threshold for some $t<k$ (\ref{eqn:instance}), 
        all cases where a codeword that is incorrect in $s$ dimensions and correct in $t-s$ dimensions meets the threshold for $t\leq k$ \eqref{eqn:errors}, and all cases where the transmitted codeword fails to meet the threshold \eqref{eqn:errors}. 
        We apply the union bound and the symmetry of the code design 
        to represent the probability of each case by the probability of an example instance times the number of instances. Equations \eqref{eqn:replacedecoder}-\eqref{eq:barterms} apply the bound in \eqref{eq:Prep} and replace decoders by the threshold rules in their definitions. 
    
The delay in applying the union bound to the first probability in \eqref{eq:barterms} is deliberate. It allows us to exploit the symmetry assumptions on the channel and to use a single threshold rule instead of $2^{k}-1$ threshold rules as in \cite{huang2012finite, jazi2012simpler, tan2014dispersions, scarlett2015constantcompMAC}. 
Applying the bound
\begin{IEEEeqnarray}{rCl}
\lefteqn{\bbP\Bigg[\bigcup_{\tw_{[s]}\in\tcW_{[s]}}\{\imath_t(\bar{X}_{[s]}^{n_t}(\tw_{[s]}),X_{[s+1:t]}^{n_t};Y_k^{n_t})>\log\gamma_t \}\Bigg]}  \\
& = & \bbP\Bigg[\bigcup_{\tw_{[s]}\in\tcW_{[s]}}\{\imath_t(\bar{X}_{[s]}^{n_t}(\tw_{[s]}),X_{[s+1:t]}^{n_t};Y_k^{n_t})>\log\gamma_t \} \notag \\
&&  \bigcap \left\{\imath_t(X_{[s+1:t]}^{n_t};Y_k^{n_t})>n_t\bbE[\imath_t(X_{[s+1:t]};Y_k)]+\lambda_{s,t}^k\right\}\Bigg] \notag  \\
& & +\bbP\Bigg[\bigcup_{\tw_{[s]}\in\tcW_{[s]}}\{\imath_t(\bar{X}_{[s]}^{n_t}(\tw_{[s]}),X_{[s+1:t]}^{n_t};Y_k^{n_t})>\log\gamma_t \} \notag \\
&&      \bigcap \left\{\imath_t(X_{[s+1:t]}^{n_t};Y_k^{n_t})\leq n_t\bbE[\imath_t(X_{[s+1:t]};Y_k)]+\lambda_{s,t}^k\right\}\Bigg] \notag\\
&\leq& \bbP\left[\imath_t(X_{[s+1:t]}^{n_t};Y_k^{n_t})>n_t\bbE[\imath_t(X_{[s+1:t]};Y_k)]+\lambda_{s,t}^k\right] \notag \\
& & +\bbP\Bigg[\bigcup_{\tw_{[s]}\in\tcW_{[s]}}\{\imath_t(\bar{X}_{[s]}^{n_t}(\tw_{[s]});Y_k^{n_t}|X_{[s+1:t]}^{n_t})>  \notag \\
&& \quad  \log\gamma_t -n_t\bbE[\imath_t(X_{[s+1:t]};Y_k)]-\lambda_{s,t}^k\}\Bigg] \label{eq:observation}
\end{IEEEeqnarray}
before applying the union bound to the first probability in \eqref{eq:barterms} yields a tighter bound. Combining \eqref{eq:barterms} and \eqref{eq:observation} and applying the union bound to the second probability in \eqref{eq:observation} completes the proof. 



\subsection{Proof of Theorem~\ref{thm:ach}}
\label{sec:achproof}
We fix $P_X$, $M$, $\{\epsilon_k\}_{k = 0}^K$, and we set the blocklengths $\{n_k\}_{k = 1}^K$ as
\begin{align}
    n_k & = \gamma_k^2 \left(\frac{e}{k}(M-k)\right)^{-2k}, \label{eq:choicenk}
\end{align}
where
\begin{align}
\log \gamma_k &=  n_k I_k - \tau_k \sqrt{n_k V_k} \label{setgamma} \\
\tau_k &\triangleq Q^{-1}\left(\epsilon_k-\frac{B_k+C_k}{\sqrt{n_k}}\right), \label{eq:tauk}
\end{align}
$C_k$ is a constant to be chosen in \eqref{eq:Ckchoice},
\begin{align}
    B_k &\triangleq \frac{6 T_k}{V_k^{3/2}} \label{eq:Bk}
\end{align}
is the Berry-Esseen constant~\cite[Chapter XVI.5 Th.~2]{feller1971introduction} 
(which is finite by the moment assumptions \eqref{eq:moment2}
and \eqref{eq:moment3}), and
\begin{align}
    T_k &\triangleq \E{\lvert\imath_k(X_{[k]} ; Y_k) - I_k \rvert^3}.
\end{align}

The choice of the threshold $\gamma_k$~\eqref{setgamma}
follows the approach established 
for the point-to-point channel in~\cite{polyanskiy2010Channel}. Solving \eqref{eq:choicenk} for $M$ and applying the Taylor series expansion to $Q^{-1}(\cdot)$, we see that the size of the codebook admits the following expansion 
\begin{align}
k \log M &=  n_k I_k - \sqrt{n_k V_k} Q^{-1} \left( \epsilon_k  \right) - \frac{1}{2}\log n_k + O(1) \label{eq:diff}
\end{align}
simultaneously for all $k \in [K]$. 
Note that the expansion in \eqref{eq:diff} is the best-known performance up to the second-order term for MACs without random access \cite{ huang2012finite, jazi2012simpler, tan2014dispersions,  scarlett2015constantcompMAC}, and we have chosen our parameters with the goal of matching that best prior performance. By \lemref{lem:mutualinfofriendly}, the resulting blocklengths satisfy $n_1 < n_2 < \cdots < n_K$ for $M$ large enough. 

We proceed to apply \thmref{thm:nonasymp} to show that under the given parameter choices, the probability of decoding error at time $n_k$ is bounded above by $\epsilon_k$.
The constants $\left\{\lambda_{s,t}^k\right\}$ used in the error probability bound \eqref{eq:confususer1}--\eqref{eq:confususer} are set as
\begin{align}
    \lambda_{s,t}^k & =  \frac{n_t}{2} \left(I_t(X_{[s]}; Y_t | X_{[s+1:t]}) 
	- \frac{s}{t} I_t \right) 	\label{eq:choiceDelta} 
\end{align}
to ensure that
$\lambda_{s,t}^k> 0$ when $s < t$ (see Lemma \ref{lem:mutualinfo}) 
and that $\lambda_{s,t}^k=0$ when $s = t$. Next, we sequentially bound
the terms in Theorem~\ref{thm:nonasymp} using the parameters chosen in \eqref{eq:choicenk}, \eqref{setgamma}, and \eqref{eq:choiceDelta}.
\begin{itemize}[leftmargin=*]
\item{\eqref{eq:dominating}:}
As noted previously, this is the dominant term. Since $\imath_k(X_{[k]}^{n_k};Y_k^{n_k})$
is a sum of $n_k$ independent random variables, by the Berry-Esseen theorem \cite[Chapter~XVI.5 Th.~2]{feller1971introduction} and \eqref{setgamma}--\eqref{eq:Bk},  
\begin{align}
\mathbb{P} \left[ \imath_k(X_{[k]}^{n_k} ; Y_k^{n_k}) \leq \log \gamma_k \right] \leq \epsilon_k - \frac{C_k}{\sqrt{n_k}}. \label{eq:outageprob}
\end{align}
\item{\eqref{eq:0transmitter}:} The test statistic $h(\cdot)$ and the threshold $\gamma_0$ given in \eqref{eq:0test} are chosen in Section \ref{sec:analysis_0test} to satisfy
\begin{align}
\Prob{h(Y_k^{n_0}) \leq \gamma_0} &\leq \frac{E_k}{\sqrt{n_k}} \label{eq:type2error} \\
\Prob{h(Y_0^{n_0}) > \gamma_0} &\leq \epsilon_0 \label{eq:type1error} 
\end{align}
for some constant $E_k > 0$. \lemref{lem:test}, below, bounds the type-II error in \eqref{eq:type2error} in terms of $n_0$ when the type-I error in \eqref{eq:type1error} is bounded by $\epsilon_0$.
\begin{lemma}\label{lem:test}
Fix $\epsilon_0 \in (0, 1)$. Assume that \eqref{eq:assump:Pyk} holds. Then there exists a test function $h(\cdot)$ such that \eqref{eq:type1error} is satisfied and
\begin{align}
\Prob{h(Y_k^{n_0}) \leq \gamma_0} \leq \exp\{-n_0 C' + o(n_0)\} \label{eq:Cprime}
\end{align}
for some $C' > 0$ depending on the output distributions $P_{Y_i}$ for $i = 0, \dots, K$.
\end{lemma}
\begin{IEEEproof}
See Section \ref{sec:analysis_0test}.
\end{IEEEproof}
From \eqref{eq:diff}, $n_k = \bigo{n_1}$ for $k \geq 1$. To make \eqref{eq:Cprime} behave as $O\left( \frac 1 {\sqrt{n_k}} \right)$ in \lemref{lem:test}, we pick $n_0$ as in \eqref{eq:n0logn} with $c_0 = \frac{1}{2C'}$.

\item{\eqref{eq:repetition}:} According to \eqref{eq:choicenk}, the upper bound $\frac{k(k-1)}{2M}$ on $\bbPrep$ in \eqref{eq:Prep} decays exponentially with $n_k$.
\item{\eqref{eq:wrongtime}:} 
Define $p$ as
\begin{align}
p \triangleq \bbP[\imath_t (X_{[t]} ; Y_k) > - \infty]. 
\end{align}
We next analyze \eqref{eq:wrongtime} for the cases $p = 1$ and $p < 1$.

Case 1: $p = 1$. 
By Lemma~\ref{lemma:expectation} and moment assumption \eqref{eq:moment2a}, 
\begin{align}
    I_t - \mathbb E \left[ \imath_t (X_{[t]} ; Y_k)\right] -\tau_t \sqrt{\frac{V_t}{n_t}} > 0 \label{eq:Itgreater0}
\end{align} 
for sufficiently large $n_t$. Chebyshev's inequality gives
\begin{align}
&\mathbb{P}[\imath_t (X_{[t]}^{n_t} ; Y_k^{n_t}) > \log \gamma_t ] \nonumber \\
&\quad \leq \frac{\textnormal{Var}[\imath_t (X_{[t]} ; Y_k)]}{n_t \left(I_t - \mathbb E \left[ \imath_t (X_{[t]} ; Y_k)\right] -\tau_t \sqrt{\frac{V_t}{n_t}}  \right)^2}. \label{eq:usersconfusion}
\end{align}
The right side of \eqref{eq:usersconfusion} 
behaves as $O\left(\frac1{n_t}\right)$.

Case 2: $p < 1$.
Here
\begin{align}
&\mathbb{P}[\imath_t (X_{[t]}^{n_t} ; Y_k^{n_t}) > \log \gamma_t ]  \nonumber \\
& \quad \leq \mathbb{P}[\imath_t (X_{[t]}^{n_t} ; Y_k^{n_t}) > -\infty ] \\
& \quad = p^{n_t}, \label{eq:pnt} 
\end{align}
where \eqref{eq:pnt} holds because $\imath_t (X_{[t]}^{n_t} ; Y_k^{n_t})$ is the sum of $n_t$ i.i.d. random variables, and that sum is greater than $-\infty$ if and only if all the summands satisfy the same inequality.
From \eqref{eq:usersconfusion} and \eqref{eq:pnt}, \eqref{eq:wrongtime} contributes $O\left(\frac1{n_k}\right)$ to our error bound. 

\item{\eqref{eq:confususer1}:} 
As in the analysis of \eqref{eq:wrongtime}, we define
\begin{align}
q \triangleq \bbP[\imath_t (X_{[s+1: t]} ; Y_k) > - \infty],
\end{align}
and treat the cases $q = 1$ and $q < 1$ separately. Observe that for $q = 1$,  Chebyshev's inequality implies
\begin{align}
	&\bbP\left[\imath_t(X_{[s+1:t]}^{n_t};Y_k^{n_t})>n_t\bbE [\imath_t(X_{[s+1:t]};Y_k) ]
		+\lambda_{t,s}^k  \right] \notag \\
	&\quad \leq \frac{\mathrm{Var} \left[ \imath_t(X_{[s+1:t]};Y_k) \right] }	
			{n_t \left(\frac12(I_t(X_{[s]};Y_t|X_{[s+1:t]})
			-\frac {s}{t} I_t)\right)^2 } \label{bound1},
\end{align}
	which is of order $O \left(\frac 1 {n_t}\right)$ by the moment assumption \eqref{eq:moment2a} and Lemma~\ref{lem:mutualinfo}.

For $q < 1$, 
\begin{align}
\bbP\left[\imath_t(X_{[s+1:t]}^{n_t};Y_k^{n_t})>n_t\bbE [\imath_t(X_{[s+1:t]};Y_k) ]+\lambda_{t,s}^k  \right] \leq q^{n_t}.
\end{align}
Therefore \eqref{eq:confususer1} contributes $\bigo{\frac1{n_k}}$ to 
our error bound.
 
 \item{\eqref{eq:confususer}:}
First, consider the case where $s < t \leq k$.  By \lemref{lemma:expectation} and Chernoff's bound,  
\begin{align}
&\mathbb{P} [\imath_t( \bar{X}_{[s]}^{n_t} ; Y_k^{n_t} | X_{[s+1: t]}^{n_t}) \nonumber \\
&\quad > \log \gamma_t - n_t \bbE [\imath_t(X_{[s+1:t]} ; Y_k)] - \lambda_{s, t}^k ] \label{eq:Pitgreater}\\
&\leq \mathbb{P} [\imath_t( \bar{X}_{[s]}^{n_t} ; Y_k^{n_t} | X_{[s+1: t]}^{n_t}) \nonumber \\
&\quad > \log \gamma_t - n_t I_t(X_{[s+1:t]}; Y_t) - \lambda_{s, t}^k ] \\
&\leq \E{\exp \left \lbrace \imath_t \left( \bar{X}_{[s]}^{n_t} ; Y_k^{n_t} | X_{[s+1: t]}^{n_t} \right) \right \rbrace} \nonumber \\
&\quad \cdot \exp{ \lbrace -(\log \gamma_t - n_t I_t(X_{[s+1:t]}; Y_t) - \lambda_{s, t}^k )\rbrace }  \\
&= \exp{ \lbrace - (\log \gamma_t - n_t I_t(X_{[s+1:t]}; Y_t) - \lambda_{s, t}^k )\rbrace }.  \label{eq:boundtechnique}
\end{align}
Using Stirling's bound
\begin{align}
 \binom{n}{k} \leq \left( \frac{e n}{k} \right)^k, \label{eq:bino}
\end{align}
we find that for all $s \leq t \leq k$
\begin{align}
\log \binom{M-k}{s} &\leq s \log \left( \frac{e (M-k)}{s} \right) \label{eq:logMks} \\
&\leq s \log \left( \frac{e (M-t)}{t} \right) + s \log \left( \frac{t}{s} \right) \\
&= \frac{s}{t} \left( \log \gamma_t - \frac 1 2 \log n_t \right) + s \log \left( \frac{t}{s} \right), \label{eq:logMkt}
\end{align}
where \eqref{eq:logMkt} follows from \eqref{eq:choicenk}. From \eqref{setgamma}, \eqref{eq:choiceDelta}, \eqref{eq:boundtechnique}, and \eqref{eq:logMkt}, we have
\begin{align} 
&\binom{M-k}{s} \mathbb{P} [\imath_t( \bar{X}_{[s]}^{n_t} ; Y_k^{n_t} | X_{[s+1: t]}^{n_t}) \notag \\ 
&\quad \quad > \log \gamma_t - n_t I_t(X_{[s+1:t]}; Y_t) - \lambda_{s, t}^k ] \\
 &\quad \leq \exp  \bigg \lbrace - n_t \frac{1}{2} \left(I_t(X_{[s]}; Y_t | X_{[s+1: t]}) - \frac{s}{t} I_t \right) \nonumber \\
 &\quad \quad + \left( 1 - \frac{s}{t} \right) \tau_t \sqrt{n_t V_t} - \frac{s}{2t} \log n_t + s \log \left(\frac t s\right) \!\! \bigg \rbrace. \!\!\!  \label{eq:expnt12}
\end{align}
\lemref{lem:mutualinfo} ensures that the exponent in \eqref{eq:expnt12} is negative for $n_t$ large enough. 

For $s = t < k$, from \eqref{eq:boundtechnique} and \eqref{eq:logMkt} with $s = t$, we get
\begin{align} 
\binom{M-k}{t} \mathbb{P} [\imath_t( \bar{X}_{[t]}^{n_t} ; Y_k^{n_t}) > \log \gamma_t] &\leq \frac{\binom{M-k}{t}}{\gamma_t} \leq \frac 1 {\sqrt{n_t}}.\label{bound5}
\end{align}

For $s = t = k$, following the change of measure technique  (e.g., \cite[Prop.~17.1]{polyanskiyLectureNotes}), one can rewrite an expectation with respect to measure $Q$ as an expectation with respect to measure $P$, giving 
\begin{align}
Q \left[ Z \in \mathcal A \right]  = \mathbb E_P\left[ \left( \frac{P[Z]}{Q[Z]}\right)^{-1} 1 \left\lbrace Z \in \mathcal A \right \rbrace \right]. 
\end{align}
Switching to the measure $P_{X_{[k]}} P_{Y_k | X_{[k]}}$ in this way, by \eqref{eq:bino} and the parameter choice \eqref{eq:choicenk}, we write
	\begin{IEEEeqnarray}{rCl}
	\lefteqn{ \binom{M-k}{k} \bbP[\imath_k( \bar{X}_{[k]}^{n_k} ; Y_k^{n_k} )  > \log \gamma_k  ] \nonumber}\\
	 &\leq & \left(\frac{e}{k}(M-k)\right)^k \label{eq:temp} \bbE \left[ \exp\{-\imath_k( X_{[k]}^{n_k};Y_k^{n_k})\} \right. \\
	 &&  \cdot \left.1\{\imath_k(X_{[k]}^{n_k};Y_k^{n_k})  >  \log\gamma_k \} \right ] \notag \\
	 & \leq & \frac{D_k}{n_k} \label{eq:measure1},
	\end{IEEEeqnarray}
	where
\begin{align}
D_{k} \triangleq 2 \left( \frac{\log 2 }{\sqrt{2 \pi V_k }} + 2 B_k \right) \label{eq:Dk}
\end{align}
and $B_k$ is defined in \eqref{eq:Bk}.
To justify \eqref{eq:measure1}, notice that $\imath_k(X_{[k]}^{n_k};Y_k^{n_k})$ is a sum of i.i.d. random variables; in \cite[Lemma 47]{polyanskiy2010Channel}, Polyanskiy \textit{et al.} derive a sharp bound on the expectation
\begin{align}
    \E{\exp\left( -\sum_{i = 1}^n Z_i\right) 1\left\{ \sum_{i = 1}^n Z_i  > \gamma \right\}}
\end{align} 
when the $Z_i$'s are independent. Applying that bound with $Z_i = \imath_k(X_{[k],i}; Y_{k, i})$ yields \eqref{eq:measure1}. Note that $D_{k}$ is finite by the moment assumptions \eqref{eq:moment2} and \eqref{eq:moment3}. 
Combining the bounds for the three cases in \eqref{eq:expnt12}, \eqref{bound5}, and \eqref{eq:measure1}, we conclude that \eqref{eq:confususer} contributes $\bigo{\frac 1 {\sqrt n_k}}$ to the total error.

Finally, we set the constant $C_k$ in \eqref{eq:tauk} to ensure
\begin{align}
\eqref{eq:0transmitter} + \eqref{eq:repetition} + \eqref{eq:wrongtime}+\eqref{eq:confususer1}+\eqref{eq:confususer} \leq \frac{C_k}{\sqrt{n_k}}. \label{eq:Ckchoice}
\end{align}
The existence of such a constant is guaranteed by our analysis above demonstrating that the terms \eqref{eq:0transmitter}--\eqref{eq:confususer} do not contribute more than $\bigo{\frac 1 {\sqrt n_k}}$ to the total.\footnote{Our bounds on \eqref{eq:0transmitter}--\eqref{eq:confususer} technically depend on $\gamma_k$ and therefore on $C_k$. However, it is easy to see that their dependence on $C_k$ is weak, and for large enough $n_k$, it can be eliminated entirely.   Thus the choice of $C_k$ satisfying \eqref{eq:Ckchoice} is possible.}
\end{itemize}

Due to \eqref{eq:outageprob} and \eqref{eq:Ckchoice}, the total probability of making an error at time $n_k$ is bounded by $\epsilon_k$, and the proof of Theorem~\ref{thm:ach} is complete.

\section{Discussion of the Main Result} \label{sec:furtherDis}

\subsection{Refining the Third-Order Term Using a Maximum Likelihood Decoder} \label{sec:refining}

For a RAC that satisfies the conditions in \thmref{thm:ach} and the conditional variance condition
\begin{align}
    \E{\Var{\imath_k(X_{[k]}; Y_k)| Y_k}} > 0 \quad \forall s \in [k], \label{eq:conditionalvariance}
\end{align}
we can improve the achievable third-order performance in \eqref{eq:mainresult} from $-\frac 1 2 \log n_k$ to $+\frac 1 2 \log n_k$. Prior work showing the achievability of the $+ \frac 1 2 \log n$ third-order term includes \cite[Th.~53]{polyanskiy2010thesis} for point-to-point channels satisfying \eqref{eq:conditionalvariance} with $k = 1$, \cite[Th.~1]{tan2015Third} for the Gaussian point-to-point channel, \cite[Th.~7]{liu2020Finiteblocklength}, \cite[Th.~14]{liu2020FiniteblocklengthArxiv} for discrete memoryless MACs satisfying \eqref{eq:conditionalvariance}, and \cite[Th.~2~and~4]{yavas2020Gaussian}, \cite[Th.~2~and~4]{yavas2020GaussianArxiv} for the Gaussian MAC and RAC. We can achieve the result here by replacing the threshold rule in \eqref{eq:thres} with a combination of a hypothesis test and a maximum likelihood decoder, giving
\begin{align}
\sfg_k(U, y^{n_k}) &=
        \hspace{-0.2em} \begin{cases} 
        \arg \max \limits_{w_{[k]}} \, \imath_k(\sff(w_{[k]})^{n_k};y^{n_k}) \hspace{-0.2em}&\text{if } h_k(y^{n_k}) \leq \gamma_k \\
        \sfe \hspace{-0.2em} &\text{otherwise,}
        \end{cases}
        \label{eq:thresML}
\end{align}
where the maximum is over the ordered message vectors $w_{[k]}$, and $h_k(\cdot)$ is a suitable test function that allows us to distinguish $P_{Y_k}$ from any $P_{Y_t}$ with $t \neq k$. As in prior work, the analysis applies the random coding union bound from \cite[Th. 16]{polyanskiy2010Channel}. 
As discussed in \secref{sec:analysis_0test}, suitable test functions $h_k(\cdot)$ can be found provided that $P_{Y_k} \neq P_{Y_t}$ for all $t \neq k$. For instance, in \cite{yavas2020Gaussian}, we use $h_k(y^{n_k}) = \left\lvert \frac 1 {n_k} \norm{ y^{n_k} }^2 - (1 + kP) \right \rvert$ for the Gaussian RAC, where $P$ is the maximal power constraint. The result does not apply to channels such as the adder-erasure RAC \eqref{eq:addererasure}, which does not satisfy the condition in \eqref{eq:conditionalvariance}.

\subsection{Choosing the Input Distribution $P_X$} \label{sec:choosinginput}
Although there are RACs for which a single input distribution $P_X$ achieves the capacity for all $k$-MACs, $k \in [K]$, (e.g., the adder-erasure channel), the permutation-invariance \eqref{eq:permutationinvariance} and reducibility \eqref{eq:reducible} assumptions do not imply that such a distribution exists for all RACs. In the following, we discuss how to choose the input distribution when the optimal input distribution varies with $k$.

Given a permutation-invariant \eqref{eq:permutationinvariance} and reducible \eqref{eq:reducible} RAC, $M$, $\bm{\epsilon} = (\epsilon_0, \dots, \epsilon_K)$, and any $P_X$ such that \eqref{eq:silence}--\eqref{eq:moment2a} are satisfied for the given RAC under input distribution $P_X$, let  
\begin{align}
&\mathcal{R}(M, \bm{\epsilon}, P_X) = \{(R_0, \dots, R_K) \colon \text{\eqref{eq:mainresult} and \eqref{eq:n0logn} hold} \}
\end{align}
denote the achievable rate region under input distribution $P_X$.
Here
\begin{align}
R_k = \frac{\log M}{n_k} \text{ for all } k \in \{0, \dots, K\}.
\end{align}
Let
\begin{align}
\mathcal{R}(M, \bm{\epsilon}) = \bigcup_{P_X: \text{ \eqref{eq:silence}--\eqref{eq:moment2a} hold} } \mathcal{R}(M, \bm{\epsilon}, P_X) \label{eq:rset}
\end{align}
denote the achievable rate region over all i.i.d. input distributions. A point in this set is called \emph{dominant} if no other points in the set are element-wise greater than or equal to that point. To optimize the achievable rate vector over the allowed input distributions, we must choose a distribution $P_{X^*}$ that achieves a dominant point for the set $\mathcal{R}(M, \bm{\epsilon})$. 
Note that for the dominant points of $\mathcal{R}(M, \bm{\epsilon})$ corresponding to different values of $P_{X^*}$, there is an $O(1)$ difference between the left and right sides of the inequalities in \eqref{eq:mainresult}. If the achievable rate region $\mathcal{R}(M, \bm{\epsilon})$ is not convex, it can be improved to its convex hull using time sharing. For the modifications to the coding strategy that enable us to incorporate time sharing, see \cite{huang2012finite, tan2014dispersions, scarlett2015constantcompMAC}. 


To illustrate what happens when different $P_{X^*}$ values achieve different dominant points of $\mathcal{R}(M, \bm{\epsilon})$, we consider the following example.

\begin{example} 
\normalfont Consider a RAC with $K = 2$, $\mathcal{X} = \mathcal{Y}_2 = \{0, 1\}$, and transition probability matrix $P_{Y_2|X_1, X_2}$
\begin{equation} \label{eq:exchannel}
\begin{tabular}{|c|c|c|c|c|}
\hline
$Y_2$ \textbackslash $X_1 X_2$ & $00$ & $01$ & $10$ & $11$\\
\hline
$0$ & $1-b$ & $b$ & $b$ & $1-a$\\
\hline
$1$ & $b$ & $1-b$ & $1-b$ & $a$\\
\hline
\end{tabular} 
\end{equation}
where $a, b \in [0, 1]$. This RAC is permutation-invariant since the ``01" and the ``10" columns are identical. When $k = 1$, the channel reduces to the binary symmetric channel with crossover probability $b$.
\figref{fig:rateTot} illustrates the set of achievable rate vectors $\mathcal{R}(M, \bm{\epsilon})$ (neglecting the $O(1)$ term in \eqref{eq:mainresult}) with $\log M = 1000$ and $\boldsymbol{\epsilon} = 10^{-3} \boldsymbol{1}$ for two choices of parameters in the channel in \eqref{eq:exchannel}. In \figref{fig:rate1}, $a = 0.7, b = 0.11$, and in \figref{fig:rate2}, $a = b = 0.11$; for each, the finite blocklength and capacity boundaries are demonstrated. In \figref{fig:rate1}, the dominant points are highlighted. The input distribution $P_{X^*} = (0.65, 0.35)$ (i.e., the Bernoulli(0.35) distribution) achieves the dominant point $(R_1, R_2) = (0.400, 0.204)$; the corresponding region $\mathcal{R}(M, \bm{\epsilon}, P_{X^*})$ is shown as the region bounded by the dashed lines. In \figref{fig:rate2}, the only dominant point $(0.437, 0.227)$ is achieved by the input distribution $P_{X^*} = (0.5, 0.5)$ (i.e., the Bernoulli(0.5) distribution.) Therefore, for the channel in \figref{fig:rate2}, the achievable rate region $\mathcal{R}(M, \bm{\epsilon})$ coincides with $\mathcal{R}(M, \bm{\epsilon}, P_{X^*})$, and we must choose $P_{X^*}$ as our input distribution. For this channel, $P_{X^*} = (0.5, 0.5)$ simultaneously maximizes the mutual informations $I_1$ and $I_2$, and the maxima are $I_1 = I_2 = 0.5$.
\end{example}

\begin{figure}[!htbp] 
    \centering 
    \subfloat[]{
    	\includegraphics[width=0.45\textwidth]{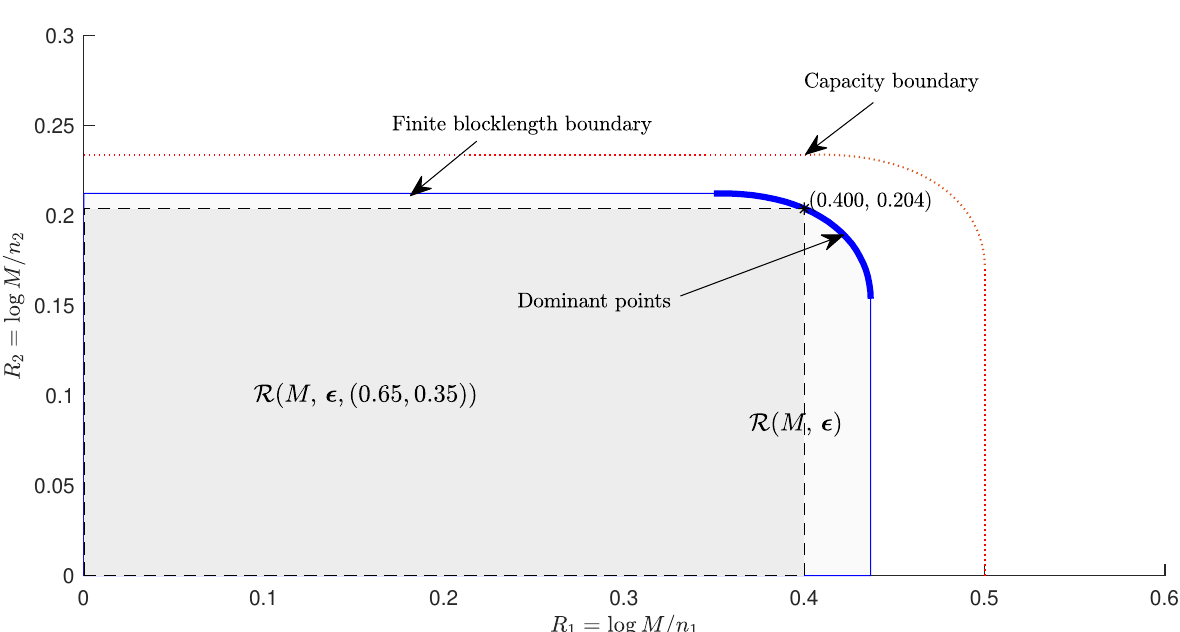} 
    	\label{fig:rate1}} \\
    \subfloat[]{	
    	\includegraphics[width=0.45\textwidth]{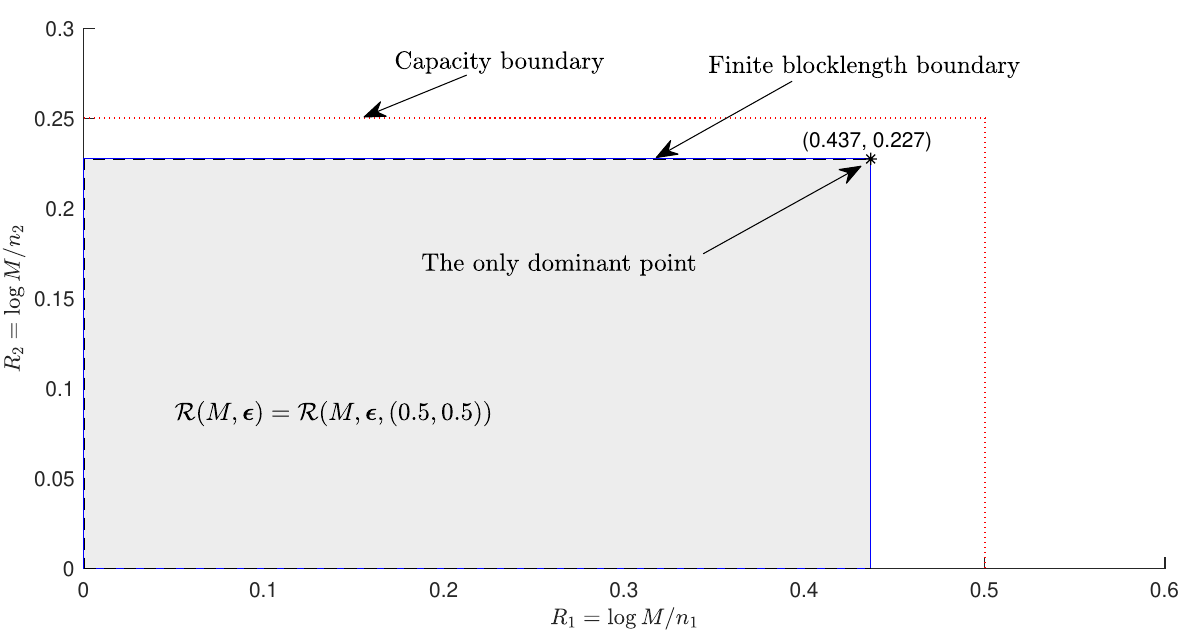}
    	\label{fig:rate2}}
    \caption{\label{fig:rateTot} The achievable rate region from \thmref{thm:ach} (excluding the $O(1)$ term) applied to the channel in \eqref{eq:exchannel} with $\log M = 1000$ and $\epsilon_k = 10^{-3}$ for $k \in [2]$. The results are shown for (a) $a = 0.7$ and $b = 0.11$ and blocklengths $(n_1, n_2) = (2501, 4904)$, and (b) for $a = b = 0.11$ and blocklengths $(n_1, n_2) = (2290, 4399)$.}
\end{figure} 

\subsection{Discussion of the Converse} \label{app:converse}
Even for MACs with only 2 transmitters, the capacity region for the MAC remains incompletely understood. A brief summary of related results follows. For any blocklength $n$ and average error probability $\epsilon \in (0, 1)$, let 
\begin{align}
    \mathcal{R}(n, \epsilon) = \left\{ \left(\frac{\log M_1}{n}, \frac{\log M_2}{n} \right) \colon \exists \text{ an } (n, M_1, M_2, \epsilon) \text{ code} \right\}
\end{align}
denote the set of achievable rate pairs, where $M_i$ is the message size for transmitter $i \in \{1, 2\}$. The capacity region of the MAC \cite{liao1971thesis, ahlswede1971multi} is 
\begin{align}
    \mathcal{C} = \bigcup_{P_Q P_{X_1|Q} P_{X_2|Q}} \{ (R_1, R_2) &\colon \notag \\
    R_1 &\leq I_2(X_1; Y_2|X_2, Q) \notag \\
    R_2 &\leq I_2(X_2; Y_2|X_1, Q)  \notag \\
    R_1 + R_2 &\leq I_2(X_1, X_2; Y_2|Q)\},
\end{align}
where $Q$ is the time sharing random variable. In \cite{dueck1981strong}, Dueck uses the blowing-up lemma to derive the first strong converse for discrete memoryless MACs. In \cite{ahlswede1982elementary}, for discrete memoryless MACs, Ahlswede uses a wringing technique to show 
\begin{align}
    \mathcal{R}(n, \epsilon) \subseteq \mathcal{C} + \bigo{\frac{\log n}{\sqrt{n}}} \boldsymbol{1}, \label{eq:ahlswede}
\end{align}
which improves Dueck's result. The coefficients of the term $\bigo{\frac{\log n}{\sqrt{n}}} \boldsymbol{1}$ in \eqref{eq:ahlswede} are bounded by a multiple of the product of input and output alphabet sizes $|\mathcal{X}_1| |\mathcal{X}_2||\mathcal{Y}_2|$. 
In \cite[Th.~1]{fong2016gaussianstrong}, Fong and Tan improve Ahlswede's second-order term $\bigo{\frac{\log n}{\sqrt{n}}} \boldsymbol{1}$ to $\bigo{\sqrt{\frac{\log n}{{n}}}} \boldsymbol{1}$ for the Gaussian MAC. They derive this result by applying Ahlswede's wringing technique \cite{ahlswede1982elementary} to quantized channel inputs. 
In \cite{kosut2020converse}, Kosut further improves the second-order term to $\bigo{\frac{1}{\sqrt{n}}} \boldsymbol{1}$. The second-order term in \cite[Th.~7]{kosut2020converse} has the same order and, for some channels, the same sign as the best-known second-order achievable term in \cite{scarlett2015constantcompMAC}. Kosut's result
applies to all discrete memoryless MACs and to the Gaussian MAC. To prove this converse, Kosut introduces a new measure of dependence between two random variables called ``wringing dependence." A key aspect of the approach is to restrict the channel inputs so that the wringing dependence between them is small. 

In \cite{moulin2013anewmeta}, Moulin proposes a new converse technique for maximum-error capacity. His approach relies on strong large deviations for binary hypothesis tests and leads to a second-order term as in \eqref{eq:mainresult} when no time sharing is needed. Since the capacity regions for the maximum and average error probability can differ \cite{dueck1978maximal}, Moulin's result does not give a converse for the average-error capacity. 
Whether it is possible to derive a converse for the average-error capacity with a second-order term matching the ones in \cite{huang2012finite, jazi2012simpler, scarlett2015constantcompMAC, molavianjazi2015second, tan2014dispersions} remains an open problem.

In the sparse recovery literature, where achievability proofs typically consider the expected error probability evaluated under i.i.d. codebook design (see, e.g.,  \cite{malyutov1978, malyutov1980planning, scarlettPhase2016, atia2012, scarlett2017limits}), converses derive lower bounds on the expected error probability \textit{assuming i.i.d. code design}. Although a lower bound on the expected error probability for our problem could be derived using tools from \cite{scarlettPhase2016}, such a bound would yield a bound for the best i.i.d. random code rather than a bound for all possible codes.

\subsection{A RAC Code That Decodes Transmitter Identity}  
\label{sec:transiden}
While the use of identical encoding at all transmitters has a number of practical advantages, the techniques employed in this work are not limited to that case. 

We next briefly explore the use of distinct encoders at each transmitter of a RAC. Under permutation-invariance \eqref{eq:permutationinvariance} and identical encoding, the decoder cannot distinguish which transmitter sent each of the decoded messages. Maintaining permutation-invariance but replacing identical encoders with a different instance of the same random codebook for each encoder, we get a code that achieves the same first- and second-order terms as in \thmref{thm:ach}, with a decoder that can also associate the corresponding transmitter identity to each decoded message. The following definition formalizes the resulting RAC codes.

\begin{definition} \label{def:codeprime}
An $(M, \lbrace (n_k, \epsilon_k) \rbrace_{k = 0}^K )$ identity-preserving code comprises a collection of encoding functions
\begin{align}
\mathsf f_k\colon \, \mathcal{U} \times [M] \to \mathcal X^{n_K}, \quad k = 1, \dots, K,
\end{align}
and a collection of decoding functions
\begin{align}
\mathsf g_k \colon \, \mathcal{U} \times \mathcal Y_k^{n_k} \to  \left\{ [M]^k  \times \binom{[K]}{k} \right\} \cup \{\mathsf
e\}, \,\,k = 0, 1, \ldots, K,
\end{align}
where erasure symbol $\mathsf{e}$ is the decoder's output when the decoder is not ready to decode.
At the start of each epoch, a random variable $U \in \mathcal{U}$, with $U \sim P_U$, is generated independently of the transmitter activity, and revealed to the transmitters and the receiver for use in initializing the encoders and the decoder. If the set of active transmitters $\mathcal{A} \subseteq [K]$ satisfies $|\mathcal{A}| = k > 0$, i.e., $k$ transmitters are active, then the messages of $\mathcal{A}$ and their corresponding transmitter identities are decoded correctly at time $n_k$, with probability at least $1-\epsilon_k$, i.e.,
\begin{align}
\frac{1}{M^k}&\sum_{w_{\mathcal{A}} \in [M]^k} \mathbb{P}\Bigg[\cB{\mathsf
g_k(U, Y_k^{n_k})\neq (w_{\mathcal{A}}, \mathcal{A})}\bigcup \notag\\
& \left.\cB{\bigcup_{t = 0}^{k-1}\cB{\mathsf{g}_t(U, Y_k^{n_t}) \neq \sfe}} \middle | \right. W_{\mathcal{A}} = w_{\mathcal{A}} \Bigg] \leq
\epsilon_k,
\end{align}
where $W_{\mathcal{A}}$ are the independent and equiprobable messages of the transmitters in $\mathcal{A}$, and the given probability is calculated using the conditional distribution $P_{Y_k^{n_k} | X_{\mathcal A}^{n_k}} = P_{Y_k | X_{\mathcal{A}}}^{n_k}$ where $X_i^{n_k} = \mathsf f_i(U, W_i)^{n_k}$, $i \in \mathcal{A} $. If $\mathcal{A} = \emptyset$, then the probability that at time $n_0$ the receiver decodes to the unique message in set $[M]^0 = \{0\}$ is no smaller than $1-\epsilon_0$. That is,
\begin{align}
\Prob{\mathsf{g}_0(U, Y_0^{n_0}) \neq \es | W_{[0]} = \es} \leq
\epsilon_0.
\end{align}
\end{definition} 

If we continue to assume 
permutation-invariance \eqref{eq:permutationinvariance} and to employ
the same input distribution $P_X$ at all encoders, then
the channel output statistics again depend on
the dimension of the channel input
but not on the identity of the active transmitters. In this case, we can apply the proof from the identical-encoding single-threshold-decoding argument in \secref{sec:codedesign} to derive an achievability result for the general case.\footnote{This simple argument was suggested by Dr. Jonathan Scarlett.} In particular, consider a code with $K M$ (rather than $M$) messages. Replacing $M$ by $KM$ in \thmref{thm:ach} implies that our RAC code with identical encoders gives a penalty of $-k \log K$ on the right-hand side of the rate bound \eqref{eq:mainresult}. Suppose that we use this identical-encoding code to design a general code in which codewords indexed from $(t-1)M + 1$ to $t M$ are used exclusively by transmitter $t$ for $t = 1, \dots, K$. Since each message belongs to a single transmitter, the list of decoded messages reveals the identities of the active transmitters. Under this allocation of codewords, the repetition error $\mathbb{P}_{\mathrm{rep}}$ in \eqref{eq:Prep} disappears since transmitters send messages from distinct sets. The error probability from decoding the wrong codeword values decreases since there are fewer legitimate codeword combinations to consider. Therefore, in the case where $K$ is a finite constant and the receiver decodes both messages and transmitter identities, the first three terms in \eqref{eq:mainresult} are preserved, and the penalty $-k \log K$ only affects the constant term $O(1)$ in \eqref{eq:mainresult}.

When applied to a scenario with $M = 1$ and identity decoding, the bound in \thmref{thm:nonasymp}, modified as described in the preceding paragraph, extends the non-asymptotic achievability bound in the group testing problem \cite[Th. 4]{scarlettPhase2016} to the scenario where an unknown number $k$ out of a total of $K$ items are defective. 
In the scenario considered in \cite{scarlettPhase2016}, the number of defective items $k$ is known, and our MAC bound 
\eqref{eq:MACbound} with $K$ replaced by $k$, $M$ replaced by $KM = K$, and the term $\frac{K(K-1)}{2M}$ removed applies.
The resulting bound is similar to \cite[Th. 4]{scarlettPhase2016}. The difference is that the bound in \eqref{eq:MACbound} uses a single information density threshold rule, while \cite[Th. 4]{scarlettPhase2016} uses $2^{k}-1$ simultaneous information density threshold rules.

\subsection{Per-user Probability of Error} \label{sec:PUPE}
We extend the definition of the PUPE from \cite[Def.~1]{polyanskiy2017perspective} to the RAC with $k \in [K]$ active transmitters as
\begin{align}
    e_k \triangleq \frac{1}{M^k} \sum_{w_{[k]} \in [M]^k} \sum_{i = 1}^k \frac{1}{k} \Prob{w_i \notin \mathsf{g}_T(U, Y_k^{n_T}) | W_{[k]} = w_{[k]}}, \label{eq:PUPE}
\end{align}
where $Y_k^{n_T}$ is the received output at time $n_T$, and
\begin{align}
    T \triangleq \min \{t \in \{0\} \cup [K] \colon \mathsf{g}_t(U, Y_k^{n_t}) \neq \mathsf{e} \}
\end{align}
is the random variable describing the decoder's estimate of the number of active transmitters.\footnote{Note that the joint error probability in \eqref{eq:errorcriterion} can likewise be written as
\begin{align*}
    \frac{1}{M^k} \sum_{w_{[k]} \in [M]^k} \Prob{\mathsf{g}_T(U, Y_k^{n_T}) \stackrel{\pi}{\neq} w_{[k]} \middle | W_{[k]} = w_{[k]}}.
\end{align*}}
We set $T = K$ if $\mathsf{g}_t(U, Y_k^{n_t}) = \mathsf{e}$ for all $t \in \{0\} \cup [K]$. For $k = 0$, we define $e_0 \triangleq \Prob{\mathsf{g}_0(U, Y_0^{n_0}) \neq \es | W_{[0]} = \es}$ as in \eqref{eq:probzeroerror}.

For a RAC with a total of $K$ transmitters and a MAC with $K$ transmitters, the following corollary to \thmref{thm:nonasymp} gives non-asymptotic achievability bounds under the PUPE criterion \eqref{eq:PUPE}. 
\begin{corollary}\label{cor:pupe}
Fix constants $\gamma_0$, $\lambda_{s, t}^k \geq 0$, and $\gamma_t > 0$ for all $1 \leq s \leq t \leq k$. For any $k$ and $n$, let $(X_{[k]}^n, \bar{X}_{[k]}^n, Y_{k}^n) $ be a random sequence drawn i.i.d. $\sim P_{X_{[k]} \bar{X}_{[k]} Y_{k}}(x_{[k]}, \bar{x}_{[k]}, y_{k}) = \left( \prod_{i = 1}^k P_X(x_i) P_X(\bar{x}_i) \right) P_{Y_k|X_{[k]}}(y_k | x_{[k]}) $.
\begin{enumerate}[leftmargin=*, label=\Alph*)]
\item For any RAC $\left\{\left(\cX^k, P_{Y_k | X_{[k]}}(y_k|x_{[k]}), \cY_k\right)\right\}_{k=0}^K$ satisfying \eqref{eq:permutationinvariance} and \eqref{eq:reducible}, any $K \leq \infty$, and any fixed input distribution $P_X$, there exists an $(M, \lbrace (n_k, e_k) \rbrace_{k = 0}^K )$ RAC code under the PUPE criterion \eqref{eq:PUPE} such that
\begin{align}
    e_0 \leq \Prob{ h(Y_0^{n_0}) > \gamma_0  },
\end{align}
and for all $k\geq 1$,
		\begin{IEEEeqnarray}{rCl}
        \IEEEyesnumber \label{eq:achPUPE}
        \IEEEyessubnumber*
        e_k &\leq& \bbP[\imath_k(X_{[k]}^{n_k}; Y_k^{n_k}) \leq \log \gamma_k]  \\
        &&+  \Prob{  h(Y_k^{n_0}) \leq \gamma_0  }  + \frac{k(k-1)}{2M} \\
        &&+\sum_{t = 1}^{k - 1} \binom{k-1}{t} \bbP[\imath_t(X_{[t]}^{n_t}; Y_k^{n_t}) > \log \gamma_t ]  \label{eq:PUPEeqfirst} \\
        &&+\sum_{t = 1}^k \sum_{s = 1}^{t-1} \binom{k-1}{t-s}  \mathbb{P}\Big[\imath_t(X_{[s+1:t]}^{n_t}; Y_k^{n_t}) \notag \\
        &&\quad > n_t \mathbb{E} [\imath_t(X_{[s+1:t]}; Y_k) ] + \lambda_{s, t}^k \Big]  \\
        &&+ \sum_{t = 1}^k\sum_{s = 1}^t \binom{k-1}{t-s} \binom{M-k}{s}  \notag \\
        && \quad \mathbb{P} \Big[\imath_t( \bar{X}_{[s]}^{n_t} ; Y_k^{n_t} | X_{[s+1: t]}^{n_t})  \notag \\
        &&\quad  > \log \gamma_t - n_t \mathbb{E} [\imath_t(X_{[s+1:t]}; Y_k)]  - \lambda_{s, t}^k \Big]. \label{eq:PUPEeqlast} \IEEEeqnarraynumspace
		\end{IEEEeqnarray}
\item For a MAC with $K$ transmitters satisfying \eqref{eq:permutationinvariance}, there exists a MAC code for $M$ messages and decoding blocklength $n_K$ such that
\begin{IEEEeqnarray}{rCl}
    e_K &\leq& \bbP[\imath_K(X_{[K]}^{n_K}; Y_K^{n_K}) \leq \log \gamma_K] + \frac{K(K-1)}{2M} \notag \\
        &&\,\,+ \sum_{s = 1}^{K-1} \binom{K-1}{K-s} \mathbb{P}\Big[\imath_K(X_{[s+1:K]}^{n_K}; Y_K^{n_K}) \notag \\
        &&\quad \quad > n_K \mathbb{E} [\imath_K(X_{[s+1:K]}; Y_K) ] + \lambda_{s, K}^K \Big]  \notag  \\
        &&\,\,+ \sum_{s = 1}^K \binom{K-1}{K-s} \binom{M-K}{s}  \notag \\
        &&\quad \mathbb{P} \Big[\imath_K( \bar{X}_{[s]}^{n_K} ; Y_K^{n_K} | X_{[s+1: K]}^{n_K}) > \log \gamma_K \notag \\
        &&\quad \quad - n_K \mathbb{E} [\imath_K(X_{[s+1:K]}; Y_K) ] - \lambda_{s, K}^K \Big]. \label{eq:MACboundPUPE}
\end{IEEEeqnarray}
\end{enumerate}
\end{corollary}
\begin{IEEEproof} 
Notice that in \eqref{eq:achPUPE}, the only modification from \thmref{thm:nonasymp} is the replacement of the coefficients $\binom{k}{t}$ in \eqref{eq:wrongtime} and $\binom{k}{t-s}$ in \eqref{eq:confususer1}--\eqref{eq:confususer} by the coefficients $\binom{k-1}{t}$ and $\binom{k-1}{t-s}$, respectively. To see how \corref{cor:pupe} is derived from \thmref{thm:nonasymp}, observe that the PUPE \eqref{eq:PUPE} measures the fraction of transmitted messages missing from the list of decoded messages. Therefore, to bound the PUPE for the RAC, we can multiply the error probability bounds in \eqref{eq:nonasymptotic} that correspond to the case where $t$ out of $k$ messages are decoded by $\frac{k-(t-s)}{k}$, where $s$ is the number of messages decoded incorrectly. 

Similarly, under the PUPE, the coefficient $\binom{K}{K-s}$ in the $K$-transmitter MAC bound \eqref{eq:MACbound} is replaced by $\binom{K-1}{K-s}$ in \eqref{eq:MACboundPUPE} since we can multiply the error probability bounds in \eqref{eq:MACbound1}--\eqref{eq:MACbound2}, corresponding to the case where $s$ out of $K$ messages are decoded incorrectly, by $\frac{s}{K}$.
\end{IEEEproof}

From the proof of \thmref{thm:ach}, the error probability bounds in \eqref{eq:PUPEeqfirst}--\eqref{eq:PUPEeqlast} behave as $O\left(\frac{1}{\sqrt{n_k}}\right)$. This implies that under the PUPE criterion \eqref{eq:PUPE}, our encoding and decoding scheme described in \secref{sec:codedesign} achieves the same first three order terms as \thmref{thm:ach}. Only the constant $O(1)$ term in \eqref{eq:mainresult} is affected by the change from the joint error probability to the PUPE.



The PUPE criterion becomes critical in applications of the Gaussian RAC with $K \to \infty$, where the energy per bit ($\frac{nP}{2 \log_2 M}$) and the number of bits sent by each transmitter ($\log_2 M$) are fixed as the blocklength $n$ grows, and all $K$ transmitters are active. In \cite{polyanskiy2017perspective}, Polyanskiy shows that in this regime, the joint error probability goes to 1 as $K \to \infty$. As we saw in \eqref{eq:MACboundPUPE}, the
PUPE introduces scaling factors $\frac{s}{K}$ in front of the error terms corresponding to $s$ out of $K$ messages decoded incorrectly, for $s = 1, \dots, K$. In the regime $K \to \infty$, the number of these terms is infinite, and the PUPE can be strictly less than $1$ even as the joint error probability approaches 1. In \cite{polyanskiy2017perspective}, Polyanskiy shows that the PUPE behaves nontrivially in this regime.

\section{Tests for No Active Transmitters} \label{sec:analysis_0test}
In this section, we give an analysis of the error probabilities of the composite binary hypothesis test that we use to decide between $H_0$: ``no active transmitters," and $H_1$: ``$k \in [K]$ active transmitters;" that is 
\begin{align}
&H_0: Y^{n_0} \sim P_{Y_0}^{n_0}  \notag \\
&H_1: Y^{n_0} \sim P_{Y_k}^{n_0} \mbox{ for some } 1 \leq k \leq K. \label{eq:originalht}
\end{align}
In the context of \thmref{thm:nonasymp}, the maximal number of transmitters, $K$, can be infinite. In that case, enumerating all alternative possibilities as in \eqref{eq:originalht} becomes infeasible, and a universal (goodness-of-fit) test
\begin{align}
&H_0\colon Y^{n} \sim P_{Y_0}^{n} \notag \\
&H_1\colon Y^{n} \nsim P_{Y_0}^{n}
\end{align}
is appropriate.

 Following \cite{zeitouni1991}\label{def:dr}, a \emph{test statistic} $h_n \colon \mathcal{Y}^n \mapsto \mathbb{R}$ is a function that maps the observed sequence $y^n$ to a real number used to measure the correspondence between that sequence and the null hypothesis. A (randomized) test corresponding to the test statistic $h_n$ is a binary random variable that depends only on $h_n(Y^n)$. The test is deterministic if it outputs $H_0$ if $h_n(y^n) \leq \gamma_0$ for some constant $\gamma_0$, and $H_1$ otherwise.
 
Type-I and type-II errors corresponding to a deterministic test with the statistic $h_n$ are defined as
\begin{align}
\alpha(h_n) &\triangleq P_{Y_0}[h_n(Y^n) > \gamma_0] \\
\beta(h_n) &\triangleq Q[h_n(Y^n) \leq \gamma_0], 
\end{align}
where $Q$ is the unknown alternative distribution of $Y$, and $\gamma_0$ is a constant determined by the desired error criterion. Throughout the following discussion and in our application of these results in \lemref{lem:test}, we employ deterministic tests. For these deterministic tests, we choose $\gamma_0$ to ensure that we meet the zero-transmitter error bound $\alpha(h_n) \leq \epsilon_0$, and then we show that $\beta(h_n)$ decays exponentially with $n$ for each $Q$ in $\{P_{Y_1}, \dots, P_{Y_K}\}$ to ensure \eqref{eq:n0logn} in \thmref{thm:ach}. 

In Sections A and B, below, we consider Hoeffding's test and the Kolmogorov-Smirnov test as possible hypothesis tests for recognizing the zero-transmitter scenario. Both tests are universal in the sense that the test statistic does not vary with the alternative output distributions $P_{Y_1}, \dots, P_{Y_K}$. They both give an exponentially decaying type-II error for a fixed type-I error $\epsilon_0 \in (0, 1)$. The disadvantage of Hoeffding's test is that its traditional form requires the channel output alphabet to be finite for every $k$ (as in the adder-erasure RAC in \eqref{eq:addererasure}); the advantage of Hoeffding's test is that it achieves the same exponent as the Neyman-Pearson Lemma, which is optimal for a given collection of output distributions $P_{Y_1}, \dots, P_{Y_K}$, but is not universal, meaning that a different test statistic is necessary for each collection $\left\{P_{Y_k} \colon k \in [K]\right\}$. In contrast to Hoeffding's test, the Kolmogorov-Smirnov test does not require $\mathcal{Y}$ to be finite; however, when applied to a setting with finite $\mathcal{Y}$, it achieves a type-II error exponent that is inferior to that achieved by Hoeffding's test.
In Section \ref{sec:LRTT}, we compare the performances of these universal test statistics to that of the log-likelihood ratio (LLR) threshold test, which is third-order optimal in terms of the type-II error exponent for composite hypothesis testing \cite{huang2014strong} and relies explicitly on alternative output distributions $P_{Y_1}, \dots, P_{Y_K}$.

\subsection{Hoeffding's Test}\label{sec:htest}
Denote the empirical distribution of an observed sequence $y_1, \dots, y_n$ by
\begin{align}
\hat{P}_{y^n}(a) \triangleq \frac 1 n \sum_{i = 1}^n 1\{y_i = a\}  \quad \forall \, a \in \mathcal{Y}.
\end{align}

Hoeffding's test is based on the relative entropy, denoted by $D(\cdot \| \cdot)$, between $\hat{P}_{y^n}$ and $P_{Y_0}$, giving the test statistic
\begin{align}
h_n^{H}(y^n) = D(\hat{P}_{y^n} \| P_{Y_0}). \label{eq:hoefftest}
\end{align}
Note that if $P_{Y_0}$ is a continuous distribution, $h_n^H(y^n) = +\infty$.

\begin{theorem}[Hoeffding's test\cite{hoeffding1965}]
	Let $\mathcal{Y}$ be a finite set, and let $Q$ be an unknown alternative distribution for $Y_0$. If $P_{Y_0}$ is absolutely continuous with respect to $Q$, and $P_{Y_0} \neq Q$, then the type-I and type-II errors of Hoeffding's test satisfy
	\begin{align}
	\alpha(h_n^H) &\leq \exp\{ -n \gamma_0 + O(\log n)\}  \\
	\beta(h_n^H) &\leq \exp\left\{ - n \inf_{P: D(P \| P_{Y_0}) < \gamma_0} D(P \| Q) + O(\log n) \right\}.
	\end{align}
\end{theorem}
In \cite{hoeffding1965}, a more restrictive assumption ($P_{Y_0}(y) > 0$ and $Q(y) > 0$ for all $y \in \mathcal{Y}$) is used. Absolute continuity is sufficient according to the proofs given in \cite{zeitouni1991} and \cite[Th.~2.3]{csiszar2004}, which both rely on Sanov's theorem.
The error exponents of Hoeffding's test coincide with the exponents of the optimal (Neyman-Pearson Lemma) binary hypothesis test. Therefore, Hoeffding's test is asymptotically universally most powerful.

Setting $\gamma_0 = \frac{|\mathcal{Y}| \log n}{n}$ achieves type-I error $\epsilon_0 \to 0$ as $n \to \infty$; therefore, the type-I error condition is satisfied for any $\epsilon_0 > 0$ and sufficiently large $n$. Under this choice, type-II error $\exp\{ -n D(P_{Y_0} \| Q) + o(n) \}$ is achieved (see \cite[Th.~2.3]{csiszar2004}). Therefore, in \eqref{eq:Cprime}, the maximum type-II error decays with exponent 
\begin{align}
C' &= \inf_{k \in [K]} D(P_{Y_0} \| P_{Y_k}) \label{eq:Hmaxbk} \\
& \geq 2 \inf_{k \in [K]} \Bigg\{\left( \sup_{x \in \mathbb{R}}  | F_k(x) - F_0(x) | \right)^2 \notag \\
&\quad + \frac 4 9 \left( \sup_{x \in \mathbb{R}}  | F_k(x) - F_0(x) | \right)^4 \Bigg\} \label{eq:gibbsused}\\
&\geq 2 \delta_0^2 + \frac 4 9 \delta_0^4. \label{eq:DKSrelation}
\end{align}
 The inequality in \eqref{eq:gibbsused} is due to \cite[eq. (5)-(6)]{gibbs2002onchoosing} and Pinsker's inequality \cite{kullback1967lower}. The inequality in \eqref{eq:DKSrelation} follows from \eqref{eq:assump:Pyk}.

In \cite{zeitouni1991}, Zeitouni and Gutman extend Hoeffding's test to continuous distributions. Their test, which also uses the empirical distribution, employs ``$\delta$-smoothing" of the decision regions obtained by a relative entropy comparison. The Zeitouni-Gutman test is optimal under a slightly weaker optimality criterion than the standard first-order type-II error exponent criterion. Using \cite[Th.~2]{zeitouni1991}, it can be shown that the Zeitouni-Gutman test also yields the desired exponentially decaying maximum type-II error.

\subsection{Kolmogorov-Smirnov Test}\label{sec:kstest}

The Kolmogorov-Smirnov test \cite{kolmogorov1933, smirnov1944} relies on the empirical CDF 
\begin{align}
\hat{F}^{(n)}(x|y^n) \triangleq \frac{1}{n} \sum_{i = 1}^n 1\{y_i \leq x\} \quad \forall\, x\in \mathbb{R} \label{eq:Fhatn}
\end{align}
of the observed sequence $y_1, \dots, y_n \in \mathbb{R}$. The Kolmogorov-Smirnov test uses a deterministic test
\begin{align}
h^{KS}_n(y^n) =  \sup_{x \in \mathbb{R}} | \hat{F}^{(n)}(x|y^n) - F_0(x) | \label{eq:ks_test}
\end{align}
to test whether the observed sequence $y^n$ is well-explained by $P_{Y_0}$ with the CDF $F_0$.

The following theorem bounds the probability that the Kolmogorov-Smirnov statistic exceeds a threshold $\gamma_0$.
\begin{theorem}[Dvoretzky-Kiefer-Wolfowitz \cite{dvoretzky1956, massart1990}] \label{thm:DKW}
Let $Y_1, \dots, Y_n$ be drawn i.i.d. according to an arbitrary distribution $P_{Y_0}$ with the CDF $F_0$ on $\mathbb{R}$. For any $n \in \mathbb{N}$ and $\gamma_0 > 0$, it holds that
\begin{align}
\alpha(h^{KS}_n) \leq 2 \exp\{-2 n \gamma_0^2\}. \label{eq:DKW}
\end{align}
\end{theorem}
In \cite{dvoretzky1956}, Dvoretzky \textit{et al.} prove \thmref{thm:DKW} with an unspecified multiplicative constant $C$ in front of the exponential on the right side of \eqref{eq:DKW}. In \cite{massart1990}, Massart establishes that $C = 2$. 

In our operational regime of interest, we set the type-I error to a given constant $\epsilon_0$, which by \thmref{thm:DKW} corresponds to setting the threshold $\gamma_0$ to
\begin{align}
\gamma_0 = \sqrt{\frac{\log \frac{2}{\epsilon_0}}{2 n}} = O\left(\frac 1 {\sqrt{n}}\right). \label{eq:lambdaset}
\end{align}
We next bound the type-II errors for every $k \in [K]$. For each $k \in \{0, \dots, K\}$, let $F_k$ denote the CDF of $P_{Y_k}$. 
The type-II error when $k \geq 1$ transmitters are active is bounded as
\begin{align}
\beta_k(h_n^{KS}) &= \Prob{\sup_{x \in \mathbb{R}} | \hat{F}^{(n)}(x|Y_k^n) - F_0(x) | \leq \gamma_0} \\
&\leq \mathbb{P}\bigg[ \sup_{x \in \mathbb{R}} \Big(| F_k(x) - F_0(x) | \notag \\
&\quad - | \hat{F}^{(n)}(x|Y_k^n) - F_k(x) | \Big) \leq \gamma_0 \bigg] \label{eq:triangleineq}\\
&\leq \mathbb{P}\bigg[ \sup_{x \in \mathbb{R}} | \hat{F}^{(n)}(x|Y_k^n) - F_k(x) | \notag \\
&\quad \geq  \sup_{x \in \mathbb{R}}  | F_k(x) - F_0(x) |  - \gamma_0 \bigg]\\
&\leq 2 \exp\bigg\{-2 n \, \left(\sup_{x \in \mathbb{R}}  | F_k(x) - F_0(x) | \right)^2 \notag \\
&\quad + O(\sqrt{n})\bigg\} \label{eq:dkwused},
\end{align}
where \eqref{eq:triangleineq} follows from triangle inequality $|x + y| \geq |x| - |y|$, and \eqref{eq:dkwused} follows from \thmref{thm:DKW} and \eqref{eq:lambdaset}. Applying \eqref{eq:assump:Pyk} to \eqref{eq:dkwused}, we conclude that the maximum type-II error in \eqref{eq:Cprime} decays exponentially with $n$, with exponent 
\begin{align}
C' &= 2 \inf_{k \in [K]} \left( \sup_{x \in \mathbb{R}}  | F_k(x) - F_0(x) | \right)^2 \label{eq:KSactual} \\
&\geq 2 \delta_0^2. \label{eq:KSmaxbk}
\end{align}
Comparing \eqref{eq:KSactual} and \eqref{eq:gibbsused}, from \eqref{eq:assump:Pyk}, we see that the type-II error exponent achieved by the Kolmogorov-Smirnov test is always inferior to that achieved by Hoeffding's test.

\subsection{The Optimal Composite Hypothesis Test} \label{sec:LRTT}

From \eqref{eq:DKSrelation} and \eqref{eq:KSmaxbk}, we know that there exists a positive constant $c_0$ such that 
\begin{align}
n_0 \geq c_0 \log n_1 + o(\log n_1) \label{eq:n0first3}
\end{align}
suffices to meet the error requirements of the composite hypothesis test given in \eqref{eq:type2error} and \eqref{eq:type1error}. Since the proposed tests are universal, \thmref{thm:nonasymp} allows us to decode any message set of $k \leq K$ active transmitters without knowing the total number of transmitters, $K$. In this section, we find the smallest first three terms on the right side of \eqref{eq:n0first3} that we can achieve when $K$ is finite and we allow the composite hypothesis test to depend on the distributions $P_{Y_1}, \dots, P_{Y_K}$. 

Let $\beta_{\epsilon_0}(P_{Y_0}, \{P_{Y_k}\}_{k = 1}^K)$ denote the minimax type-II error among the alternative distributions $P_{Y_1}, \dots, P_{Y_K}$ such that type-I error (under $P_{Y_0}$) does not exceed $\epsilon_0$; that is,
\begin{align}
\beta_{\epsilon_0}(P_{Y_0}, \{P_{Y_k}\}_{k = 1}^K) \triangleq \min_{h_n: \alpha(h_n) \leq \epsilon_0} \max_{ k \in [K] } \beta_k(h_n), \label{eq:minmaxerror}
\end{align}
where the minimum is over all tests including deterministic and randomized tests.

The LLR test statistic $h_n^{\mathrm{LLR}} \colon \, \mathcal{Y}^n \mapsto \mathbb{R}^K$ is given by
\begin{align}
h_n^{\mathrm{LLR}}(y^n) = \sum_{i = 1}^n h_1^{\mathrm{LLR}}(y_i), \label{eq:LLR}
\end{align}
where
\begin{align}
h_1^{\mathrm{LLR}}(y) \triangleq  \begin{bmatrix} 
\log \frac{P_{Y_0}(y)}{P_{Y_1}(y)} \\
\log \frac{P_{Y_0}(y)}{P_{Y_2}(y)} \\
\vdots \\
\log \frac{P_{Y_0}(y)}{P_{Y_K}(y)} 
\end{bmatrix}.
\label{eq:0test2}
\end{align}
Given a threshold vector $\bm{\tau} \in \mathbb{R}^K$, the corresponding LLR test outputs $H_0$ if $h_n^{\mathrm{LLR}}(y^n) \geq \bm{\tau}$, and $H_1$ otherwise. 

The gap in the type-II error exponent ($C'$ in \eqref{eq:Cprime}) between the general optimal tests and the LLR tests with the optimal threshold vector $\bm{\tau}$ is $O\left(\frac{1}{n} \right)$ \cite{huang2014strong}; therefore, we only consider minimizing over the LLR tests in \eqref{eq:minmaxerror} for asymptotic optimality.

Denote by $\mathbf{D}$ and $\mathsf{V}$ the mean and covariance matrix of the random vector $h_1^{\mathrm{LLR}}(Y_0)$, respectively. Define 
\begin{align}
D_{\min} &\triangleq \min_{k \in [K]} D(P_{Y_0} \| P_{Y_k}) \label{eq:Dmin} \\
\mathcal{I}_{\min} &\triangleq \{k \in [K] \colon D(P_{Y_0} \| P_{Y_k}) = D_{\min}\} \label{eq:Imin}\\
\mathsf{V}_{\min} &\triangleq \textnormal{Cov}\left[\left(h_1^{\mathrm{LLR}}(Y_0)\right)_{\mathcal{I}_{\min}}\right] \in \mathbb{R}^{ |\mathcal{I}_{\min}| \times  |\mathcal{I}_{\min}|}.
\end{align}
The following theorem gives the asymptotics of the minimax type-II error defined in \eqref{eq:minmaxerror}. 
\begin{theorem} \label{thm:huangmoulin}
Assume that $P_{Y_0}$ is absolutely continuous with respect to $P_{Y_k}$, $0  < D(P_{Y_0}\|P_{Y_k}) < \infty $ for $k = 1, \dots, K$, $\mathsf{V}$ is positive definite,  and $T =  \mathbb{E}[ \lVert h_1^{\mathrm{LLR}}(Y_0) - \mathbf{D}\rVert_2^3 ] < \infty$. Then for any $\epsilon_0 \in (0, 1)$, the asymptotic minimax type-II error satisfies
\begin{align}
\beta_{\epsilon_0}(P_{Y_0}, \{P_{Y_k}\}_{k = 1}^K) 
&= \exp\Big\{-nD_{\min} + \sqrt{n} b \notag \\
&\quad -\frac 1 2 \log n + O(1) \Big\}, \label{eq:emin}
\end{align}
where $b$ is the solution to
\begin{align}
\Prob{\mathbf{Z} \leq b\boldsymbol{1}} = 1 - \epsilon_0, \label{eq:bsolution}
\end{align}
for $\mathbf{Z} \sim \mathcal{N}(\boldsymbol{0}, \mathsf{V}_{\min}) \in \mathbb{R}^{| \mathcal{I}_{\min} |}$. Moreover, the minimax error in \eqref{eq:emin} is achieved by a LLR test with some threshold vector $\bm{\tau}$.
\end{theorem}
\begin{IEEEproof}
See \appref{app:hypothesis}.
\end{IEEEproof}

%
%

Rewriting \eqref{eq:emin}, defining $b$ as given in \eqref{eq:bsolution}, and using the condition in \eqref{eq:type2error} with any fixed $E_k$, we see that a decision about whether any of the transmitters are active can be made at time 
\begin{align}
n_0 &= \frac 1 {2 D_{\min}} \log n_1 + \frac {b}{\sqrt{2D_{\min}^3}}\sqrt{\log n_1} \notag \\
&\quad - \frac 1 {2 D_{\min}} \log \log n_1 + O(1) \label{eq:huangn0}
\end{align}
while guaranteeing both that the probability that we do not decode at time $n_0$ when no transmitters are active does not exceed $\epsilon_0$ and that the probability that we decode at time $n_0$ when $k > 0$ transmitters are active does not exceed $\frac{E_k}{\sqrt{n_k}}$. Note that $E_k$ only affects the constant term $O(1)$ in \eqref{eq:huangn0}.
\thmref{thm:huangmoulin} implies that the coefficients in front of $\log n_1$, $\sqrt{\log n_1}$, and $\log \log n_1$ in \eqref{eq:huangn0} are optimal. Juxtaposing \eqref{eq:Hmaxbk} and \eqref{eq:huangn0}, we see that Hoeffding's test achieves the optimal first-order error exponent (that is, the optimal coefficient in front of $\log n_1$).

\section{Conclusion} \label{sec:summary}
We study the agnostic random access model, in which each transmitter knows nothing about the set of active transmitters beyond what it learns from limited scheduled feedback from the receiver, and the receiver knows nothing about the set of active transmitters beyond what it learns from the channel output. In our proposed rateless coding strategy, the decoder attempts to decode only at a fixed, finite collection of decoding times. At each decoding time $n_t$, it sends a single bit of feedback to all transmitters indicating whether or not its estimate for the number of active transmitters is $t$. We prove non-asymptotic and second-order achievability results for the equal rate point $(R, \dots, R)$ under our assumptions on the channel (permutation-invariance \eqref{eq:permutationinvariance}, reducibility \eqref{eq:reducible}, friendliness \eqref{eq:silence}, and interference \eqref{eq:interference}). For a nontrivial class of discrete, memoryless RACs, our proposed RAC code performs as well in its capacity and dispersion terms as the best-known code for the discrete memoryless MAC in operation; that is, it performs as well as if the transmitter set were known {\em{a priori}}. The assumptions of permutation-invariance \eqref{eq:permutationinvariance}, reducibility \eqref{eq:reducible}, and interference \eqref{eq:interference} together with our use of identical encoding guarantee (by \lemref{lem:mutualinfo}) that the equal rate point always lies on the sum-rate boundary rather than on one of the corner points. For example, for two users, the capacity region is a symmetric pentagon. This ensures that our simplified, single-threshold decoding rule results in no loss in the first- or second-order achievable rate terms, making the codes far more practical than prior schemes \cite{huang2012finite, jazi2012simpler, tan2014dispersions, scarlett2015constantcompMAC} in which decoders employ $2^k -1$ simultaneous threshold-rules. In \secref{sec:transiden}, we show that as long as $K < \infty$, there is no loss in the first two terms even if the decoder is tasked with decoding transmitter identity.

We also provide a tight approximation for the capacity and dispersion of the adder-erasure RAC \eqref{eq:addererasure}, which is an example channel satisfying our symmetry conditions. 

In order to decide whether there are any active transmitters without enumerating all $K$ alternative hypotheses, we analyze universal hypothesis tests. Results are given both for the case where the channel output alphabet is finite and the case where the channel output alphabet is countably or uncountably infinite. Using existing literature, it is possible in both cases to obtain exponentially decaying maximum type-II error under the condition that $\sup_{x \in \mathbb{R}} |F_k(x) - F_0(x)| \geq \delta_0 > 0 \, \text{ for all } k \in [K]$. We also derive the best third-order asymptotics of the minimax type-II error (\thmref{thm:huangmoulin}).


We conclude the paper by giving some directions for generalizations and future work.

\begin{itemize}
    \item \textit{Achievability of unequal rate points}: While identical encoding is appealing from a practical perspective, it is also possible to design codes with different transmitters operating at different rates. Such codes would employ non-identical encoding at the transmitters, and they could also employ a decoding rule with multiple, simultaneous threshold rules. In \cite[Section~VI]{chen2019SW}, Chen \textit{et al.} use a similar strategy to derive third-order achievability and converse results for the random access source coding problem where operation at both identical and distinct rates is allowed.
    \item \textit{Unordered decoding times:}
    Example scenarios where unordered decoding times can arise include channels that do not satisfy the assumptions \eqref{eq:silence} or \eqref{eq:interference}, applications characterized by small message sizes (e.g., in the internet of things), and scenarios where the system designer chooses unordered decoding times (e.g., when a quick error is preferable to a long period of low individual data rates caused by unusually high traffic in the network). It is easy to modify our nonasymptotic bound (\thmref{thm:nonasymp}) to capture the case where $n_0, \ldots, n_K$ are unordered.
    \item \textit{Non-i.i.d. input distributions at the random encoders:} Our random coding design generalizes to scenarios where an arbitrary input distribution $P_{X^{n_K}}$ is employed instead of $P_{X} \times \cdots \times P_X$. For example, in \cite[Th.~4]{yavas2020Gaussian}, we improve the achievable second- and third-order term for the Gaussian RAC by employing uniform distributions over spheres instead of i.i.d. distributions. To prove \cite[Th.~4]{yavas2020Gaussian}, random codewords are formed by concatenating $K$ independent sub-codewords,
    drawn from uniform distributions over the power spheres with dimensions $n_1, n_2-n_1, \dots, n_K-n_{K-1}$. This non-product input distribution satisfies the maximal power constraints for all decoding times $n_1, n_2, \dots, n_K$. 
    More broadly, non-stationary input distributions can arise in communication over RACs where no single $P_X$ simultaneously maximizes all mutual informations $I_k$. While we explore in \secref{sec:choosinginput} how to choose the ``best" single-letter input distribution $P_X$ for this scenario, it is possible to employ different input distributions for each of the sub-codewords $n_1, n_2-n_1, \dots, n_K - n_{K-1}$ to achieve higher rates. 

    \item \textit{Fading channels}:  A rateless code design for quasi-static fading RACs where the channel fading coefficients are unavailable either at the transmitters or at the receiver would constitute one of the most practically relevant extensions of this work. In the quasi-static fading channel model with a fixed blocklength, the achievable rate is dictated by a quantity called the \textit{outage probability}\cite{wang2014quasi}. If the fading coefficient is small in a communication epoch, then the channel is declared to be in \textit{outage} and reliable communication is not achieved. However, using rateless codes, it is possible to maintain reliable communication at the expense of reduced rates (i.e., larger decoding times) when the fading coefficient is small while achieving larger rates when the fading coefficient is large. While Kowshik \textit{et al.} \cite{kowshik2020RAC} derive achievability results for the quasi-static fading RAC in the fixed blocklength regime under the PUPE \eqref{eq:PUPE} criterion, \textit{rateless coding} over fading RACs is yet to be fully explored.
\end{itemize}

\appendices
\section{Proofs of Lemmas \ref{lem:mutualinfofriendly}--\ref{lemma:expectation}} \label{app:proofs}
\renewcommand{\theequation}{\thesection.\arabic{equation}}
\setcounter{equation}{0}
We first state and prove Lemma~\ref{lemma:mutualincrease}, which we then use to prove Lemmas~\ref{lem:mutualinfo}, \ref{lem:mutualinfofriendly}, and \ref{lemma:expectation} (in that order).
\begin{lemma} \label{lemma:mutualincrease}
Let $X_1, X_2, \ldots, X_k$ be i.i.d., and let the interference \eqref{eq:interference}, permutation-invariance \eqref{eq:permutationinvariance}, and reducibility \eqref{eq:reducible} assumptions hold. Then $I_k(X_i ; Y_k | X_{[i-1]})$ is strictly increasing in $i$, i.e., for all $i < j \leq k$,
\begin{align}
I_k(X_i ; Y_k | X_{[i-1]}) < I_k(X_j ; Y_k | X_{[j - 1]}) \label{eq:mutualincreaseini}.
\end{align}
\end{lemma}
\begin{IEEEproof}[Proof of \lemref{lemma:mutualincrease}]
By permutation-invariance \eqref{eq:interference} and the i.i.d. distribution of $X_1, \dots, X_k$, we have
\begin{align}
I_k(X_i ; Y_k | X_{[i-1]}) = I_k(X_j ; Y_k | X_{[i - 1]}). \label{eq:permutationinvij}
\end{align}
Let $(U, V, T)$ be mutually independent random variables. Then $I(U; V) = I(U; T, V) = 0$. Since $I(U; T, Y ) \leq I(U; T, V, Y)$, the chain rule
implies that
\begin{align}
I(U; Y | T) \leq I(U; Y | T, V). \label{eq:indepchainrulegiven}
\end{align}
Setting $U$ to $X_{j}$, $Y$ to $Y_k$, $T$ to $X_{[i-1]}$, and  $V$ to $X_{[i: j-1]}$ in \eqref{eq:indepchainrulegiven} and then applying \eqref{eq:permutationinvij} gives \eqref{eq:mutualincreaseini} with $<$ replaced by $\leq$. Equality in  \eqref{eq:indepchainrulegiven} is attained if and only if $U$ and $V$ are conditionally independent given $(Y, T)$. As a result, equality in our modified form of \eqref{eq:mutualincreaseini} occurs if and only if $X_j$ and $X_{[i:j-1]}$ are conditionally independent given $(Y_k, X_{[i-1]})$. We proceed to show that this is not possible using a proof by contradiction.

Assume that $X_j$ and $X_{[i:j-1]}$ are conditionally independent given $(Y_k,  X_{[i -1]})$, i.e.,
\begin{align}
P_{X_{[i:j]} | Y_k, X_{[i-1]}} = P_{X_{[i:j-1]} | Y_k, X_{[i - 1]}} \, P_{X_j | Y_k, X_{[i - 1]}}   \label{eq:condindepassum}.
\end{align}
Set $X_{[i-1]} = 0^{i-1}$ and use Bayes' rule to show 
\begin{align}
P_{X_{[i:j]} | Y_k, X_{[i-1]} = 0^{i-1}} &= P_{X_{[j - (i-1)]} | Y_{k - (i-1)}}  \\
P_{X_{[i:j-1]} | Y_k, X_{[i-1]} = 0^{i-1}} &= P_{X_{[2:j-(i-1)]} | Y_{k - (i-1)}}  \\
P_{X_j | Y_k, X_{[i-1]} = 0^{i-1}} &= P_{X_{1} | Y_{k - (i-1)}}  
\end{align}
due to reducibility \eqref{eq:permutationinvariance}, permutation-invariance \eqref{eq:reducible}, and the i.i.d. distribution of $X_1, \dots, X_k$. 
Therefore, \eqref{eq:condindepassum} implies that $X_1$ and $X_{[2:j-(i-1)]}$ are conditionally independent given $Y_{k-(i-1)}$, which is not possible by interference assumption \eqref{eq:interference}. 
\end{IEEEproof}



\begin{IEEEproof}[Proof of \lemref{lem:mutualinfo}]
We wish to show that
\begin{align}
\frac1kI_k(X_{[k]};Y_k)<\frac1s I_k(X_{[s]};Y_k|X_{[s+1:k]}) \label{eq:lemma2states}.
\end{align}
 By the chain rule for mutual information, the left-hand side of \eqref{eq:lemma2states} equals the average of $k$ terms
\begin{align}
\frac1kI_k(X_{[k]};Y_k)=\frac1k\sum_{i=1}^kI_k(X_i;Y_k|X_{[i-1]}).
\end{align}
By permutation-invariance \eqref{eq:permutationinvariance} and the chain rule, the right-hand side of \eqref{eq:lemma2states} equals the average of the last $s$ of those $k$ terms
\begin{align}
\frac{1}{s} I_k(X_{[s]};Y_k|X_{[s+1:k]})
&= \frac{1}{s} I_k(X_{[k-s+1:k]};Y_k|X_{[k-s]}) \\
&= \frac{1}{s} \sum_{i=k-s+1}^k I_k(X_i;Y_k|X_{[i-1]}). 
\end{align}
Since the terms in these averages are strictly increasing in $i$ by \lemref{lemma:mutualincrease}, we have the desired result.
\end{IEEEproof}

\begin{IEEEproof}[Proof of \lemref{lem:mutualinfofriendly}]
We wish to show that $\frac 1 s I_s > \frac 1 k I_k$. We proceed by representing $I_s$ in terms of $I_k$ as
\begin{align}
\frac 1 s I_s &= \frac 1 s I_k(X_{[s]}; Y_k | X_{[s+1:k]} = 0^{k-s}) \label{eq:redused}\\
&\geq \frac 1 s I_k(X_{[s]}; Y_k | X_{[s+1:k]}) \label{eq:frused}\\\
&> \frac{1}{k} I_k, \label{eq:lemma2used}
\end{align}
where \eqref{eq:redused} follows from reducibility \eqref{eq:reducible}, \eqref{eq:frused} follows from friendliness \eqref{eq:silence}, and (\ref{eq:lemma2used}) follows from \lemref{lem:mutualinfo}.
\end{IEEEproof}

\begin{IEEEproof}[Proof of \lemref{lemma:expectation}]
To derive the bound $\mathbb{E}[\imath_t(X_{[s]}; Y_k)] \leq I_k(X_{[s]}; Y_k) < I_t(X_{[s]}; Y_t)$, we write
\begin{IEEEeqnarray}{rCl}
\mathbb{E}[\imath_t(X_{[s]}; Y_k)] 
&=& \bbE \left[{\log \frac{P_{Y_t | X_{[s]} }(Y_k | X_{[s]}) }{P_{Y_t}(Y_k)}} \right]\\
&=& -D(P_{X_{[s]}} P_{Y_k|X_{[s]}} \|  P_{X_{[s]}} P_{Y_t|X_{[s]}})  \nonumber \\
&& + D(P_{Y_k} \| P_{Y_t})  \nonumber \\
&& + D(P_{X_{[s]}} P_{Y_k|X_{[s]}} \|  P_{X_{[s]}} P_{Y_k} ) \\
&=& -D( P_{X_{[s]}} P_{Y_k|X_{[s]}} \|  P_{X_{[s]}} P_{Y_t|X_{[s]}} )  \nonumber \\
&&+ D(P_{Y_k} \| P_{Y_t})  + I_k(X_{[s]}; Y_k)  \\
&\leq& I_k(X_{[s]}; Y_k) \label{ineq51}\\
&=& \sum_{i = 1}^s I_k(X_i; Y_k | X_{[i-1]}) \label{eq:chainruleIk} \\
&<& \sum_{i = 1}^s I_k(X_i; Y_k | X_{[i-1]}, X_{[s+1:s+k-t]}) \IEEEeqnarraynumspace \label{eq:chainruleIk2} \\
&=& I_k(X_{[s]}; Y_k | X_{[t+1:k]}) \label{eq:expperm}\\
&\leq& I_k(X_{[s]}; Y_k | X_{[t+1:k]} = 0^{k-t}) \label{eq:reduceused}\\
&=& I_t(X_{[s]}; Y_t), \label{ineq52}  
\end{IEEEeqnarray}
where \eqref{ineq51} follows from data processing inequality of relative entropy (e.g., \cite[Th.~2.2.5]{polyanskiyLectureNotes}), \eqref{eq:chainruleIk} follows from the chain rule, \eqref{eq:chainruleIk2} follows from permutation-invariance \eqref{eq:permutationinvariance} and  \lemref{lemma:mutualincrease}, \eqref{eq:expperm} follows from permutation-invariance \eqref{eq:permutationinvariance} and the chain rule, and (\ref{eq:reduceused}) and (\ref{ineq52}) follow from friendliness (\ref{eq:silence}) and reducibility (\ref{eq:reducible}), respectively.

\end{IEEEproof}
\ver{}{
\section{Achievable Rates on the Sum-rate Boundary of a MAC}\label{app:sumrateboundary}
\begin{IEEEproof}[Proof of \eqref{eq:L1L2}]
	For the 2-transmitter MAC, the second-order achievability region is expressed in terms of dispersion matrix, $\mathsf{V}_2$, defined as the covariance matrix of the random vector
	\begin{align}
	\boldsymbol{\imath}_2 \triangleq  \begin{bmatrix}
	\imath_2(X_1; Y| X_2) \\ \imath_2(X_2; Y| X_1) \\ \imath_2(X_1, X_2; Y) \end{bmatrix}. 
	\end{align}
	
	Let $\mathbf{Z}$ be a $k$-dimensional Gaussian vector with mean $\mathbf{0}$ and covariance matrix $\mathsf{V}$. For $\epsilon \in (0, 1)$, define the multidimensional counterpart of the $Q^{-1}(\cdot)$ function
	\begin{align}
	\mathcal{Q}_{\textnormal{inv}}(\mathsf{V}, \epsilon) \triangleq \left\{ \bm{\tau} \in \mathbb{R}^k : \Prob{\mathbf{Z} \leq \bm{\tau}} \geq 1 - \epsilon \right\}. 
	\end{align}
	It is shown in \cite{huang2012finite, tan2014dispersions} that the achievable rate region satisfies
	\begin{align}
	(R_1, R_2): \begin{bmatrix}
	R_1 \\ R_2 \\ R_1 + R_2 \end{bmatrix} &\supseteq \bigcup_{P_{X_1}, P_{X_2}} \bigg\{\E{\boldsymbol{\imath}_2} - \frac{1}{\sqrt{n}} Q_{\textnormal{inv}}(\mathsf{V}_2, \epsilon) \notag\\
	&+ O\left( \frac{\log n}{n} \right) \mathbf{1} \bigg\}. \label{eq:vectorach}
	\end{align}
	Let $(R_1, R_2)$ converge to a point on the sum-rate boundary, i.e., $R_1 = D_1 + \frac{1}{\sqrt{n}} L_1 + O \left( \frac{\log n}{n} \right)$ and $R_2 = D_2 + \frac{1}{\sqrt{n}} L_2 + O \left( \frac{\log n}{n} \right)$, where
	\begin{align}
	D_1 &= \lambda I_2(X_1; Y) + \bar{\lambda} I_2(X_1; Y | X_2) \notag \\
	D_2 &= \bar{\lambda} I_2(X_2; Y) +  \lambda I_2(X_2; Y | X_1), \label{eq:I1I2}
	\end{align}
	for some $0 < \lambda < 1$, and $\bar{\lambda} = 1- \lambda$. We are going to prove that under this condition, \eqref{eq:vectorach} and \eqref{eq:L1L2} are equivalent. The proof is similar to \cite[Th.~2]{molavianjazi2015second}, which proves the reduction for Gaussian MACs on the sum-rate boundary. 
	
	For any nontrivial MAC, we have
	\begin{align}
	\Delta_1 &= I_2(X_1; Y | X_2) - I_2(X_1; Y) > 0 \\
	\Delta_2 &= I_2(X_1; Y | X_2) - I_2(X_1; Y) > 0.
	\end{align}
	
	We consider a Gaussian random vector $\mathbf{Z} = (Z_1, Z_2, Z_3) \sim \mathcal{N}(\boldsymbol{0}, \mathsf{V}_2) \in \mathbb{R}^3$. Due to the definition of $Q_{\textnormal{inv}}(\mathsf{V}, \epsilon)$ in \eqref{eq:qinvvec} and \eqref{eq:vectorach}, $\mathbf{Z}$ satisfies
	\begin{align}
	&1-\epsilon \nonumber \\
	&\leq \mathbb{P} \left[ \mathbf{Z} \leq \sqrt{n} \begin{bmatrix} I_2(X_1; Y| X_2) - R_1 \\ I_2(X_2; Y | X_1) - R_2 \\ I_2(X_1, X_2; Y) - (R_1 + R_2)  \end{bmatrix} \right. \notag \\
	& \quad +  O \left( \frac{\log n}{\sqrt{n}} \right) \mathbf{1} \Bigg] \\
	&= \mathbb{P} \left[ \mathbf{Z} \leq \begin{bmatrix} \sqrt{n} \lambda \Delta_1 - L_1 \\ \sqrt{n} \bar{\lambda} \Delta_2 - L_2 \\ -(L_1 + L_2)  \end{bmatrix}  \right] +  O \left( \frac{\log n}{\sqrt{n}} \right),  \label{eq:Zleq}
	\end{align}
	where \eqref{eq:Zleq} follows from \eqref{eq:I1I2} and Taylor expansion of $Q_{\textnormal{inv}}(\mathsf{V}, \cdot)$. 	Denote
	\begin{align}
	&p_{\textnormal{diff}} \triangleq \Prob{Z_3 \leq -(L_1 + L_2)} - 
	\mathbb{P} \left[ \mathbf{Z} \leq \begin{bmatrix} \sqrt{n} \lambda \Delta_1 - L_1 \\ \sqrt{n} \bar{\lambda} \Delta_2 - L_2 \\ -(L_1 + L_2)  \end{bmatrix}  \right] \notag \\
	&= \mathbb{P}\Big[\cB{\cB{Z_1 > \sqrt{n} \lambda \Delta_1 - L_1} \cup \cB{Z_2 > \sqrt{n} \lambda \Delta_2 - L_2}} \notag \\
	&\quad \cap \{Z_3 \leq -(L_1 + L_2)\} \Big] \\
	&\leq \exp\cB{-cn} \label{eq:gausschernoff} \\
	&= O \left( \frac{\log n}{\sqrt{n}} \right),
	\end{align}
	for some $c > 0$ constant, where \eqref{eq:gausschernoff} follows from Chernoff bound on Gaussian random variables and union bound. This implies that \eqref{eq:Zleq} is equivalent to
	\begin{align}
	1 - \epsilon \leq \Prob{Z_3 \leq -(L_1 + L_2)} + O \left( \frac{\log n}{\sqrt{n}} \right),
	\end{align}
	where $Z_3 \sim \mathcal{N}(0, V_2)$, which proves the claim in \eqref{eq:L1L2}. In \thmref{thm:ach}, notice that $I_2(X_1; Y| X_2) > I_2(X_1; Y)$ follows from \lemref{lem:mutualinfo} and that we pick the symmetrical rate point $R_1 = R_2$.
\end{IEEEproof} 
}

\section{Proof of \lemref{lem:scarlett}}\label{app:scarlett}
\setcounter{equation}{0}
To prove \lemref{lem:scarlett}, we first derive the saddle point condition for the MAC.
\begin{theorem}[Saddle point condition for the MAC]\label{thm:saddlepoint}
    Let $\mathcal{P}_1$ and $\mathcal{P}_2$ be convex set of distributions on alphabets $\mathcal{X}_1$ and $\mathcal{X}_2$, respectively. Suppose that there exists a product distribution $P_{X_1^*} P_{X_2^*}$ such that
    \begin{align}
        \sup\limits_{\substack{P_{X_1} P_{X_2} \\ P_{X_1} \in \mathcal{P}_1, P_{X_2} \in \mathcal{P}_2}} I_2(X_1, X_2; Y_2) = I_2(X_1^*, X_2^*; Y_2^*) = I_2^*, \label{eq:I2starthm}
    \end{align}
    where $P_{Y_2^*|X_1^*, X_2^*} = P_{Y_2|X_1, X_2}$. Then, for all $P_{X_1} \in \mathcal{P}_1$ and for all $Q_{Y_2}$, it holds that
    \begin{align}
        &D(P_{X_1}P_{X_2^*} P_{Y_2|X_1, X_2} \| P_{X_1}P_{X_2^*} P_{Y_2^*}) \notag \\
        \leq&I_2^*\label{eq:saddlepointleft} \\
        \leq& D(P_{X_1^*} P_{X_2^*} P_{Y_2|X_1, X_2} \| P_{X_1^*} P_{X_2^*} Q_{Y_2}). \label{eq:saddlepointright}
    \end{align}
\end{theorem}

\begin{IEEEproof}[Proof of \lemref{lem:scarlett}]
\lemref{lem:scarlett} follows by an application of \thmref{thm:saddlepoint} to the setting where 
$\mathcal{P}_1$ includes the set of all distributions with a singleton on $\mathcal{X}_1$ having probability 1, i.e., $\{\delta_{x_1} \colon x_1 \in \mathcal{X}_1\} \subseteq \mathcal{P}_1$, and $I_2^* < \infty$. Particularizing $P_{X_1}$ in \eqref{eq:saddlepointleft} to any $P_{X_1} = \delta_{x_1}$ with $x_1 \in \mathcal{X}_1$ yields
\begin{align}
    &D( P_{X_2^*} P_{Y_2|X_1 = x_1, X_2} \|  P_{X_2^*} P_{Y_2^*}) \leq I_2^* \label{eq:Dcondstar}
\end{align}
for all ${x_1 \in \mathcal{X}_1}$. Since the left-hand side of \eqref{eq:Dcondstar} is equal to the conditional expectation of $\imath_2(X_1^*, X_2^*; Y_2^*)$ given $X_1^* = x_1$,
\eqref{eq:EcondC} follows with less than or equal to. The equality in \eqref{eq:EcondC} follows since otherwise \eqref{eq:Dcondstar} would give the contradiction $I_2(X_1^*, X_2^*; Y_2^*) < I_2^*$. 
\end{IEEEproof}

\begin{IEEEproof}[Proof of \thmref{thm:saddlepoint}]
    The proof of \thmref{thm:saddlepoint} is similar to the proof of the saddle point condition for point-to-point channels in \cite[Th.~4.4]{polyanskiyLectureNotes} and extends \cite[Th.~4.4]{polyanskiyLectureNotes} to the MAC. Although the optimization in \eqref{eq:I2starthm} is not convex in general \cite{watanabe1996}, the optimization
\begin{align}
    \sup\limits_{P_{X_1} \in \mathcal{P}_1} I_2(X_1, X_2^*; Y_2),
\end{align}
where $P_{X_1 X_2^* Y_2} = P_{X_1} P_{X_2^*} P_{Y_2|X_1, X_2}$ is convex.

Inequality \eqref{eq:saddlepointright} follows from the golden formula (e.g., \cite[Th.~3.3]{polyanskiyLectureNotes})
\begin{align}
    I_2^* &= D(P_{X_1^*} P_{X_2^*} P_{Y_2|X_1, X_2} \| P_{X_1^*} P_{X_2^*} P_{Y_2^*}) \\
    &= D(P_{X_1^*} P_{X_2^*} P_{Y_2|X_1, X_2} \| P_{X_1^*} P_{X_2^*} Q_{Y_2}) - D( P_{Y_2^*}\| Q_{Y_2})
\end{align}
and the nonnegativity of the relative entropy. Notice that for $I_2^* = \infty$, \eqref{eq:saddlepointleft} is trivial. Assume that $I_2^* < \infty$. Fix any $P_{X_1} \in \mathcal{P}_1$. Let $\lambda \in (0, 1)$. Set
\begin{align}
    P_{X_{1\lambda}} = \lambda P_{X_1} + (1-\lambda) P_{X_1^*} \in \mathcal{P}_1.
\end{align} 
Let $\theta \sim \textrm{Bernoulli}(\lambda)$, so that $P_{X_{1\lambda} | \theta = 0} = P_{X_1^*}$ and $P_{X_{1\lambda}| \theta = 1} = P_{X_1}$, and let 
\begin{align}
    P_{X_{1\lambda} X_2^* Y_{2\lambda}} = P_{X_{1\lambda}} P_{X_2^*} P_{Y_2|X_1, X_2}.
\end{align}
Then
\begin{IEEEeqnarray}{rCl}
    I_2^* &\geq& I_2(X_{1\lambda}, X_2^*; Y_{2\lambda}) \label{eq:saddlepointfirst} \\
    &=& D(P_{X_{1\lambda}} P_{X_2^*} P_{Y_2|X_1, X_2} \| P_{X_{1\lambda}} P_{X_2^*} P_{Y_{2\lambda}}) \\
    &=& \lambda D( P_{X_1} P_{X_2^*} P_{Y_2|X_1, X_2} \|  P_{X_1} P_{X_2^*} P_{Y_{2\lambda}}) \notag \\
    && + (1-\lambda) D(P_{X_1^*} P_{X_2^*} P_{Y_2|X_1, X_2} \| P_{X_1^*} P_{X_2^*} P_{Y_{2\lambda}}) \IEEEeqnarraynumspace \\
    &\geq& \lambda D(P_{X_1} P_{X_2^*}  P_{Y_2|X_1, X_2} \| P_{X_1} P_{X_2^*} P_{Y_{2\lambda}}) \IEEEeqnarraynumspace \notag \\
    &&+ (1-\lambda) I_2^*, \label{eq:saddlepointlaststep}
\end{IEEEeqnarray}
where \eqref{eq:saddlepointlaststep} follows from \eqref{eq:saddlepointright}. By subtracting $(1-\lambda) I_2^*$ from both sides of \eqref{eq:saddlepointlaststep} and dividing by $\lambda$, we get
\begin{align}
    I_2^* \geq D(P_{X_1} P_{X_2^*} P_{Y_2|X_1, X_2} \| P_{X_1} P_{X_2^*} P_{Y_{2\lambda}}). \label{eq:lowersc}
\end{align}
By taking $\liminf_{\lambda \to 0}$ in \eqref{eq:lowersc} and applying the lower semicontinuity of the relative entropy (e.g., \cite[Th.~3.6]{polyanskiyLectureNotes}), \eqref{eq:saddlepointleft} is proved.
\end{IEEEproof}
Note that $(P_{X_1^*}, P_{X_2^*})$ does not have to be unique for \thmref{thm:saddlepoint} and \lemref{lem:scarlett} to hold.

\section{Adder-erasure RAC}\label{app:adder}
Here, we approximate the sum-capacity and dispersion of the adder-erasure RAC for a large number of transmitters $(k)$. 
\begin{theorem}\label{thm:ikvk}
The optimal input distribution for the adder-erasure RAC defined in \eqref{eq:addererasure} is the Bernoulli(1/2) distribution at all encoders. That input distribution achieves the sum-rate capacity, and
\begin{align}
\setcounter{equation}{0}
&I_k = (1 - \delta) \left( \frac 1 2  \log \frac{\pi e k}{2} - \frac{\log e}{12 k^2 } \right)+ O(k^{-3})  \label{eq:addercapacity}\\
&V_k = (1 - \delta) \Bigg[ \frac{\delta}{4} \log^2 \frac{\pi e k}{2} + \frac{\log^2 e}{2}  - \frac{ \log^2 e}{2 k} \notag \\
&\quad - \left( \frac{ \log e}{2} + \frac{\delta \log \frac{\pi e k}{2} }{12} \right) \frac{\log e}{k^2} \Bigg] + O\left( \frac{\log k}{k^3} \right)  \label{eq:adderdispersion}.
\end{align}  
\end{theorem}

The calculation leading to \thmref{thm:ikvk} is presented in Lemmas~\ref{lem:demoivre}--\ref{lem:binom}, which rely on Stirling's approximation and the Taylor series expansion.  

Consider a binomial random variable $X\sim \mbox{Binom}(n, 1/2)$.
\lemref{lem:demoivre}, below, shows that the probability mass that this Binomial distribution puts at $k$ is well approximated by
\begin{align}
\tilde{P}_X(k) &\triangleq \frac{1}{\sqrt{ \frac {\pi n} 2 }} e^{-\frac{(k - \frac n 2)^2}{{\frac n 2} }}  \left( 1 + \frac{f(k)}{n} + \frac{g(k)}{n^2} \right),  \label{eq:ptilde}
\end{align}
where
\begin{align}
f(x) &\triangleq -\frac 1 {12} \frac{\left( 2x - n \right)^4}{n^2} + \frac 1 2 \frac{\left( 2 x - n \right)^2}{n} - \frac 1 4 \\
g(x) &\triangleq \frac 1 {288} \frac{\left(2 x - n \right)^8}{n^4}  - \frac 3 {40} \frac{\left( 2 x - n \right)^6}{n^3}  + \frac {19}{48} \frac{\left( 2 x - n \right)^4}{n^2} \notag \\
&\quad - \frac {11}{24} \frac{\left(2 x - n \right)^2}{n} + \frac 1 {32}.
\end{align}
Define the interval
\begin{align}
\mathcal{K} \triangleq \left[\frac{n}{2} - \frac{A}{2}  \sqrt{n \log n}, \,\, \frac{n}{2} + \frac{A}{2}  \sqrt{n \log n} \right] \label{eq:kinterval}
\end{align}
for some constant $A > 0$. 

\begin{lemma} \label{lem:demoivre}
Let $X \sim \mathrm{Binom}(n, 1/2)$. Then for any $k \in \mathcal{K}$,
\begin{align}
P_X(k) = \binom{n}{k} 2^{-n} = \tilde{P}_X(k) \left( 1 + O \left( \frac {\log^6 n}{n^3} \right) \right). \label{eq:pptilde}
\end{align}
\end{lemma}
\begin{IEEEproof}[Proof of \lemref{lem:demoivre}]
We apply Stirling's approximation \cite[eq.~(6.1.37)]{abramowitz1972handbook} 
\begin{align}
n! = \sqrt{2 \pi} n^{n + \frac 1 2} e^{-n} \left( 1 + \frac{1}{12n} + \frac{1}{288n^2 } + O(n^{-3}) \right), \label{eq:stirlingseries}
\end{align}
and a Taylor series expansion of $\binom{n}{k}$ around $x = 0$, where 
\begin{align}
k = \frac{n}{2} + \frac{x}{2} \sqrt{n \log n}, \label{eq:kequals}
\end{align}
to $P_X(k) = \binom{n}{k} 2^{-n}$, to derive \eqref{eq:pptilde}.
\ver{}{
\begin{align}
\phi(j) \triangleq 1 + \frac {1}{12 j} + \frac 1 {288j^2}.
\end{align}
By \lemref{lem:stirling}, we obtain
\begin{align}
\binom{n}{k} 2^{-n} &= \frac{1}{\sqrt{2 \pi}} \frac{2^{-n} n^{n + \frac 1 2}}{k^{k + \frac 1 2} (n-k)^{n-k + \frac 1 2}} \notag\\
&\,\, \cdot \frac{\phi(n)}{\phi(k) \phi(n-k)} + O(n^{-3}). \label{eq:binommult} 
\end{align}
Using \eqref{eq:kequals}, we expand the Taylor series around $x = 0$ of each factor in \eqref{eq:binommult} to get
\begin{align}
 &\left(\frac{ 2^{-n} n^{n + \frac 1 2}}{k^{k + \frac 1 2} (n-k)^{n-k + \frac 1 2}}\right) = \frac 1 {\sqrt { \frac n 4} } \exp \left\{ -  \frac {x^2 \log n}{2} \right\} \notag \\
 &\cdot \bigg( 1 - \frac { x^4 \log^2 n}{12n} + \frac{ x^2 \log n}{2n} + \frac{x^8 \log^4 n}{288 n^2} - \frac{3 x^6 \log^3 n}{40 n^2} \notag \\
 & + \frac{3 x^4 \log^2 n}{8 n^2} + O \left( \frac{\log^6 n}{n^3} \bigg) \right) \label{eq:taylor2},
\end{align}
and
\begin{align}
& \frac{\phi(n)}{\phi(k) \phi(n-k)} = 1 - \frac {1}{4n} + \frac{3 - 32x^2 \log n}{96n^2} + O \left( \frac{\log^2 n}{n^3} \right) \label{eq:taylor3}.
\end{align}

Substituting the Taylor series expansions in \eqref{eq:taylor2} and \eqref{eq:taylor3} into \eqref{eq:binommult} completes the proof.
}
\end{IEEEproof}

Let $V(X)$ 
\begin{align}
V(X)& = \mathrm{Var} \left[ \log \frac{1}{P_X(X)} \right].\label{eq:varent_sum}
\end{align}
denote the varentropy of $X$.

\begin{lemma}[Entropy and varentropy of $\mathrm{Binom}\left(n, 1/2 \right)$] \label{lem:binom}
For $X \sim \mathrm{Binom}\left(n, 1/2 \right)$, 
\begin{align}
H(X) &= \frac 1 2 \log  \frac {\pi e n} {2} - \frac{\log e}{12n^2} + O(n^{-3}) \label{eq:HXbinom}\\
V(X) &= \frac{\log^2 e} 2  - \frac{\log^2 e}{2n} -  \frac{\log^2 e}{2n^2} + O(n^{-3}). \label{eq:VXbinom}
\end{align}
\end{lemma}
\begin{IEEEproof}[Proof of \lemref{lem:binom}]
\ver{}{We denote the centered version of $X$ by 
\begin{align}
\bar X \triangleq X - \frac n 2, 
\end{align}
and we calculate its moments as follows: 
\begin{align}
\mathbb{E}\left[\bar X^2\right ] &= \frac n 4  \notag \\
\mathbb{E}\left[\bar X^4\right ] &= \frac n {16} (3n - 2) \notag \\
\mathbb{E}\left[\bar X^6\right ] &= \frac n {64}  (15n^2 - 30 n + 16) \notag  \\
\mathbb{E}\left[\bar X^8\right ] &= \frac n {256} (105n^3 - 420 n^2 + 588 n - 272). \label{eq:binommoments}
\end{align}
}
Let $\tilde{T}(k)$ denote the first 3 terms of the Taylor series expansion of $\log \frac 1{\tilde{P}_X(k)} $ around $ \frac n 2$ evaluated at $k$, giving
\begin{align}
\tilde{T}(k) &\triangleq \frac 1 {2 } \log { \frac {\pi n} 2 } +  \log e \bigg(\frac {(k-\frac n 2)^2}{ \frac n 2} \notag \\
&\quad -  \frac{f(k)}{n} + \frac{-g(k) + \frac{f^2(k)}{2}}{n^2} \bigg). \label{eq:Tk}
\end{align}
Recall the definition of interval $\mathcal{K}$ from \eqref{eq:kinterval}. Then we can write the entropy $H(X)$ as
\begin{IEEEeqnarray}{rCl}
H(X) &=&  \sum_{k = 0}^n  \frac{\binom{n}{k}}{2^n} \log \left( \frac{2^n}{\binom{n}{k}}  \right)  \label{eq:binoment} \\
&=& \E{\tilde{T}(X)} \notag \\
&& + \E{ \left( \log \frac 1{P_X(X) } - \tilde{T}(X) \right) 1\{ X \in \mathcal{K} \} } \notag \\
&& + \E{ \left( \log \frac 1{P_X(X)} - \tilde{T}(X) \right) 1\{ X \notin \mathcal{K} \} }. \IEEEeqnarraynumspace  \label{eq:HX2}
\end{IEEEeqnarray}
Using the moments of $\text{Binom}\left(n, 1/2 \right)$ (e.g., \cite[eq.~(26.1.20)]{abramowitz1972handbook}), the first term in \eqref{eq:HX2} is 
\begin{align}
\E{\tilde{T}(X)} =  \frac 1 2 \log \frac{\pi e n}{2} -\frac{\log e}{12n^2}.  \label{eq:Tfirst}
\end{align}
By \lemref{lem:demoivre}, the second term in \eqref{eq:HX2} is
\begin{align}
 \E{ \left( \log \frac 1{P_X(X) } - \tilde{T}(X) \right) 1\{ X \in \mathcal{K} \} } = O \left( \frac{\log^6 n}{n^3}\right). \label{eq:Tsecond}
\end{align}
By Hoeffding's inequality, 
\begin{align}
\mathbb{P}\left[X \notin \mathcal{K} \right] \leq 2 n^{-\frac{A^2 \log e} 2 } \label{eq:hoeffding},
\end{align}
where $A$ is the constant in \eqref{eq:kinterval}.
Since the minimum of $P_X(k)$ over $k$ is achieved at $k = n$, using \eqref{eq:hoeffding}, we get 
\begin{align}
\E{  \log \frac 1{P_X(X) } 1\{ X \notin \mathcal{K} \} }  = O \left( \frac{\log^6 n}{n^3}\right) \label{eq:hoeffdingsmall}
\end{align} 
for $A \geq \frac{3}{\sqrt{\log e}}$. Similarly, by taking the derivative of $\tilde{T}(k)$, one can show that $\tilde{T}(k) \leq \tilde{T}(n) \leq n$ for all $k \in [0, n]$, which gives 
\begin{align}
\E{  \tilde{T}(X)  1\{ X \notin \mathcal{K} \} }  = O \left( \frac{\log^6 n}{n^3}\right). \label{eq:Tthird2}
\end{align}
Combining \eqref{eq:HX2}--\eqref{eq:Tsecond}, \eqref{eq:hoeffdingsmall}--\eqref{eq:Tthird2} gives 
\begin{align}
H(X) = \frac 1 2 \log \frac{\pi e n}{2} -\frac{\log e }{12n^2} + O \left( \frac{\log^6 n}{n^3} \right) \label{eq:bmoments2}.
\end{align}

Via an argument similar to \eqref{eq:hoeffdingsmall} and \eqref{eq:Tthird2}, we can show that for $A \geq \frac{4}{\sqrt{\log e}}$, the contribution of $k \notin \mathcal{K}$ to the varentropy is $O \left( \frac{\log^6 n}{n^3} \right) $. Therefore, using the moments of $\text{Binom}(n, 1/2)$ and \lemref{lem:demoivre}, we can approximate the varentropy $V(X)$ as 
\begin{align}
V(X) &= \E{\log^2 \frac 1 {P_X(X)} } - (H(X))^2 \\
& = \E{(\tilde{T}(X))^2} - (H(X))^2 + O \left( \frac{\log^6 n}{n^3} \right) \\
&=  \log^2 e \left(\frac{1}{2} - \frac {1} {2n} - \frac {1} {2n^2} \right) + O \left( \frac{\log^6 n}{n^3} \right).
\end{align}

The above analyses use the first 3 terms of the Stirling series \eqref{eq:stirlingseries} to obtain the remainder $O \left( \frac{\log^6 n}{n^3} \right)$. Applying the same analyses with 4 terms of the Stirling series improves the remainder to $O(n^{-3})$, as claimed in \eqref{eq:HXbinom} and \eqref{eq:VXbinom} in the statement of \lemref{lem:binom}.
\end{IEEEproof}
We are now equipped to prove \thmref{thm:ikvk}.
\begin{IEEEproof}[Proof of \thmref{thm:ikvk}]
Define
\begin{align}
E \triangleq 1\{Y = \mathsf{e} \}.
\end{align} 
By the chain rule for entropy, we have for the adder-erasure RAC
\begin{align}
I_k(X_{[k]}; Y_k) &= H(Y_k) - H(Y_k | X_{[k]}) \\
&= H(Y_k, E) - H(E) \\
&= H(Y_k | E) \\
&= (1-\delta) H(Y_k | E = 0). 
\end{align}
Given the independent inputs $X_i \sim \mbox{Bernoulli}(p_i)$ for $i \in [k]$, $H(Y_k | E = 0)$ is equal to the entropy of the sum of $k$ independent Bernoulli random variables with parameters $(p_1, \dots, p_k)$, which is maximized when $p_i = 1/2 $ for all $i$ \cite{shepp1981multinomial}. Therefore, for any $\delta \in [0, 1]$, the equiprobable input distribution at all encoders, $X_i^* \sim \text{Bernoulli}(1/2)$, maximizes the mutual information $I_k(X_{[k]}; Y_k)$ for all $k$. Let $(X_{[k]}^* Y_k^*) \sim P_{X_{[k]}^*} P_{Y_k|X_{[k]}}$. Then
\begin{align}
I_k(X_{[k]}^*; Y_k^*) = (1-\delta) H(Z),
\end{align}
where $Z \sim \mbox{Binom}(k, 1/2)$, and \eqref{eq:addercapacity} follows from \lemref{lem:binom}.
Furthermore,
\begin{align}
\imath_k(X^*_{[k]}; Y_k^*) = \begin{cases}
0 &\mbox{ w.p. } \delta \\
\log \frac{2^k}{\binom{k}{i}} &\mbox{ w.p. } (1-\delta) \frac{\binom{k}{i}}{2^k},  \quad 0 \leq i \leq k \label{eq:midensityadder},\\ 
\end{cases}
\end{align}
which gives 
\begin{align}
V_k = \Var{\imath_k(X^*_{[k]}; Y_k^*)} 
&= (1-\delta) \left[ V(Z) + \delta (H(Z))^2 \right],
\end{align}
and \eqref{eq:adderdispersion} follows from \lemref{lem:binom}.
\end{IEEEproof}

\section{Bound on the Cardinality $|\mathcal{U}|$}\label{app:caratheodory}
\setcounter{equation}{0}
While the analysis in Section~\ref{sec:thmpe} employs common randomness $U$ with $|\mathcal{U}| = |\cX|^{Mn_K}$, \cite[Th.~19]{polyanskiy2011feedback} shows that $|\mathcal{U}| \leq K + 2$ suffices to achieve the optimal performance. \thmref{thm:caratheodory}, stated next, improves the cardinality bound on $|\mathcal{U}|$ from $K + 2$ \cite[Th.~19]{polyanskiy2011feedback} to $K + 1$ by using the connectedness of the set of achievable error vectors defined in \eqref{eq:Gu}.
\begin{theorem}\label{thm:caratheodory} If an $(M, \lbrace (n_k, \epsilon_k) \rbrace_{k = 0}^K )$ RAC code exists, then there exists an $(M, \lbrace (n_k, \epsilon_k) \rbrace_{k = 0}^K )$ RAC code with $|\mathcal{U}| \leq K + 1$. 
\end{theorem}
\begin{IEEEproof}[Proof of \thmref{thm:caratheodory}]
	For fixed $M, n_0, \dots, n_K$, let $G_u$ denote the set of achievable error vectors compatible with message size $M$, blocklengths $n_0, \dots, n_K$, and cardinality $|\mathcal{U}| \leq u$; that is,
	\begin{align}
	G_u &= \{(\epsilon'_0, \dots, \epsilon'_K): \exists (M, \lbrace (n_k, \epsilon'_k) \rbrace_{k = 0}^K ) \mbox { code with } \notag \\
	&\quad |\mathcal{U}| \leq u \}. \label{eq:Gu}
	\end{align}
	Let $G$ denote the set of achievable error vectors compatible with message size $M$ and blocklengths $n_0, \dots, n_K$; that is,
	\begin{align}
	G &= \{(\epsilon'_0, \dots, \epsilon'_K): \exists (M, \lbrace (n_k, \epsilon'_k) \rbrace_{k = 0}^K ) \mbox { code} \}.
	\end{align}
	As observed in \cite[Proof of Th.~19]{polyanskiy2011feedback}, $G = G_{|\cX|^{Mn_K}}$ is the convex hull of $G_1$. Indeed, every vector $(\epsilon_0', \dots, \epsilon_K')$ in $G$ is a convex combination of vectors in $G_1$, and the coefficients of the convex combination are determined by the distribution of the common randomness random variable $U$. 
	
	Furthermore, $G_1$ is a connected set. To see this, take any $\bm{\epsilon}_1, \bm{\epsilon}_2 \in G_1$. For any $\bm{\epsilon}' \geq \bm{\epsilon}$ with $\bm{\epsilon} \in G_1$, the line segments $L_i = \{\lambda \bm{\epsilon}_i + (1-\lambda) \boldsymbol{1} \colon \lambda \in [0, 1]\}$, $i = 1, 2$, also belong to $G_1$, and the path $L_1 \cup L_2$ connects $\bm{\epsilon}_1$ and $\bm{\epsilon}_2$. Therefore, $G_1$ is a connected set.
	
	Since $G = \textnormal{conv}(G_1) \subset \mathbb{R}^{K + 1}$, and $G_1$ is a connected set, by Fenchel-Eggleston-Carath\'eodory's theorem \cite[Th.~18~(ii)]{eggleston1958}, $G = G_{K + 1}$ holds. Therefore, $(\epsilon_0, \dots, \epsilon_K) \in G$ implies that $(\epsilon_0, \dots, \epsilon_K) \in G_{K + 1}$.
\end{IEEEproof}

\section{Composite Hypothesis Testing} \label{app:hypothesis}
\setcounter{equation}{0}
We begin with a lemma that is used in the proof of \thmref{thm:huangmoulin}. See \figref{fig:qq} for an illustration of \lemref{lem:generalminmax}.
\ver{}{
Throughout this section,
\begin{align}
\mathbf{Z} \sim \mathcal{N}(\boldsymbol{0}, \mathsf{V}).
\end{align}
\begin{theorem}[Multidimensional Berry-Esseen (Bentkus)\cite{bentkus2003}] \label{th:mbe}
Let $\mathbf{U}_1, \dots, \mathbf{U}_n$ be zero-mean i.i.d. random vectors in $\mathbb{R}^d$. Assume that the minimum eigenvalue of the covariance matrix $\mathsf{V} = \textnormal{Cov}[\mathbf{U_1}]$ satisfies $\lambda_{\min}(\mathsf{V}) > 0$, and that $T = \E{ \Vert{ \mathbf{U}_1 }\Vert_2^3 } < \infty$. Then
\begin{align}
\sup_{\mathcal{A} \in \mathfrak{C}_d} \left \lvert \Prob{\frac 1 {\sqrt{n}} \sum_{i = 1}^n \mathbf{U}_i \in \mathcal{A}} - \mathbb{P}[\mathbf{Z} \in \mathcal{A}] \right \rvert \leq \frac{400 d^{1/4} T}{\lambda_{\min}(\mathsf{V})^{3/2} \sqrt{n}}, \label{eq:multbe}
\end{align}
where $\mathfrak{C}_d$ is the set of all convex, Borel measurable subsets of $\mathbb{R}^d$.
\end{theorem}

The following lemma is the reversed version of \cite[Lemma 47]{polyanskiy2010Channel} for random vectors.
\begin{lemma} \label{thm:rev47mult}
	Under the assumptions of \thmref{th:mbe}, for any constant vector $\mathbf{C}$, there exists a constant vector $\mathbf{e}$ such that
\begin{align}
&\E{\exp\left\{ -\sum_{j = 1}^n \mathbf{U}_j \right\} 1 \left\{ \sum_{j = 1}^n \mathbf{U}_{j} \geq \sqrt{n}\mathbf{C} \right\} } \notag \\
& \geq  \exp\left\{-\sqrt{n}\mathbf{C}- \frac 1 2 \log n \boldsymbol{1}+  \mathbf{e} \right\},
\end{align}
where vector inequalities are understood element-wise.  
\end{lemma}
\begin{IEEEproof}
Fix constants $c > 0$, and $t \in [d]$.
By \thmref{th:mbe}, we have
\begin{align}
&\Prob{\mathbf{C} \leq \frac{1}{\sqrt{n}} \sum_{j = 1}^n \mathbf{U}_j  \leq \mathbf{C} + \mathbf{a}_{t, n} } \notag \\
&\quad \geq \Prob{ \mathbf{C} \leq  \mathbf{Z} \leq \mathbf{C} + \mathbf{a}_{t, n} } - \frac{\gamma}{\sqrt{n}},\label{eq:mbeused8}
\end{align}
where 
\begin{align}
\mathbf{a}_{t, n} = \left(c, \dots, c, \frac{\delta}{\sqrt{n}}, c, \dots, c\right),
\end{align}
and $\delta$ in the $t^{\mathrm{th}}$ coordinate of $\mathbf{a}_{t, n}$ is a positive constant to be chosen later, and  
\begin{align}
\gamma =  \frac{400 k^{1/4} T}{\lambda_{\min}(\mathsf{V})^{3/2}} \label{eq:gammadef}
\end{align}
is the constant from \eqref{eq:multbe}.
For each $t$, denote the minimum of the pdf of $\mathbf{Z}$ over the set $\{\mathbf{z} : \mathbf{C} \leq \mathbf{z} \leq \mathbf{C} + \mathbf{a}_{t, n} \}$ by
\begin{align}
D_{t, n} \triangleq \min_{\mathbf{C} \leq \mathbf{z} \leq \mathbf{C} + \mathbf{a}_{t, n}} \frac{\exp\{ -\frac 1 2 \mathbf{z}^T \mathsf{V}^{-1} \mathbf{z} \} }{\sqrt{(2\pi)^d |\mathsf{V}|}}.
\end{align}
Since $\frac{\delta}{\sqrt{n}}$ decreases with $n$, 
\begin{align}
\inf_n D_{t, n} = D_{t, 1} > 0, \label{eq:Dt1}
\end{align}
and by \eqref{eq:mbeused8} and \eqref{eq:Dt1}, we get
\begin{align}
\Prob{\mathbf{C} \leq \frac{1}{\sqrt{n}} \sum_{j = 1}^n \mathbf{U}_j  \leq \mathbf{C} + \mathbf{a}_{t, n} } \geq \frac{\delta}{\sqrt{n}} c^{d-1} D_{t, 1} - \frac{\gamma}{\sqrt{n}}. \label{eq:reverseprob}
\end{align} 
We conclude
\begin{align}
&\E{\exp\left\{ -\sum_{j = 1}^n \mathbf{U}_j \right\}  1 \left\{ \sum_{j = 1}^n \mathbf{U}_{j} \geq \sqrt{n} \mathbf{C} \right\} }  \label{eq:reverseind}\\
&\geq \exp\left\{-(\sqrt{n}\mathbf{C} + \delta \boldsymbol{1}) - \frac 1 2 \log n \boldsymbol{1} \right\} \left(\delta c^{d-1} \min_{t \in [d]} D_{t, 1} - \gamma \right). \label{eq:reverselast}
\end{align}
To obtain \eqref{eq:reverselast}, we lower bounded the indicator in \eqref{eq:reverseind} by $1\left\{ \mathbf{C} \leq \frac 1 {\sqrt{n}}\sum_{j = 1}^n \mathbf{U}_j  \leq \mathbf{C} + \mathbf{a}_{t, n} \right\}$, and used \eqref{eq:reverseprob}.
The proof is completed by picking $\delta = 1 + \frac{\gamma}{c^{d-1} \min_{t \in [d]} D_{t, 1}}$.
\end{IEEEproof}
}

\begin{lemma}\label{lem:generalminmax}
Let $f \colon \mathbb{R}^d \to \mathbb{R}$ be a continuous function that satisfies coordinate-wise partial ordering, i.e., $f(\mathbf{x}) \leq f(\mathbf{y})$ for any $\mathbf{x}, \mathbf{y} \in \mathbb{R}^d$ with $\mathbf{x} \leq \mathbf{y}$. Then for any $a$ in the image of $f$ (denoted $a \in \mathrm{Im} f$), it holds that
\begin{align}
b^\star = \min_{ \mathbf{b} \in \mathbb{R}^d: f(\mathbf{b}) \geq a} \max_{1\leq j \leq d} b_j = \min_{\substack{x \in \mathbb{R}: f(x \boldsymbol{1}) \geq a}} x. \label{eq:lemmaminmax}
\end{align}
\end{lemma}
\begin{IEEEproof}
Since $a \in \mathrm{Im} f$, there exists some $\mathbf{b} \in \mathbb{R}^d$ such that $f(\mathbf{b}) = a$. Denote by $b_{\min}$ and $b_{\max}$ the minimum and maximum components of $\mathbf{b}$, respectively.
Since $f$ is nondecreasing, \begin{align}
f(b_{\min}\boldsymbol{1})\leq  a = f(\mathbf{b}) \leq f({b_{\max}} \boldsymbol{1}).
\end{align}
Therefore, since the function mapping $b$ to $f(b \boldsymbol{1})$ is continuous and nondecreasing, by the intermediate value theorem there exists some $b \leq b_{\max}$ such that $f(b \boldsymbol{1}) = a$. Equation \eqref{eq:lemmaminmax} follows.
\end{IEEEproof}
\begin{figure}[htbp] 
    \centering
    	\includegraphics[width=0.48\textwidth]{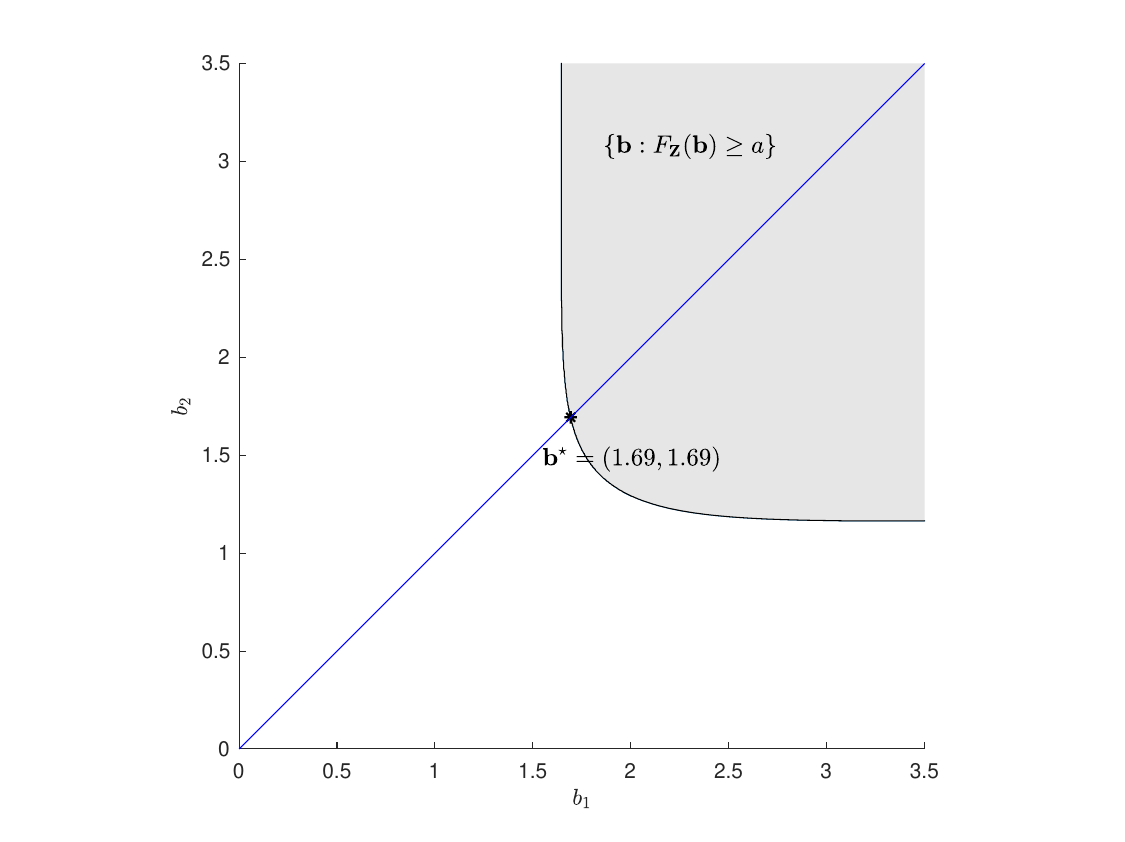}     
    	    \caption{An example to illustrate \lemref{lem:generalminmax}. Here $f(\mathbf{b}) = F_{\mathbf{Z}}(\mathbf{b})$ is the CDF of $\mathbf{Z} \sim \mathcal{N}(\boldsymbol{0}, \mathsf{V})$, where $\mathsf{V} = \left[\protect \begin{smallmatrix}1&0.4\\0.4&0.5 \protect \end{smallmatrix}\right]$. The shaded region illustrates the set $\{\mathbf{b} \in \mathbb{R}^2: f(\mathbf{b}) \geq a = 0.95\}$. \lemref{lem:generalminmax} shows that the minimax on this set is achieved at a point described by a scalar multiple of $\boldsymbol{1}$. For this example, the optimizer is $\mathbf{b}^\star = (1.69, 1.69)$. }
    	    \label{fig:qq} 
\end{figure}

Let $\mathbf{Z} \sim \mathcal{N}(\boldsymbol{0}, \mathsf{V})$. Define the multidimensional counterpart of the function $Q^{-1}(\cdot)$ as
\begin{align}
\mathcal{Q}_{\textnormal{inv}}(\mathsf{V}, \epsilon) \triangleq \left\{ \mathbf{z} \in \mathbb{R}^K : \Prob{\mathbf{Z} \leq \mathbf{z}} \geq 1 - \epsilon \right\}. \label{eq:qinvvec}
\end{align}

\ver{}{
\begin{theorem}\label{thm:huangregion}
Under the assumptions of \thmref{thm:huangmoulin}, it holds that
\begin{align}
&\exp\left\{ -n \mathbf{D} + \sqrt{n} \mathcal{Q}_{\textnormal{inv}}(\mathsf{V}, \epsilon_0) - \frac 1 2 \log n \boldsymbol{1} + O(1) \boldsymbol{1}\right\} \notag \\
&\subset \mathcal{E}_{\epsilon_0} (P_{Y_0}, \{P_{Y_k}\}_{k = 1}^K )  \notag \\
&\subset \exp\{ -n \mathbf{D} + \sqrt{n} \mathcal{Q}_{\textnormal{inv}}(\mathsf{V}, \epsilon_0) + o(\sqrt{n}) \boldsymbol{1}\} \label{eq:huang_conv}.
\end{align}
\end{theorem}

Huang and Moulin \cite[Proposition 1]{huang2014strong} presented \thmref{thm:huangregion} with $o(\sqrt{n})$ term replaced by $-\frac 1 2 \log n + O(1)$ in \eqref{eq:huang_conv}, indicating that the inner and outer bounds on the achievable error vectors match up to the third-order term in the exponent. However, according to the definition in \eqref{eq:qinvvec}, there exists some $\bm{\tau} \in Q_{\textnormal{inv}}(\mathsf{V}, \epsilon_0)$ such that $\bm{\tau}$ has an arbitrarily large coordinate. Therefore, the assumption in \cite{huang2014strong} that $Q_{\textnormal{inv}}(\mathsf{V}, \epsilon_0) = O(1) \boldsymbol{1}$ with respect to $n$, does not hold. We reprove the theorem considering this observation. The achievability part of \eqref{eq:huang_conv} is unaffected in the new proof. For the converse part of \eqref{eq:huang_conv}, we show that the second-order term in the error exponent remains unchanged, while we cannot confirm that the third-order term is $-\frac 1 2 \log n$ unless the threshold vector satisfies $\bm{\tau} = n \mathbf{D} - O(\sqrt{n}) \boldsymbol{1}$.

The proof of the achievability part of \thmref{thm:huangregion} follows an argument similar to \cite{huang2014strong}; the difference is that we use \cite[Lemma~47]{polyanskiy2010Channel} due to Polyanskiy \textit{et al.} instead of \cite[Lemma~2]{huang2014strong}.
\begin{IEEEproof}[Proof of the achievability part in \eqref{eq:huang_conv}]
We prove the achievability using the LLR tests given in \eqref{eq:LLR}. Pick any $\bm{\tau} = (\tau_1, \dots, \tau_K) \in n \mathbf{D} - \sqrt{n} \mathcal{Q}_{\textnormal{inv}}(\mathsf{V}, \epsilon_0 - \frac{\gamma}{\sqrt{n}})$, where $\gamma$ is defined in \eqref{eq:gammadef}.
Then by definition of $\mathcal{Q}_{\textnormal{inv}}(\mathsf{V}, \epsilon_0)$ in \eqref{eq:qinvvec}, we have 
\begin{align}
\Prob{\sqrt{n} \mathbf{Z} \geq \bm{\tau} - n \mathbf{D}} \geq 1 - \epsilon_0 + \frac{\gamma}{\sqrt{n}}.
\end{align}
Recall $h_n^{\mathrm{LLR}}(y^n)$ from \eqref{eq:LLR}, and let
\begin{align}
\mathbf{S}_k = (S_{k1}, \dots, S_{kK}) \triangleq h_n^{\mathrm{LLR}}(Y_k^n),
\end{align}
where $Y_k^n \sim P_{Y_k}^n$ for $k = 0, \dots, K$. By \thmref{th:mbe}, we have
\begin{align}
\mathbb{P}[ \mathbf{S}_0 \geq \bm{\tau} ] \geq 1 - \epsilon_0.
\end{align}
Thus, the threshold vector $\bm{\tau}$ satisfies the type-I error condition. Then $\beta_k$, the error probability in deciding $H_0$ when the true distribution is $P_{Y_k}$, is bounded as
\begin{align}
\beta_k &= \mathbb{P}[ \mathbf{S}_k \geq \bm{\tau}] \\
&= \mathbb{E}[ e^{-S_{0k}} 1 \{ \mathbf{S}_0 \geq \bm{\tau} \} ]  \label{eq:measurechange} \\ 
& \leq \mathbb{E}[ e^{-S_{0k}} 1 \{ S_{0k} \geq \tau_{k} \} ] \label{eq:step1} \\
& \leq 2 \left( \frac{ \log 2 }{\sqrt{2 \pi} }+ \frac{12 T_k}{\sigma_k^2} \right) \frac{1}{\sigma_k \sqrt{n}} \exp\{-\tau_{k}\}  \label{eq:lemma47used} \\
&= \exp\left\{ -\tau_{k} - \frac 1 2 \log n + O(1) \right\} \label{eq:huangachexp},
\end{align}
where \eqref{eq:measurechange} follows from a  measure change from $P_{Y_k}$ to $P_{Y_0}$, and \eqref{eq:lemma47used} follows from \cite[Lemma 47]{polyanskiy2010Channel} and the definitions $\sigma_k^2 = \textnormal{Var}\left[ \log \frac{P_{Y_0}(Y_0)}{P_{Y_k}(Y_0)}\right] > 0$ and $T_k = \mathbb{E}\left[ \left\lvert  \log \frac{P_{Y_0}(Y_0)}{P_{Y_k}(Y_0)} - D(P_{Y_0} \| P_{Y_k}) \right\rvert^3 \right] \leq T < \infty$.
By the Taylor series expansion of $Q_{\textnormal{inv}}(\mathsf{V}, \cdot)$, we conclude
\begin{align}
&\mathcal{E}_{\epsilon_0} (P_{Y_0}, \{P_{Y_k}\}_{k = 1}^K ) \notag \\
&\supset \exp\left\{ -n \mathbf{D} + \sqrt{n} \mathcal{Q}_{\textnormal{inv}}(\mathsf{V}, \epsilon_0) - \frac 1 2 \log n \boldsymbol{1} + O(1) \boldsymbol{1}\right\}.
\end{align}
\end{IEEEproof}

The following proof of the converse part of \thmref{thm:huangregion} corrects \cite[Proposition 1]{huang2014strong}  by taking into account the possibility that $Q_{\textnormal{inv}}(\mathsf{V}, \epsilon_0)$ may have arbitrarily large coordinates.
\begin{IEEEproof}[Proof of the converse part \eqref{eq:huang_conv}]
From \cite[Section IV-D]{huang2014strong}, the gap in the type-II error exponent between the LLR tests and the optimal tests is $O\left(\frac{1}{n}\right)$. Therefore, we only consider the LLR tests in the converse proof. For any LLR test, by definition \eqref{eq:LLR}, we have
\begin{align}
\mathbb{P}[ \mathbf{S}_0 \geq \bm{\tau} ] \geq 1 - \epsilon_0.
\end{align}
Then, by \thmref{th:mbe}, we have
\begin{align}
\Prob{\sqrt{n}\mathbf{Z} \geq \bm{\tau} - n\mathbf{D} } \geq 1 - \epsilon_0 - \frac{\gamma}{\sqrt{n}},   
\end{align}
where $\gamma$ is the constant defined in \eqref{eq:multbe}. By \eqref{eq:qinvvec},
\begin{align}
\bm{\tau} &\in n \mathbf{D} - \sqrt{n} \mathcal{Q}_{\textnormal{inv}}\left(\mathsf{V}, \epsilon_0 + \frac{\gamma}{\sqrt{n}}\right) \notag \\
 &= n \mathbf{D} - \sqrt{n} \mathcal{Q}_{\textnormal{inv}}(\mathsf{V}, \epsilon_0) + O(1) \boldsymbol{1} \label{eq:tauconv}.
\end{align}
Observe that \eqref{eq:tauconv} implies that $\bm{\tau} \leq n \mathbf{D} - \sqrt{n} \mathbf{C}  + o(\sqrt{n})\boldsymbol{1}$ for some constant vector $\mathbf{C}$. The type-II error $\beta_k$ is lower bounded as
\begin{align}
\beta_k &= \mathbb{P}[ \mathbf{S}_k \geq \bm{\tau}] \\
&\geq \mathbb{P}[ \mathbf{S}_k \geq n \mathbf{D} -  \sqrt{n} \mathbf{C} + o(\sqrt{n})\boldsymbol{1}] \\
&= \mathbb{E}[ e^{-S_{0k}} 1 \{ \mathbf{S}_0 \geq n \mathbf{D} -\sqrt{n} \mathbf{C}  + o(\sqrt{n})\boldsymbol{1} \} ]  \\ 
&\geq \exp\left\{-n D_k+ \sqrt{n} C_k + o(\sqrt{n}) \right\} \label{eq:steptau}, 
\end{align}
for all $k \in [K]$, where $\mathbf{D} = (D_1, \dots, D_K)$, $\mathbf{C} = (C_1, \dots, C_K)$, and \eqref{eq:steptau} follows from \lemref{thm:rev47mult}.

Define the set of indices
\begin{align}
\mathcal{B} \triangleq \{k \in [K]: \tau_{k} = nD_k - \sqrt{n}E_k + o(\sqrt{n}) \}, \label{eq:setB}
\end{align}
for some constants $E_1, \dots, E_K$. For all $k \notin \mathcal{B}$, we can express $\tau_k$ as
\begin{align}
{\tau}_{k} = n D_k - g_k(n)\sqrt{n} + o(g_k(n)\sqrt{n}), 
\end{align}
where $g_k(n)$ is a function of $n$ such that $\lim_{n \to \infty} g_k(n) = \infty$ (For any constant $\epsilon_0 \in (0, 1)$, $\lim_{n \to \infty} g_k(n) = -\infty$ is not possible.)

For $k \in \mathcal{B}$, we can bound the type-II errors as
\begin{align}
\beta_k &= \mathbb{P}[ \mathbf{S}_k \geq \bm{\tau}] \\
&\geq \exp\left\{-n D_k + \sqrt{n} E_k + o(\sqrt{n}) \right\}, \label{eq:step6}
\end{align}
where \eqref{eq:step6} follows from \eqref{eq:steptau}. We have for $k \notin \mathcal{B}$ that
\begin{align}
\beta_k &= \mathbb{P}[ \mathbf{S}_k \geq \bm{\tau}] \\
&= \exp\{-nD_k + h_k(n) \sqrt{n}  + o(h_k(n) \sqrt{n} ) \},
\end{align}
where $h_k(n) \leq g_k(n)$ by the upper bound in \eqref{eq:huangachexp}, and $\lim_{n\to \infty} h_k(n) = \infty$, since we can pick arbitrarily large $C_k$ for $k \notin \mathcal{B}$ in \eqref{eq:steptau}.

Define $\bm{\lambda} = (\lambda_k: \lambda_k = E_k \,\,\forall k \in \mathcal{B}, \, \lambda_k = h_k(n)\,\, \forall k \notin \mathcal{B})$. By union bound, \eqref{eq:tauconv} and \eqref{eq:setB},
\begin{align}
&\Prob{\bigcup_{k = 1}^k \{Z_k > \lambda_k\}} \\
&\leq \Prob{\bigcup_{k \notin \mathcal{B}} \{Z_k > h_k(n) \}}  + \Prob{\bigcup_{k \in \mathcal{B}} \{Z_k > E_k\}} \notag \\
&\leq o(1) + \epsilon_0.
\end{align}
Therefore, by Taylor series expansion of $Q_{\textnormal{inv}}(\mathsf{V}, \cdot)$, we obtain that
\begin{align}
\bm{\lambda} \in  \mathcal{Q}_{\textnormal{inv}}(\mathsf{V}, \epsilon_0 + o(1)) =   \mathcal{Q}_{\textnormal{inv}}(\mathsf{V}, \epsilon_0) + o(1) \boldsymbol{1},
\end{align}
which completes the proof.
\end{IEEEproof}

Finally, we derive the minimax error in \thmref{thm:huangmoulin} by applying \lemref{lem:generalminmax} to \eqref{eq:huang_conv}.
}
\begin{IEEEproof}[Proof of \thmref{thm:huangmoulin}]
For any $\epsilon_0 \in (0, 1)$, consider all composite hypothesis tests in the form given in \eqref{eq:originalht} that achieve type-I error no greater than $\epsilon_0$. Let
\begin{align}
&\mathcal{E}_{\epsilon_0} (P_{Y_0}, \{P_{Y_k}\}_{k = 1}^K ) \triangleq \Big\{ (e_1, \dots, e_K) : \exists \textnormal { a (randomized) test } \notag \\
&\text{such that} \notag \\
&\quad \Prob{\text{Decide } H_1 | H_0 } \leq \epsilon_0, \notag \\
& \quad \Prob{\text{Decide } H_0 | H_1  } = e_k, 1 \leq k \leq K \Big\} \label{def:error_region}
\end{align}
denote the set of type-II errors achievable by these tests.
Huang and Moulin \cite[Th.~1]{huang2014strong}\footnote{In the converse part of the proof of \cite[Th.~1]{huang2014strong}, Huang and Moulin show that for any LLR test \eqref{eq:LLR} with threshold vector $\bm{\tau}$ such that the type-I error is bounded by $\epsilon_0$, it holds that $
\bm{\tau} = n \mathbf{D} - \sqrt{n} \mathbf{b} + O(1) \boldsymbol{1}$
for some $\mathbf{b} \in Q_{\text{inv}}(\mathsf{V}, \epsilon_0)$. Then, it is assumed that $\mathbf{b} = O(1) \boldsymbol{1}$, and \cite[Lemma 2]{huang2014strong} is applied. However, according to the definition of $Q_{\text{inv}}(\mathsf{V}, \epsilon_0)$ in \eqref{eq:qinvvec}, $\mathbf{b}$ can have coordinates growing with $n$, which violates this assumption. In \cite[Th.~11]{chen2019SW}, Chen \textit{et al.} confirm that the asymptotic expansion in \eqref{eq:huangm} holds. They prove the converse part of the expansion \eqref{eq:huangm} by evaluating a converse bound that they derive in \cite[Lemma~9]{chen2019SW} for the composite hypothesis testing.

} show that the asymptotic form of the error region defined in \eqref{def:error_region} is given by
\begin{IEEEeqnarray}{rCl}
\IEEEeqnarraymulticol{3}{l}
{\mathcal{E}_{\epsilon_0} (P_{Y_0}, \{P_{Y_k}\}_{k = 1}^K )} \notag \\
\hspace{-0.3em} &=& \exp\left\{ -n \mathbf{D} + \sqrt{n} \mathcal{Q}_{\textnormal{inv}}(\mathsf{V}, \epsilon_0) -\frac 1 2 \log n \boldsymbol{1} + O(1) \boldsymbol{1}\right\}. \IEEEeqnarraynumspace \label{eq:huangm}
\end{IEEEeqnarray}

By the definition of the minimax error \eqref{eq:minmaxerror} and the characterization of the achievable error region asymptotics in \eqref{eq:huangm}, we have
\begin{align}
&\beta_{\epsilon_0}(P_{Y_0}, \{P_{Y_k}\}_{k = 1}^K) \notag \\
&= \min_{\substack{ \mathbf{z} \in \exp\{ -n \mathbf{D} + \sqrt{n} \mathcal{Q}_{\textnormal{inv}}(\mathsf{V}, \epsilon_0) -\frac 1 2 \log n \boldsymbol{1} + O(1) \boldsymbol{1}\}}} \max_{ 1 \leq k \leq K } z_k.
\end{align}
Applying \lemref{lem:generalminmax} with $f(\mathbf{z}) = \Prob{- n \mathbf{D} + \sqrt{n} \mathbf{Z} \leq \mathbf{z}}$ and $a = 1- \epsilon_0$, where $\mathbf{Z} \sim \mathcal{N}\left(\boldsymbol{0}, \mathsf{V}\right)$, we obtain
\begin{align}
&\beta_{\epsilon_0}(P_{Y_0}, \{P_{Y_k}\}_{k = 1}^K) \notag \\
&\quad = \min_{z \in \mathbb{R}: f(z \boldsymbol{1}) \geq 1 - \epsilon_0} \exp\left\{ z - \frac 1 2 \log n + O(1)  \right\}. \label{eq:zminac}
\end{align}
Since $f(z\boldsymbol{1})$ is nondecreasing and continuous in $z$, 
\begin{align}
f(z^{\star} \boldsymbol{1}) = 1- \epsilon_0 \label{eq:zstar1minus}
\end{align} 
holds, where $z^{\star}$ is the argument that achieves the minimum on the right-hand side of \eqref{eq:zminac}. Recall the definitions of $D_{\min}$ and $\mathcal{I}_{\min}$ from \eqref{eq:Dmin}--\eqref{eq:Imin}. By Chernoff's bound on $f(\mathbf{z})$, for any $z = nE + o(n)$ with $E > -D_{\min}$, we have $f(z \boldsymbol{1}) = 1 - o(1)$. Similarly, for $E < -D_{\min}$, we have $f(z \boldsymbol{1}) = o(1)$, giving 
\begin{align}
z^{\star} = -n D_{\min} + o(n). \label{eq:zstar}
\end{align}
We proceed to show that the minimum on the right-hand side of \eqref{eq:zminac} is achieved at
\begin{align}
z^{\star} = -n D_{\min} + \sqrt{n} b + O\left(1 \right) , \label{eq:zstarshow}
\end{align}
where $b$ is defined in \eqref{eq:bsolution}. Here
\begin{IEEEeqnarray}{rCl}
\IEEEeqnarraymulticol{3}{l}{\Prob{-n D_{\min} \boldsymbol{1} + \sqrt{n}\mathbf{Z}_{\mathcal{I}_{\min}} \leq z^\star \boldsymbol{1}}} \notag \\
&=& \Prob{- n \mathbf{D} + \sqrt{n} \mathbf{Z} \leq z^\star \boldsymbol{1}} \notag\\
&&+ \mathbb{P}\Big[ \{-n D_{\min} \boldsymbol{1} + \sqrt{n} \mathbf{Z}_{\mathcal{I}_{\min}} \leq z^\star \boldsymbol{1}\} \notag \\
&& \quad \bigcap \left\{ -{n \mathbf{D}}_{\mathcal{I}_{\min}^c} + \sqrt{n} \mathbf{Z}_{\mathcal{I}_{\min}^c} \nleq z^{\star} \boldsymbol{1} \right\} \Big]\\
&=& 1 - \epsilon_0 + O\left( \frac 1 n\right), \label{eq:chernoffgaussian}
\end{IEEEeqnarray}
where \eqref{eq:chernoffgaussian} follows from \eqref{eq:zstar1minus}, \eqref{eq:zstar}, and the union bound and Chebyshev's inequality on $\Prob{ -{n \mathbf{D}}_{\mathcal{I}_{\min}^c} + \mathbf{Z}_{\mathcal{I}_{\min}^c} \nleq z^{\star} \boldsymbol{1}}$. By the Taylor series expansion of $Q_{\textnormal{inv}}(\mathsf{V}, \cdot)$, we conclude that
\begin{align}
\Prob{\mathbf{Z}_{\mathcal{I}_{\min}} \leq \frac{1}{\sqrt{n}}(z^{\star} + n D_{\min}) \boldsymbol{1}+ O\left( \frac 1 n \right)} = 1 - \epsilon_0,
\end{align}
which implies \eqref{eq:zstarshow}. Combining \eqref{eq:zminac} and \eqref{eq:zstarshow} completes the proof.
\end{IEEEproof}

\section*{Acknowledgment}
We are grateful to the reviewers---Professor Jonathan Scarlett and two anonymous reviewers---for their thorough, careful, and insightful feedback, which is reflected in the paper. 

\bibliographystyle{IEEEtran}
\bibliography{mac} 

\begin{IEEEbiographynophoto}
{Recep Can Yavas}(S'19) is currently a Ph.D. candidate in electrical engineering at the California Institute of Technology (Caltech). He received the B.S. degree from Bilkent University in Ankara, Turkey, in 2016 and the M.S. degree from Caltech in 2017, both in electrical engineering. His research interests include information theory and probability theory. 
\end{IEEEbiographynophoto}

\begin{IEEEbiographynophoto}{Victoria Kostina}(S'12-M'14)
received the bachelor's degree from the Moscow institute of Physics and Technology in 2004, the master's degree from the University of Ottawa in 2006, and the Ph.D. degree from Princeton University in 2013. In 2014, she joined Caltech, where she is currently a Professor of electrical engineering. Her research spans information theory, coding, control, learning, and communications. She received the Natural Sciences and Engineering Research Council of Canada master's scholarship in 2009, the Princeton Electrical Engineering Best Dissertation Award in 2013, the Simons-Berkeley Research Fellowship in 2015, and the NSF CAREER Award in 2017.
\end{IEEEbiographynophoto}

\begin{IEEEbiographynophoto}{Michelle Effros}(S'93-M'95-SM'03-F'09) 
received the B.S. degree with distinction in 1989, the M.S. degree in 
1990, and the Ph.D. degree in 1994, all in electrical engineering from 
Stanford University. She joined the faculty at the California Institute of 
Technology in 1994, where she is currently the George Van Osdol Professor 
of Electrical Engineering. Her research interests include information 
theory, network coding, data compression, and communications.

Prof. Effros received Stanford's Frederick Emmons Terman Engineering 
Scholastic Award (for excellence in engineering) in 1989, the Hughes 
Masters Full-Study Fellowship in 1989, the National Science Foundation 
Graduate Fellowship in 1990, the AT\&T Ph.D. Scholarship in 1993, the NSF 
CAREER Award in 1995, the Charles Lee Powell Foundation Award in 1997, the 
Richard Feynman-Hughes Fellowship in 1997, and an Okawa Research Grant in 
2000. She was cited by Technology Review as one of the world's top young 
innovators in 2002. She and her co-authors received the Communications 
Society and Information Theory Society Joint Paper Award in 2009. She 
became a fellow of the IEEE in 2009. She is a member of Tau Beta Pi, Phi 
Beta Kappa, and Sigma Xi. She served as the Editor of the IEEE Information 
Theory Society Newsletter from 1995 to 1998 and as a Member of the Board 
of Governors of the IEEE Information Theory Society from 1998 to 2003 and 
from 2008 to 2017, serving in the role of President of the Information 
Theory Society in 2015. She was a member of the Advisory Committee and the 
Committee of Visitors for the Computer and Information Science and 
Engineering (CISE) Directorate at the National Science Foundation from 
2009 to 2012 and in 2014, respectively. She served on the IEEE Signal 
Processing Society Image and Multi-Dimensional Signal Processing (IMDSP) 
Technical Committee from 2001 to 2007 and on ISAT from 2006 to 2009. She 
served as Associate Editor for the 2006 joint special issue on Networking 
and Information Theory in the {\em IEEE Transactions on Information 
Theory} and the {\em IEEE Transactions on Networking/ACM Transactions on 
Networking} and as Associate Editor for Source Coding for the {\em IEEE 
Transactions on Information Theory} from 2004 to 2007. She has served on 
numerous technical program committees and review boards, including serving 
as general co-chair for the 2009 Network Coding Workshop and technical 
program committee co-chair for the 2012 IEEE International Symposium on 
Information Theory.
\end{IEEEbiographynophoto}

\end{document}